\documentclass{article}

\usepackage[utf8]{inputenc} %

\usepackage{amsthm, amsfonts, amssymb, amsmath,enumitem, natbib, bbm, mathabx, stmaryrd, tablefootnote, fancyref}
\usepackage{mathtools}
\usepackage[margin=1.2in]{geometry}
\mathtoolsset{showonlyrefs}
\newtheorem{thm}{Theorem}%
\newtheorem{lem}{Lemma}%
\newtheorem{prop}{Proposition}%
\newtheorem{rem}{Remark}%

\usepackage{graphicx} %
\usepackage{color, float}

\usepackage[ruled,vlined,linesnumbered]{algorithm2e}
\usepackage{array} %
\usepackage{verbatim} %
\usepackage{subfig} %
\usepackage[colorlinks=true, pdfstartview=FitV,
linkcolor=blue, citecolor=blue, urlcolor=blue]{hyperref}
\makeatletter
\def\namedlabel#1#2{\begingroup
	#2%
	\def\@currentlabel{#2}%
	\phantomsection\label{#1}\endgroup
}
\makeatother
\newcommand{\vertiii}[1]{{\left\vert\kern-0.25ex\left\vert\kern-0.25ex\left\vert #1
		\right\vert\kern-0.25ex\right\vert\kern-0.25ex\right\vert}}

\usepackage[titles,subfigure]{tocloft} %

\newcommand{\rev}[1]{#1}%
\SetKwInput{KwInput}{Input}
\SetKwInput{KwOutput}{Output}

\begin{document}
\title{Nested conformal prediction\\ and quantile out-of-bag ensemble methods}

\author{Chirag Gupta, Arun K. Kuchibhotla, Aaditya K. Ramdas\\
Carnegie Mellon University, \\
\texttt{chiragg@cmu.edu, arunku@cmu.edu, aramdas@cmu.edu}
}

\maketitle
\begin{abstract}
Conformal prediction is a popular tool for providing valid prediction sets for classification and regression problems, without relying on any distributional assumptions on the data. While the traditional description of conformal prediction starts with a nonconformity score, we provide an alternate (but equivalent) view that starts with a sequence of nested sets and calibrates them to find a valid prediction set. 
The nested framework subsumes all nonconformity scores, including recent proposals based on quantile regression and density estimation. While these ideas were originally derived based on sample splitting, our framework seamlessly extends them to other aggregation schemes like cross-conformal, jackknife+ and out-of-bag methods.
We use the framework to derive a new algorithm (QOOB, pronounced cube) that combines four ideas: quantile regression, cross-conformalization, ensemble methods and out-of-bag predictions. \rev{We develop a computationally efficient implementation of cross-conformal, that is also used by QOOB.} In a detailed numerical investigation, QOOB performs either the best or close to the best on all simulated and real datasets. Code for QOOB is available at \url{ https://github.com/aigen/QOOB}.
\end{abstract}

\newpage
\tableofcontents
\newpage

\section{Introduction}
\rev{Traditional machine learning algorithms provide point predictions, such as mean estimates for regression and class labels for classification, without conveying uncertainty or confidence. However, sensitive applications like medicine and finance often require valid uncertainty estimates. In this paper, we discuss  quantification of predictive uncertainty through predictive inference, wherein we provide prediction sets rather than point predictions. } 

Formally, the problem of distribution-free predictive inference is described as follows: given dataset $D_n \equiv \{(X_i, Y_i)\}_{i=1}^n$ drawn i.i.d. from $P_{XY} = P_X \times P_{Y|X}$ on $\smash{\mathcal{X} \times \mathcal{Y}}$\footnote{The spaces $\mathcal{X}$ and $\mathcal{Y}$ are without restriction. For example, take $\mathcal{X} \equiv \mathbb{R}^d$ and let $\mathcal{Y}$ be a subset of $\mathbb{R}$,  or a discrete space such as in multiclass classification. Though it is not formally required, it may be helpful think of $\mathcal{Y}$ as a totally ordered set or a metric space.
}, and $\smash{X_{n+1} \sim P_X}$, we must construct a prediction set $C(D_n,\alpha,X_{n+1}) \equiv C(X_{n+1})$ for $Y_{n+1}$ that satisfies:
\begin{equation}\label{eq:marginal-validity}
\text{ for any distribution } P_{XY}, ~ \mathbb{P}(Y_{n+1} \in C(X_{n+1})) \geq 1-\alpha.
\end{equation}
Here the probability is taken over all $n+1$ points and $\alpha \in (0,1)$ is a predefined confidence level. %
\rev{As long as \eqref{eq:marginal-validity} is true, the `size' of $C(X_{n+1})$ conveys how certain we are about the prediction at $X_{n+1}$.} Methods with property~\eqref{eq:marginal-validity} are called \emph{marginally} valid to differentiate them from \emph{conditional} validity:
\[
\forall P_{XY},~ \mathbb{P}(Y_{n+1} \in C(X_{n+1}) \mid X_{n+1}=x) \geq 1-\alpha, \text{ for $P_X$-almost all $x$.} 
\]
Conditional validity is known to be impossible without assumptions on $P_{XY}$; see~\citet[Section 2.6.1]{balasubramanian2014conformal} and~\citet{barber2019limits}. \rev{If $Y_{n+1} \in C(X_{n+1})$, we say that $C(X_{n+1})$ covers $Y_{n+1}$, and we often refer to marginal (conditional) validity as marginal (conditional) coverage. In this paper, we develop methods that are (provably) marginally valid, but have reasonable conditional coverage in practice, using conformal prediction.}

Conformal prediction is a universal framework for constructing marginally valid prediction sets. It acts like a wrapper around any prediction algorithm; in other words, any black-box prediction algorithm can be ``conformalized'' to produce valid prediction sets instead of point predictions. 
We refer the reader to the works of~\citet{vovk2005algorithmic} and~\citet{balasubramanian2014conformal} for details on the original framework.
In this paper, we provide an alternate and equivalent viewpoint for accomplishing the same goals, called \emph{nested conformal prediction}. %

Conformal prediction starts from a nonconformity score. As we will see with explicit examples later, these nonconformity scores are often rooted in some underlying geometric intuition about how a good prediction set may be discovered from the data. Nested conformal acknowledges this geometric intuition and makes it explicit: instead of starting from a score, it instead starts from a sequence of all possible prediction sets $\{\mathcal{F}_t(x)\}_{t\in\mathcal{T}}$ for some ordered set $\mathcal{T}$. \rev{Just as we suppressed the dependence of set $C(\cdot)$ on the labeled data $D_n$ in property \eqref{eq:marginal-validity}, here too $\mathcal{F}_t(\cdot)$ will actually depend on $D_n$ but we suppress this for notational simplicity. }

These prediction sets are `nested', that is, for every $t_1 \leq t_2 \in \mathcal{T}$, we require that $\mathcal{F}_{t_1}(x) \subseteq \mathcal{F}_{t_2}(x)$; also $\mathcal{F}_{\inf \mathcal{T}}=\emptyset$ and $\mathcal{F}_{\sup \mathcal{T}}=\mathcal{Y}$. Thus large values of $t$ correspond to larger prediction sets. Given a tolerance $\alpha \in [0, 1]$, we wish to identify the smallest $t \in \mathcal{T}$ such that 
\[\mathbb{P}(Y \in \mathcal{F}_{t}(X)) \geq 1 - \alpha. 
\]
In a nutshell, nested conformal learns a data-dependent mapping $\alpha \to t(\alpha)$ using the conformal principle. Note that the mapping must be decreasing; lower tolerance values $\alpha$ naturally lead to larger prediction sets.

\rev{We now briefly describe the steps involved in split/inductive conformal prediction \cite{papadopoulos2002inductive}, \cite{lei2018distribution},~\cite[Chapter 2.3]{balasubramanian2014conformal}, and use it to illustrate the nested principle. First, split $D_n$ into a training set $D_1 \equiv \{(X_i, Y_i)\}_{1\le i\le m}$ and a calibration set $D_2 \equiv \{(X_i, Y_i)\}_{m < i\le n}$. Using $D_1$, construct an estimate $\widehat{\mu}(\cdot)$ of the conditional mean of $Y$ given $X$. Then construct the nonconformity score as the residuals of $\widehat{\mu}$ on $D_2$: 
$r_i := |Y_i - \widehat{\mu}(X_i)|, \text{for } i \in D_2.$
Finally, define
\[
C(X_{n+1}) = \Big\{y\in\mathbb{R}:\,|y - \widehat{\mu}(X_{n+1})| < Q_{1-\alpha}(\{r_i\}_{i\in D_2})\Big\},
\] 
where $Q_{1-\alpha}(A)$ for a finite set $A$ represents the $(1-\alpha)$-th quantile of elements in $A$. Due to the exchangeability of order statistics, $C(\cdot)$ can be shown to be marginally valid (see Proposition~\ref{thm:Conformal-main-result}).}

We now give an alternate derivation of the above set using nested conformal:
\begin{enumerate}
	\item After learning $\widehat \mu$ using $D_1$ (as done before), construct a sequence of nested prediction sets corresponding to symmetric intervals around $\widehat{\mu}(\cdot)$: 
  \[
  \{\mathcal{F}_t(\cdot)\}_{t\ge0} := \{[\widehat{\mu}(\cdot) - t, \widehat{\mu}(\cdot) + t]:t\ge0\}.
  \] 
  Note that $\mathcal{F}_t(\cdot)$ is a random set since it is based on $\widehat{\mu}(\cdot)$ which is random through ${D}_1$.
  It is clear that regardless of $\widehat{\mu}$, for any distribution of  $(X, Y)$, and any $\alpha \in [0,1]$, there exists a (minimal) $t = t(\alpha)$ such that
$\mathbb{P}(Y\in\mathcal{F}_t(X)) \ge 1-\alpha.$
	Hence we can rewrite our nested family  as
	$$\Big\{[\widehat{\mu}(\cdot) - t, \widehat{\mu}(\cdot) + t]:t\ge0\Big\} = \Big\{[\widehat{\mu}(\cdot) - t(\alpha), \widehat{\mu}(\cdot) + t(\alpha)]:\alpha\in[0,1]\Big\}.$$
    \item The only issue now is that we do not know the map $\alpha\mapsto t(\alpha)$, that is, given $\alpha$ we do not know which of these prediction intervals to use. Hence we use the calibration data to ``estimate'' the map $\alpha\to t(\alpha)$. This is done by finding the smallest $t$ such that $\mathcal{F}_t(X_i)$ contains $Y_i$ for at least $1-\alpha$ fraction of the calibration points $(X_i, Y_i)$ (we provide formal details later). Because the sequence $\{\mathcal{F}_t(\cdot)\}_{t \geq 0}$ is increasing in $t$, finding the smallest $t$ leads to the smallest prediction set within the nested family.  
\end{enumerate}

Embedding nonconformity scores into our nested framework allows for easy comparison between the geometric intuition of the scores; see Table~\ref{tab:summary_literature}. Further, this interpretation enables us to extend these nonconformity scores beyond the split/inductive conformal setting that they were originally derived in. Specifically, we seamlessly derive cross-conformal, jackknife+ and OOB versions of these methods, including our new method called QOOB (pronounced cube). 

\rev{A final reason that the assumption of nestedness is natural is the fact that the optimal prediction sets are nested: Suppose $Z_1, \ldots, Z_{n}$ are exchangeable random variables with a common distribution that has density $p(\cdot)$ with respect to some underlying measure. The ``oracle'' prediction set~\citep{lei2013distribution} for a future observation $Z_{n+1}$ is given by the \emph{level set} of the density with valid coverage, that is, $\{z:p(z) \ge t(\alpha)\}$ with $t(\alpha)$ defined by largest $t$ such that $\mathbb{P}(p(Z_{n+1}) \ge t) \ge 1-\alpha$. Because $\{z:p(z) \ge t\}$ is decreasing with $t$, $\{z:p(z) \ge t(\alpha)\}$ is decreasing with $\alpha\in[0,1]$, forming a nested sequence of sets. See~Appendix~\ref{appsec:optimal-prediction-set} for more details. }

\subsection{Organization and contributions}
For simplicity, our discussion focuses on the regression setting: $\mathcal{Y} = \mathbb{R}$. However, all ideas are easily extended to other prediction settings \rev{including classification}. The paper is organized as follows:
\begin{enumerate}
\item In Section~\ref{sec:split-conformal}, we formalize the earlier discussion and present split/inductive conformal~\citep{papadopoulos2002inductive,lei2018distribution} in the language of nested conformal prediction, and translate various conformity scores developed in the literature for split conformal into nested prediction sets. 
\item  In Section~\ref{sec:Jackknife-plus}, we rephrase the jackknife+~\citep{barber2019predictive} and cross-conformal prediction~\citep{vovk2015cross} in terms of the nested framework. 
This allows the jackknife+ to use many recent score functions, such as those based on quantiles, which were originally developed and deployed in the split framework. 
In Section~\ref{subsec:compuation-cross} we provide an efficient implementation of cross-conformal that matches the jackknife+ prediction time for a large class of nested sets that includes all standard nested sets. 

\item In Section~\ref{sec:OOB-conformal}, we extend the Out-of-Bag conformal~\citep{johansson2014regression} and jackknife+ after bootstrap~\citep{kim2020predictive} methods to our nested framework. These are based on ensemble methods such as random forests, and are relatively computationally efficient because only a single ensemble needs to be built. 
\item In Section~\ref{sec:QOOB}, we consolidate the ideas developed in this paper to construct a novel conformal method called QOOB (Quantile Out-of-Bag, pronounced cube), that is both computationally and statistically efficient. 
QOOB combines four ideas: quantile regression, cross-conformalization, ensemble methods and out-of-bag predictions. Section~\ref{sec:numerical-experiments} demonstrates  QOOB's strong empirical performance. 
\end{enumerate}

\rev{In Appendix~\ref{appsec:Equivalence}, we show that nested conformal is equivalent to the standard conformal prediction based on nonconformity scores. We also formulate full transductive conformal prediction in the nested framework. }In Appendix~\ref{app:k-fold} we derive K-fold variants of jackknife+/cross-conformal in the nested framework and in Appendix~\ref{sec:sampling-nested}, we develop the other aggregated conformal methods of subsampling and bootstrap in the nested framework. In Appendix~\ref{app:empty-case}, we discuss cross-conformal computation and the jackknife+ in the case when our nested sequence could contain empty sets. This is a subtle but important issue to address when extending these methods to quantile-based nested sets of \citet{romano2019conformalized}, and thus relevant to QOOB as well. %
Appendix~\ref{appsec:proofs} contains all proofs.

\section{Split conformal based on nested prediction sets}\label{sec:split-conformal}
In the introduction, we showed that in a simple regression setup with the nonconformity scores as held-out residuals, split conformal intervals can be naturally expressed in terms of nested sets. Below, we introduce the general nested framework and recover the usual split conformal method with general scores using this framework. We show how existing nonconformity scores in literature exhibit natural re-interpretations in the nested framework. The following description of split conformal follows descriptions by \citet{papadopoulos2002inductive} and~\citet{lei2018distribution} but rewrites it in terms of nested sets. 

Suppose $(X_i, Y_i)\in\mathcal{X}\times\mathcal{Y}, i\in[n]$ denotes the training dataset. Let $[n] = \mathcal{I}_1\cup\mathcal{I}_2$ be a partition of $[n]$. For $\mathcal{T}\subseteq\mathbb{R}$ and each $x\in\mathcal{X}$, let $\{\mathcal{F}_t(x)\}_{t\in\mathcal{T}}$ (with $\mathcal{F}_t(x)\subseteq\mathcal{Y}$) denote a nested sequence of sets constructed based on the first split of training data $\{(X_i, Y_i):i\in\mathcal{I}_1\}$, that is, $\mathcal{F}_t(x) \subseteq \mathcal{F}_{t'}(x)$ for $t \le t'$. \rev{The sets $\mathcal{F}_t(x)$ are not fixed but random through $\mathcal{I}_1$. We suppress this dependence for notational simplicity. } Consider the score 
\begin{equation}\label{eq:Score-definition}
r(x, y) ~:=~ \inf\{t\in\mathcal{T}:\,y\in\mathcal{F}_t(x)\},
\end{equation}
where $r$ is a mnemonic for ``radius'' and $r(x,y)$ can be informally thought of as the smallest ``radius'' of the set that captures $y$ (and perhaps thinking of a multivariate response, that is $\mathcal{Y}\subseteq \mathbb{R}^d$, and $\{\mathcal{F}_t(x)\}$ as representing appropriate balls/ellipsoids might help with that intuition).
Define the scores for the second split of the training data $\{r_i = r(X_i, Y_i)\}_{i\in\mathcal{I}_2}$ and set
\[
Q_{1-\alpha}(r, \mathcal{I}_2) ~:=~ \lceil (1-\alpha)(1 + 1/|\mathcal{I}_2|)\rceil\mbox{-th quantile of }\{r_i\}_{i\in\mathcal{I}_2}.
\]
(that is, $Q_{1-\alpha}(r, \mathcal{I}_2)$ is the $\lceil(1-\alpha)(1 + 1/|\mathcal{I}_2|)\rceil$-th largest element of the set $\{r_i\}_{i\in\mathcal{I}_2}$). 
The final prediction set is given by 
\begin{equation}
C(x) ~:=~ \mathcal{F}_{Q_{1-\alpha}(r,\mathcal{I}_2)}(x) = \{y \in \mathcal{Y}:\,r(x, y) \le Q_{1-\alpha}(r, \mathcal{I}_2)\}.
\label{eq:Nested-prediction-set}
\end{equation}
The following well known sample coverage guarantee holds true \citep{papadopoulos2002inductive,lei2018distribution}.
\begin{prop}\label{thm:Conformal-main-result}
If $\{(X_i, Y_i)\}_{i\in[n]\cup\{n+1\}}$ are exchangeable, then the prediction set $C(\cdot)$ in \eqref{eq:Nested-prediction-set} satisfies 
\[
\mathbb{P}\left(Y_{n+1}\in C(X_{n+1})\mid \{(X_i, Y_i):i\in\mathcal{I}_1\}\right) ~\ge~ 1-\alpha.
\]
If the scores $\{r_i, i\in\mathcal{I}_2\}$ are almost surely distinct, then $C(\cdot)$ also satisfies
\begin{equation}\label{eq:Upper-bound}
\mathbb{P}\left(Y_{n+1}\in C(X_{n+1})\mid \{(X_i, Y_i):i\in\mathcal{I}_1\}\right) ~\le~ 1-\alpha + \frac{1}{|\mathcal{I}_2|+1}.
\end{equation}
\end{prop}
See Appendix~\ref{appsec:split-conformal} for a proof. 
Equation~\eqref{eq:Score-definition} is the key step that converts a sequence of nested sets $\{\mathcal{F}_t(x)\}_{t\in\mathcal{T}}$ into a nonconformity score $r$. Through two examples, we demonstrate how natural sequences of nested sets in fact give rise to standard nonconformity scores considered in literature, via equation~\eqref{eq:Score-definition}. 

\begin{table}[t]
\caption{Examples from the literature covered by nested conformal framework. The methods listed are split conformal, locally weighted conformal, CQR, CQR-m, CQR-r, distributional conformal and conditional level-set conformal. Functions $\widehat{q}_{a}$ represents a conditional quantile estimate at level $a$, and $\widehat{f}$ represents a conditional density estimate.}
\label{tab:summary_literature}
\resizebox{\textwidth}{!}{%
\begin{tabular}{llll}\hline
Reference & $\mathcal{F}_t(x)$ & $\mathcal{T}$ & Estimates \\\hline\hline
\citet{lei2018distribution} & $[\widehat{\mu}(x) - t, \widehat{\mu}(x) + t]$ & $[0, \infty)$ & $\widehat{\mu}$ \\
\citet{lei2018distribution} & $[\widehat{\mu}(x) - t\widehat{\sigma}(x),\widehat{\mu}(x) + t\widehat{\sigma}(x)]$ & $[0, \infty)$ & $\widehat{\mu}, \widehat{\sigma}$ \\
\citet{romano2019conformalized} & $[\widehat{q}_{\alpha/2}(x) - t, \widehat{q}_{1-\alpha/2}(x) + t]$ & $(-\infty, \infty)$ & $\widehat{q}_{\alpha/2}, \widehat{q}_{1-\alpha/2}$ \\
\citet{kivaranovic2019adaptive} & $(1+t)[\widehat{q}_{\alpha/2}(x), \widehat{q}_{1-\alpha/2}(x)] - t\widehat{q}_{1/2}(x)$ & $(-\infty, \infty)$ & $\widehat{q}_{\alpha/2}, \widehat{q}_{1-\alpha/2}, \widehat{q}_{1/2}$ \\
\citet{sesia2019comparison} & $[\widehat{q}_{\alpha/2}(x), \widehat{q}_{1-\alpha/2}(x)] \pm t(\widehat{q}_{1-\alpha/2}(x) - \widehat{q}_{\alpha/2}(x))$ & $(-1/2,\infty)$ & $\widehat{q}_{\alpha/2}, \widehat{q}_{1-\alpha/2}$ \\
\citet{chernozhukov2019distributional} & $[\widehat{q}_{t}(x), \widehat{q}_{1-t}(x)]$ & $(0,1/2)$ & $\{\widehat{q}_{\alpha}\}_{\alpha\in[0,1]}$\\
\citet{izbicki2019distribution} & $\{y:\widehat{f}(y|x) \geq \widecheck{t}_{\delta}(x)\}$\tablefootnote{$\ \widecheck{t}_{\delta}(x)$ is an estimator of $t_{\delta}(x)$, where $t_{\delta}(x)$ is defined the largest $t$ such that $\mathbb{P}(f(Y|X) \ge t_{\delta}(X)\big|X = x) \ge 1 - \delta$; see~\citep[Definition 3.3]{izbicki2019distribution} for details.} & $[0,1]$ & $\widehat{f}$\\
\hline
\end{tabular}%
}
\end{table}

\begin{enumerate}
  \item \textbf{Split/Inductive Conformal \citep{papadopoulos2002inductive,lei2018distribution}.} Let $\widehat{\mu}(\cdot)$ be an estimator of the regression function $\mathbb{E}[Y|X]$ based on $(X_i, Y_i), i\in\mathcal{I}_1$, and consider nested sets corresponding to symmetric intervals around the mean estimate: 
  \[
  \mathcal{F}_t(x) ~:=~ [\widehat{\mu}(x) - t, \widehat{\mu}(x) + t], \ t \in \mathcal{T} = \mathbb{R}^+.
  \]
  Observe now that
  \begin{align*}
  \inf\{t\ge0:y\in\mathcal{F}_t(x)\} &= \inf\{t\ge0:\widehat{\mu}(x) - t \le y\le \widehat{\mu}(x) + t\}\\
  &= \inf\{t\ge0:\, -t \le y - \widehat{\mu}(x) \le t\} ~=~ |y - \widehat{\mu}(x)|,
  \end{align*}
 which is exactly the nonconformity score of split conformal. 
  \item \textbf{Conformalized Quantiles \citep{romano2019conformalized}.} 
For any $\beta \in (0,1)$, let the function $q_\beta(\cdot)$ be the conditional quantile function. Specifically, for each $x$, define  $q_{\beta}(x):= \sup\{a: \mathbb{P}(Y \le a \mid X = x) \le \beta$\}.
  Let $\widehat{q}_{\alpha/2}(\cdot), \widehat{q}_{1-\alpha/2}(\cdot)$ be any conditional quantile estimators based on $(X_i, Y_i),i\in\mathcal{I}_1$.
  If the quantile estimates are good, we hope that $\mathbb{P}(Y \in [\widehat{q}_{\alpha/2}(X), \widehat{q}_{1-\alpha/2}(X)]) \approx 1-\alpha$, but this cannot be guaranteed in a distribution-free or assumption lean manner. However, it may be possible to  achieve this with a symmetric expansion or shrinkage of the interval $[\widehat{q}_{\alpha/2}(X), \widehat{q}_{1-\alpha/2}(X)]$ (assuming $\widehat{q}_{\alpha/2}(X) \leq \widehat{q}_{1 - \alpha/2}(X)$). 
  Following the intuition, consider
  \begin{equation}\label{eq:CQR}
  \mathcal{F}_t(x) ~:=~ [\widehat{q}_{\alpha/2}(x) - t, \widehat{q}_{1-\alpha/2}(x) + t], \ t\in\mathbb{R}.
  \end{equation}

  Note that the sets in~\eqref{eq:CQR} are increasing in $t$ if $\widehat{q}_{\alpha/2}(x) \le \widehat{q}_{1-\alpha/2}(x)$, and
  \begin{align*}
  \inf\{t\in\mathbb{R}:\,y\in\mathcal{F}_t(x)\} &= \inf\{t\in\mathbb{R}:\,\widehat{q}_{\alpha/2}(x) - t \le y \le \widehat{q}_{1-\alpha/2}(x) + t\}\\
  &= \max\{\widehat{q}_{\alpha/2}(x) - y, y - \widehat{q}_{1-\alpha/2}(x)\}.
  \end{align*}
  Hence $r(X_i, Y_i) = \max\{\widehat{q}_{\alpha/2}(X_i) - Y_i,\,Y_i - \widehat{q}_{1-\alpha/2}(X_i)\}$ for $i\in\mathcal{I}_2$. This recovers exactly the nonconformity score proposed by \citet{romano2019conformalized}.  
  
\end{enumerate}

We believe that it is more intuitive to start with the shape of the predictive set, like we did above, than a nonconformity score. In this sense, nested conformal is a formalized technique to go from statistical/geometric intuition about the shape of the prediction set to a nonconformity score. See Table~\ref{tab:summary_literature} for more translations between scores and nested sets.

Split conformal prediction methods are often thought of as being statistically inefficient because they only make use of one split of the data for training the base algorithm, while the rest is held-out for calibration. Recently many extensions have been proposed~\citep{carlsson2014aggregated,vovk2015cross,johansson2014regression,bostrom2017accelerating,barber2019predictive,kim2020predictive} that make use all of the data for training. 
All of these methods can be rephrased easily in terms of nested sets; we do so in Sections~\ref{sec:Jackknife-plus} and \ref{sec:OOB-conformal}. This understanding also allows us to develop our novel method QOOB in Section~\ref{sec:QOOB}.

\section{Cross-conformal and Jackknife+ using nested sets}\label{sec:Jackknife-plus}

In the previous section, we used a part of training data to construct the nested sets and the remaining part to calibrate them for finite sample validity. This procedure, although computationally efficient, can be statistically inefficient due to the reduction of the sample size used for calibrating. Instead of splitting into two parts, it is statistically more efficient to split the data into multiple parts. In this section, we describe such versions of nested conformal prediction sets and prove their validity.
These versions in the score-based conformal framework are called cross-conformal prediction and the jackknife+, and were developed by~\citet{vovk2015cross} and~\citet{barber2019predictive}, but the latter only for a specific score function.

\subsection{Rephrasing leave-one-out cross-conformal using nested sets}
We now derive leave-one-out cross-conformal in the language of nested prediction sets. %
Suppose $\{\mathcal{F}_{t}^{-i}(x)\}_{t\in\mathcal{T}}$ for each $x\in\mathcal{X},\,i\in[n]$ denotes a collection of nested sets constructed based only on $\{(X_j, Y_j)\}_{\,j\in[n]\setminus\{i\}}$. We assume that  $\{\mathcal{F}_{t}^{-i}(x)\}_{t\in\mathcal{T}}$ is constructed invariantly to permutations of the input points $\{(X_j, Y_j)\}_{\,j\in[n]\setminus\{i\}}$ \rev{; note that this is also required in the original description of cross-conformal \citep{vovk2015cross}}. %
The $i$-th nonconformity score $r_i$ induced by these nested sets is defined as $r_i(x, y) ~=~ \inf\{t\in\mathcal{T}:\,y\in\mathcal{F}_t^{-i}(x)\}.$ The leave-one-out cross-conformal prediction set is given by
\begin{equation}\label{eq:LOO}
{C}^{\texttt{LOO}}(x) ~:=~ \left\{y\in\mathbb{R}:\,\sum_{i=1}^n \mathbbm{1}\{r_i(X_i, Y_i) < r_i(x, y)\} < (1-\alpha)(n+1)\right\}.
\end{equation}
For instance, given a conditional mean estimator $\widehat{\mu}^{-i}(\cdot)$ trained on $\{(X_j, Y_j)\}_{\,j\in[n]\setminus\{i\}}$, we can consider the nested sets $\mathcal{F}_t^{-i}(x) = [\widehat{\mu}^{-i}(x) - t, \widehat{\mu}^{-i}(x) + t]$ to realize the absolute deviation residual function $r_i(x, y) = |y - \widehat{\mu}^{-i}(x)|$. %
We now state the coverage guarantee that ${C}^{\texttt{LOO}}(\cdot)$ satisfies. 

\begin{thm}\label{thm:validity-nested-conformal-jackknife-plus}
If $\{(X_i, Y_i)\}_{i\in[n+1]}$ are exchangeable and the sets $\mathcal{F}_{t}^{-i}(x)$ constructed based on $\{(X_j, Y_j)\}_{ j\in[n]\setminus\{i\}}$ are invariant to their ordering, then
\[
\mathbb{P}(Y_{n+1}\in{C}^{\texttt{LOO}}(X_{n+1})) ~\ge~ 1 - 2\alpha.
\]
\end{thm}
See Appendix~\ref{appsec:Jackknife-plus} for a proof
 of Theorem~\ref{thm:validity-nested-conformal-jackknife-plus}, which follows the proof of Theorem 1 in~\citet{barber2019predictive} except with the new residual defined based on nested sets. In particular, Theorem~\ref{thm:validity-nested-conformal-jackknife-plus} applies when the nested sets are constructed using conditional quantile estimators as in the conformalized quantile example discussed in Section~\ref{sec:split-conformal}. The discussion in this section can be generalization to cross-conformal and the CV+ methods of~\citet{vovk2015cross} and~\citet{barber2019predictive}, which construct K-fold splits of the data and require training an algorithm only $n/K$ times (instead of $n$ times in the leave-one-out case). These are discussed in the nested framework in Appendix~\ref{app:k-fold}.
 
In a regression setting, one may be interested in constructing prediction sets that are intervals (since they are easily interpretable), whereas $C^{\texttt{LOO}}(x)$ need not be an interval in general. Also, it is not immediately evident how one would algorithmically compute the prediction set defined in~\eqref{eq:LOO} without trying out all possible values $y \in \mathcal{Y}$. We discuss these concerns in Sections~\ref{subsec:jackknife+} and \ref{subsec:compuation-cross}. 
\subsection{Jackknife+ and other prediction intervals that contain $C^{\texttt{LOO}}(x)$}\label{sec:intervals-LOO}
\label{subsec:jackknife+}
In a regression setting, prediction intervals may be more interpretable or `actionable' than prediction sets that are not intervals. To this end, intervals that contain $C^{\texttt{LOO}}(x)$ are good candidates for prediction intervals since they inherit the coverage validity of Theorem~\ref{thm:validity-nested-conformal-jackknife-plus}. For the residual function $r_i(x, y) = |y - \widehat{\mu}^{-i}(x)|$, \citet{barber2019predictive} provided an interval that always contains $C^{\texttt{LOO}}(x)$, called the jackknife+ prediction interval. In this section, we discuss when the jackknife+ interval can be defined for general nested sets. Whenever jackknife+ can be defined, we argue that another interval can be defined that contains $C^\texttt{LOO}(x)$ and is guaranteed to be no longer in width than the jackknife+ interval.%

\rev{For general nonconformity scores, an analog of the jackknife+ interval may not exist.} However, in the special case when the nested sets $\mathcal{F}_t^{-i}(x)$ are themselves either intervals or empty sets, an analog of the jackknife+ interval can be derived. Note that all the examples listed in Table~\ref{tab:summary_literature} (except for the last one) result in $\mathcal{F}_t(x)$ being either a nonempty interval or the empty set. For clarity of exposition, we discuss the empty case separately in Appendix~\ref{app:empty-case}. Below, suppose $\mathcal{F}_{r_i(X_i, Y_i)}^{-i}(x) $ is a nonempty interval and define the notation
\[
[\ell_i(x), u_i(x)]~:=~\mathcal{F}_{r_i(X_i, Y_i)}^{-i}(x)  .
\] 
With this notation, the cross-conformal prediction set can be re-written as 
\begin{align}
C^{\texttt{LOO}}(x) &= \left\{y:\,\sum_{i=1}^n \mathbbm{1}\{y\notin [\ell_i(x), u_i(x)]\} < (1-\alpha)(n+1)\right\} \nonumber \\
&= \left\{y:\,  \alpha(n+1) - 1 < \sum_{i= 1}^n \mathbbm{1}\{y \in [\ell_i(x), u_i(x)]\} \right\}. \label{eq:loo-interval-definition}
\end{align}
Suppose $y < q_{n,\alpha}^{-}(\{\ell_i(x)\})$, where 
$q_{n,\alpha}^{-}(\{\ell_i(x)\})$ denotes the 
$\lfloor\alpha(n+1)\rfloor$-th smallest value of  $\{\ell_i(x)\}_{i=1}^n$. Clearly, 
\[
\sum_{i= 1}^n \mathbbm{1}\{y \in [\ell_i(x), u_i(x)]\} \leq \sum_{i= 1}^n \mathbbm{1}\{y \geq l_i(x)\} \leq \lfloor\alpha(n+1)\rfloor - 1,
\] 
and hence $y \notin C^\texttt{LOO}(x)$. Similarly it can be shown that if $y > q_{n,\alpha}^{+}(\{u_i(x)\})$ (where $q_{n,\alpha}^{+}(\{u_i(x)\})$ denotes the $\lceil(1-\alpha)(n+1)\rceil$-th smallest value of  $\{u_i(x)\}_{i=1}^n$), $y \notin C^\texttt{LOO}(x)$. 
Hence, defining
the jackknife+ prediction interval as
\begin{equation}
C^{\texttt{JP}}(x) ~:=~ [q_{n,\alpha}^{-}(\{\ell_i(x)\}),\, q_{n,\alpha}^{+}(\{u_i(x)\})],\label{eq:JP-definition}
\end{equation}
we conclude
\begin{equation}\label{eq:loo-jp}
C^{\texttt{LOO}}(x)~\subseteq~ C^{\texttt{JP}}(x) \quad\mbox{for all}\quad x\in\mathcal{X}.
\end{equation}
However, there exists an even shorter interval containing~$C^{\texttt{LOO}}(x)$: its convex hull; this does not require the nested sets to be intervals. The convex hull of $C^{\texttt{LOO}}(x)$ is defined as the smallest interval containing itself. Hence,
\begin{equation}\label{eq:intervals-containing-loo}
C^{\texttt{LOO}}(x) ~\subseteq~ \mathrm{Conv}(C^{\texttt{LOO}}(x)) ~\subseteq~ C^{\texttt{JP}}(x).
\end{equation}   
Because of~\eqref{eq:intervals-containing-loo}, the coverage guarantee from Theorem~\ref{thm:validity-nested-conformal-jackknife-plus} continues to hold for $\mathrm{Conv}(C^{\texttt{LOO}}(x))$ and $C^{\texttt{JP}}(x)$. Interestingly, $C^{\texttt{LOO}}(x)$ can be empty but $C^{\texttt{JP}}(x)$ is non-empty if each $\mathcal{F}_{r_i(X_i, Y_i)}^{-i}(x) $ is non-empty (in particular it contains the medians of $\{\ell_i(x)\}$ and $\{u_i(x)\}$). Further, $\mathrm{Conv}(C^{\texttt{LOO}}(x))$ can be a strictly smaller interval than $C^{\texttt{JP}}(x)$; see Section~\ref{subsec:cc-jp} for details.

\subsection{Efficient computation of the cross-conformal prediction set}
\label{subsec:compuation-cross}
Equation~\eqref{eq:LOO} defines $C^{\texttt{LOO}}(x)$ implicitly, and does not address the question of how to compute the mathematically defined prediction set efficiently. If the nested sets $\mathcal{F}_t^{-i}(x)$ are themselves guaranteed to either be intervals or empty sets, jackknife+ seems like a computationally feasible alternative since it just relies on the quantiles $q_{n,\alpha}^{-}(\{\ell_i(x)\}),\, q_{n,\alpha}^{+}(\{u_i(x)\})$ which can be computed efficiently. However, it turns out that $C^{\texttt{LOO}}(x), \mathrm{Conv}(C^{\texttt{LOO}}(x)),$ and $C^{\texttt{JP}}(x)$ can all be computed in near linear in $n$ time. In this section, we provide an algorithm for near linear time computation of the aforementioned prediction sets. We will assume for simplicity that $\mathcal{F}_t^{-i}(x)$ is always an interval; the empty case is discussed separately in Appendix~\ref{app:empty-case}.

First, notice that the inclusion in~\eqref{eq:LOO} need not be ascertained for every $y \in \mathcal{Y}$ but only for a finite set of values in $\mathcal{Y}$. These values are exactly the ones corresponding to the end-points of the intervals produced by each training point $\mathcal{F}_{r_i(X_i, Y_i)}^{-i}(x) = [\ell_i(x), u_i(x)]$. This is because none of the indicators $\mathbbm{1}\{r_i(X_i, Y_i) < r_i(x, y)\}$ change value between two consecutive interval end-points. Since $\ell_i(x)$ and $u_i(x)$ can be repeated, we define the bag of all these values (see footnote\footnote{A bag denoted by $\Lbag \cdot \Rbag$ is an unordered set with potentially repeated elements. Bag unions respect the number of occurrences of the elements, eg $\Lbag 1, 1, 2 \Rbag \cup \Lbag 1, 3 \Rbag = \Lbag 1, 1, 1, 2, 3 \Rbag$. } for $\Lbag \cdot \Rbag$ notation):
\begin{equation}
\label{eq:y-k-main}
\mathcal{Y}^x := \bigcup_{i=1}^n \Lbag \ell_i(x), u_i(x)\Rbag  .
\end{equation}
Thus we only need to verify the condition
\begin{equation}
\label{eq:y-condition}
\sum_{i= 1}^n \mathbbm{1}\{y \in [\ell_i(x), u_i(x)]\}  > \alpha(n+1) - 1
\end{equation}
for the $2n$ values of $y \in \mathcal{Y}^x$ and construct the prediction sets suitably. %
Done naively, verifying \eqref{eq:y-condition} itself is an $O(n)$ operation for an overall time of $O(n^2)$. However, \eqref{eq:y-condition} can be verified for all $y \in \mathcal{Y}^x$ in one pass on the sorted $\mathcal{Y}^x$ for a running time of $O(n\log n)$; we describe how to do so.%
\begin{algorithm}[!t]
	\SetAlgoLined
	
	\KwInput{%
		Desired coverage level $\alpha$; $\mathcal{Y}^x$ and $\mathcal{S}^x$ computed as defined in  equations~\eqref{eq:y-k-main},~\eqref{eq:s-k-main} using the training data $\{(X_i, Y_i)\}_{i=1}^n$, test point $x$ and any sequence of nested sets $\mathcal{F}^{-i}_t(\cdot)$ \;}
	\KwOutput{Prediction set $C^x \subseteq \mathcal{Y}$}
	
	$threshold \gets \alpha(n+1) - 1$;
	if $threshold < 0$, then return $\mathbb{R}$ and stop;\\
	$C^x \gets \emptyset$; $count  \gets 0$; $left\_endpoint \gets 0$\;
	\For{$i \gets 1$ \textbf{to} $|\mathcal{Y}^x|$}{
		\eIf{$s^x_i = 1$}{ \label{line:condition-1}
			$count \gets count + 1$\; \label{line:increase-count}
			\If{$count  > threshold $ \textbf{and} $count -1 \leq threshold$}{ \label{line:condition-2}
				$left\_endpoint \gets y^x_i$\; \label{line:update-left-endpoint}
			}
		}{
			\If{$count > threshold$ \textbf{and} $count - 1\leq threshold$}{ \label{line:condition-3}
				$C^x \gets C^x \cup \{[left\_endpoint, y^x_i]\}$\; \label{line:update-prediction-set}
			}
			$count \gets count - 1$\; \label{line:decrease-count}
		}
	}
	\Return $C^x$\;
	\caption{Efficient cross-conformal style aggregation}
	\label{alg:efficientCC}
\end{algorithm}

Let the sorted order of the points $\mathcal{Y}^x$ be $y^x_1 \leq y^x_2 \leq \ldots \leq y^x_{|{\mathcal{Y}}^x|}$. If $\mathcal{Y}^x$ contains repeated elements, we require that the left end-points $\ell_i$ come earlier in the sorted order than the right end-points $u_i$ for the repeated elements. Also define the bag of indicators $\mathcal{S}^x$ with elements $s^x_1 \leq s^x_2 \leq \ldots \leq s^x_{|{\mathcal{Y}^x}|}$, where 
\begin{equation}
\label{eq:s-k-main}
s^x_i :=
\begin{cases}
1  & \mbox{if } y^x_i \mbox{ corresponds to a left end-point}\\
0 & \mbox{if } y^x_i \mbox{ corresponds to a right end-point}.
\end{cases}
\end{equation}
Given $\mathcal{Y}^x$ and $\mathcal{S}^x$, Algorithm~\ref{alg:efficientCC} describes how to compute the cross-conformal prediction set in one pass (thus time $O(n)$) for every test point. Thus the runtime (including the sorting) is $O(n\log n)$ time to compute the predictions $\mathcal{F}_{r_i(X_i, Y_i)}^{-i}(x)$ for every $i$. If each prediction takes time $\leq T$, the overall time is $O(n\log n)+Tn$, which is the same as jackknife+.\footnote{For jackknife+, using quick-select, we could obtain $O(n) + Tn$ randomized, but the testing time $Tn$ usually dominates the additional $n\log n$ required to sort.} 
\begin{prop}
	\label{prop:efficientCC-correctness}
	Algorithm~\ref{alg:efficientCC} correctly computes the cross-conformal prediction set~\ref{eq:LOO} given $\mathcal{Y}^x$ and $\mathcal{S}^x$. 
\end{prop}
The proof of the proposition is in Appendix~\ref{appsec:alg-proof}. The proof proceeds through a step-wise description of the algorithm that makes it transparent how the algorithm verifies \eqref{eq:y-condition}
for every value of $y \in \mathcal{Y}^x$.  %
\section{Extending ensemble based out-of-bag conformal methods using nested sets}
\label{sec:OOB-conformal}
Cross-conformal, jackknife+, and their K-fold versions perform multiple splits of the data and for every training point $(X_i, Y_i)$, a residual function $r_i$ is defined using a set of training points that does not include $(X_i, Y_i)$. 
In the previous section, our description required training the base algorithm multiple times on different splits of the data. Often each of these individual algorithms is itself an ensemble algorithm (such as random forests). As described in this section, an ensemble algorithm naturally provide multiple (random) splits of the data from a single run and need not be re-trained on different splits to produce conformal prediction sets. This makes the conformal procedure computationally efficient. %
At the same time, like cross-conformal, the conformal prediction sets produced here are often shorter than split conformal because they use all of the training data for prediction. 
In a series of interesting papers \citep{johansson2014regression, bostrom2017accelerating, linusson2019efficient, kim2020predictive}, many authors have exhibited promising empirical evidence that these ensemble algorithms improve the width of prediction sets without paying a computational cost. We call this the \emph{OOB-conformal} method (short for out-of-bag). \rev{\citet{linusson2019efficient} provided an extensive empirical comparison of OOB-conformal to other conformal methods but without formal validity guarantees.}%

We now describe the procedure formally within the nested conformal framework, thus extending it instantly to residual functions that have hitherto not been considered. Our procedure can be seen as a generalization of the OOB-conformal method~\citep{linusson2019efficient} or the jackknife+ after bootstrap method \citep{kim2020predictive}:
\begin{enumerate}
	\item Let $\{M_j\}_{j = 1}^K$ denote $K \geq 1$ independent and identically distributed random sets drawn uniformly from $\{M: M\subset[n], |M| = m\}$. This is the same as subsampling. Alternatively $\{M_j\}_{j=1}^K$ could be i.i.d. random bags, where each bag is obtained by drawing $m$ samples with replacement from $[n]$. This procedure corresponds to bootstrap. 
	\item For every $i\in [n]$, define 
	\[M_{-i} := \{j : i \notin M_j\},\]
	 which contains the indices of the training sets that are \emph{out-of-bag} for the $i$-th data point. 
	\item The idea now is to use an ensemble learning method that, for every $i$, aggregates $|M_{-i}|$ many predictions to identify a single collection of nested sets $\{\mathcal{F}_t^{-i}(x)\}_{t \in \mathcal{T}}$. For instance, one can obtain an estimate $\widehat{\mu}_{j}(\cdot)$ of the conditional mean based on the training data corresponding to $M_j$, for every $j$, and then construct 
	\[
	\mathcal{F}_t^{-i}(x) = [\widehat{\mu}_{-i}(x) - t, \widehat{\mu}_{-i}(x) + t],
	\]
	where $\widehat{\mu}_{-i}(\cdot)$ is some combination (such as the mean) of $\{\widehat{\mu}_j(\cdot)\}_{\{j: i\notin M_j\}}$.
	\item The remaining conformalization procedure is identical to $C^{\texttt{LOO}}(x)$ described in Section~\ref{sec:Jackknife-plus}. Define the residual score 
	$r_i(x, y) := \inf\left\{t\in\mathcal{T}:\,y\in\mathcal{F}_t^{-i}(x)\right\}.$
\end{enumerate}
Using the cross-conformal scheme, the prediction set for any $x \in \mathcal{X}$ is given as
\begin{equation}
C^{\texttt{OOB}}(x) := \left\{y:\,\sum_{i=1}^n \mathbbm{1}\{r_i(X_i, Y_i) < r_i(x, y)\} < (1-\alpha)(n+1)\right\}.
\label{eq:def-oob-cross}
\end{equation}
If $\mathcal{F}_t^{-i}(x)$ is an interval for all $1\le i\le n$ and $x\in\mathcal{X}$, then following the discussion in Section~\ref{sec:intervals-LOO}, we could also derive a jackknife+ style prediction interval that is guaranteed to be non-empty:
\begin{equation}
C^{\texttt{OOB-JP}}(x) := [q^{-}_{n,\alpha}(\ell_i(x)), q^{+}_{n,\alpha}(u_i(x))].
\label{eq:def-oob-jp}
\end{equation}

If $\mathcal{F}_t^{-i}(x)$ could further contain empty sets, a jackknife+ interval can still be derived following the discussion in Appendix~\ref{app:empty-case}, but we skip these details here. Once again, we have that for every $x \in \mathcal{X}$, $C^{\texttt{OOB}}(x) \subseteq C^{\texttt{OOB-JP}}(x)$; see Equation~\eqref{eq:loo-jp} for details. The computational discussion of Section~\ref{subsec:compuation-cross} extends to $C^{\texttt{OOB}}$.

Recently, \citet{kim2020predictive} provided a  $1-2\alpha$ coverage guarantee of the OOB-conformal method when $\mathcal{F}_t^{-i}(x) = [\widehat{\mu}_{-i}(x) - t, \widehat{\mu}_{-i}(x) + t]$ where $\widehat{\mu}_{-i}(\cdot)$ represents the aggregation of conditional mean estimate from $\{M_j\}_{i\notin M_j}$. We generalize their result to any sequence of nested sets and extend it to the cross-conformal style aggregation scheme. In order to obtain a coverage guarantee, the conformal method must ensure a certain exchangeability requirement is satisfied. To do so, the argument of \citet{kim2020predictive} required the number of random resamples $K$ to itself be drawn randomly from a binomial distribution. We assert the same requirement in the following theorem (proved in~Appendix~\ref{appsec:OOB-conformal}).
\begin{thm}\label{thm:OOB-cross-validity}
	Fix a permutation invariant ensemble technique that constructs sets $\{\mathcal{F}_t^{-i}\}_{t\in \mathcal{T}}$ given a collection of subsets of $[n]$. Fix integers $\widetilde{K}, m \geq 1$ and let  
	\begin{align*}
	K ~&\sim~ \mathrm{Binomial}\left(\widetilde{K}, \left(1 - \frac{1}{n+1}\right)^m\right) \quad (\mbox{in the case of bagging}),\;\mathrm{or},\\
	K ~&\sim~ \mathrm{Binomial}\left(\widetilde{K}, 1 - \frac{m}{n+1}\right)\quad (\mbox{in the case of subsampling}).
	\end{align*} 
	Then 
	$\mathbb{P}(Y_{n+1}\in  C^{\texttt{OOB}}(X_{n+1})) \ge 1- 2\alpha.$
\end{thm}
Because $C^{\texttt{OOB}}(x) \subseteq C^{\texttt{OOB-JP}}(x)$ for every $x \in \mathcal{X}$, the validity guarantee continues to hold for $C^{\texttt{OOB-JP}}(\cdot)$.
While we can only prove a $1-2\alpha$ coverage guarantee, it has been observed empirically that the OOB-conformal method with regression forests as the ensemble scheme and nested sets  $\{[\widehat{\mu}(x) - t\widehat{\sigma}(x), \widehat{\mu}(x) + t\widehat{\sigma}(x)]\}_{t \in \mathbb{R}^+}$ satisfies $1-\alpha$ coverage while providing the shortest prediction sets on average~\citep{bostrom2017accelerating}. %
On the other hand, the best empirically performing nested sets are the ones introduced by \citet{romano2019conformalized}: $\{[\widehat{q}_{\beta}(x) - t, \widehat{q}_{\beta}(x) + t]\}_{t \in \mathbb{R}}$ (for an appropriately chosen $\beta$). Using nested conformal, we show how these these two ideas can be seamlessly integrated: quantile based nested sets with an OOB-style aggregation scheme. In Section~\ref{sec:QOOB} we formally develop our novel method QOOB, and in Section~\ref{sec:numerical-experiments} we empirically verify that it achieves competitive results in terms of the length of prediction sets.

\section{QOOB: A novel conformal method using nested sets}
\label{sec:QOOB}
\rev{The nested conformal interpretation naturally separates the design of conformal methods into two complementary aspects: 
\begin{enumerate}[label=(\alph*)]
	\item identifying an information efficient nonconformity score based on a set of nested intervals, and
	\item performing sample efficient aggregation of the nonconformity scores while maintaining validity guarantees.
\end{enumerate}
In this section, we leverage this dichotomy to merge two threads of ideas in the conformal literature and develop a novel conformal method that empirically achieves state-of-the-art results in terms of the width of prediction sets.} 

First, we review what is known on aspect (b). While split-conformal based methods are computationally efficient, they lose sample efficiency due to sample splitting. Aggregated conformal methods such as cross-conformal, jackknife+, and OOB-conformal do not have this drawback and are the methods of choice for computationally feasible and sample efficient prediction sets. Among all aggregation techniques, the OOB-conformal method has been observed empirically to be the best aggregation scheme which uses all the training data efficiently~\citep{bostrom2017accelerating}. 

Next, we consider aspect (a), the design of the nested sets. The nested sets considered by \citet{bostrom2017accelerating} are $\{[\widehat{\mu}(x) - t\widehat{\sigma}(x), \widehat{\mu}(x) + t\widehat{\sigma}(x)]\}_{t \in \mathbb{R}^+}$ based on mean and variance estimates obtained using out-of-bag trees. On the other hand, it has been demonstrated that nested sets based on quantile estimates $\widehat{q}_s(x)$ given by $\{[\widehat{q}_{\beta}(x) - t, \widehat{q}_{1-\beta}(x) + t]\}_{t \in \mathbb{R}}$ perform better than those based on mean-variance estimates in the split conformal setting~\citep{romano2019conformalized, sesia2019comparison}. 

Building on these insights, we make the following suggestion: Quantile Out-of-Bag (QOOB) conformal; pronounced ``cube'' conformal. This method works in the following way. First, a quantile regression based on random forest~\citep{meinshausen2006quantile} with $T$ trees is learnt by subsampling or bagging the training data $T$ times. Next, the out-of-bag trees for every training point $(X_i, Y_i)$ are used to learn a quantile estimator function $\widehat{q}^{-i}_s(\cdot)$  for $s = \beta$ and $s = 1 - \beta$. Here $\beta = k\alpha$ for some constant $k$. Now for every $i$ and some $x \in \mathcal{X}$, we define the nested sets as
\[
\mathcal{F}_t^{-i}(x) ~:=~ [\widehat{q}_\beta^{-i}(x) - t, \widehat{q}_{1-\beta}^{-i}(x) + t].
\]

The nonconformity scores based on these nested sets are aggregated to provide a prediction set as described by~$C^{\texttt{OOB}}(x)$ in~\eqref{eq:def-oob-cross} of Section~\ref{sec:OOB-conformal}. Algorithm~\ref{alg:QOOB} describes QOOB procedurally. Following Section~\ref{subsec:compuation-cross}, the aggregation step of QOOB (line~\ref{line:qoob-aggregation}, Algorithm~\ref{alg:QOOB}) can be performed in time $O(n\log n)$. 
\begin{algorithm}[!t]
	\SetAlgoLined
	\KwInput{Training data $ \{(X_i, Y_i)\}_{i=1}^n$, test point $x$, desired coverage level $\alpha$, number of trees $T$, nominal quantile level $\beta$ (default $=2\alpha$)}
	\KwOutput{Prediction set $C^x \subseteq \mathcal{Y}$}
	$\{M_j\}_{j = 1}^T \gets$ training bags drawn independently from $[n]$ using subsampling or bootstrap\;
	$\{M_{-i}\}_{i = 1}^n \gets  \{j : i \notin M_j\}$\;
	\For{$j \gets 1$ \textbf{to} $T$}{
		$\phi_j \gets $ Quantile regression trees learnt using the data-points in $M_j$\\\qquad \qquad (this step could include subsampling of features)\;
	}
	\For{$i \gets 1$ \textbf{to} $n$}{
		$\Phi_{-i} \gets \{\phi_j: j \in M_{-i}\}$\;
		$\widehat{q}^{-i}(\cdot)  \gets$ quantile regression forest using the trees $\Phi_{-i}$\;
		$\mathcal{F}_t^{-i}(\cdot) \gets  [\widehat{q}_\beta^{-i}(\cdot) - t, \widehat{q}_{1-\beta}^{-i}(\cdot) + t]$\; 
		$r_i(\cdot, \cdot) \gets ((x, y) \to \inf\left\{t\in\mathcal{T}:\,y\in\mathcal{F}_t^{-i}(x)\right\})$\;
	}
	
	$C^x \gets $ OOB prediction set defined in Equation~\eqref{eq:def-oob-cross}; (call Algorithm~\ref{alg:efficientCC} with $\mathcal{Y}^x$, $\mathcal{S}^x$ computed using $\mathcal{F}_t^{-i}(\cdot)$, the training data $\{(X_i, Y_i)\}_{i=1}^n$ and test point $x$ as described in equations~\eqref{eq:y-k-main}, ~\eqref{eq:s-k-main}) \label{line:qoob-aggregation}\\
	\Return $C^x$\;
	\caption{Quantile Out-of-Bag conformal (QOOB)}
	\label{alg:QOOB}
\end{algorithm}

Since QOOB is a special case of OOB-conformal, it inherits an assumption-free $1-2\alpha$ coverage guarantee from Theorem~\ref{thm:OOB-cross-validity} if $K$ is drawn from an appropriate binomial distribution as described in the theorem. In practice, we typically obtain $1-\alpha$ coverage with a fixed $K$. In Section~\ref{sec:numerical-experiments}, we empirically demonstrate that QOOB achieves state-of-the-art performance on multiple real-world datasets. We also discuss three aspects of our method: 
\begin{enumerate}[label=(\alph*)]
	\item how to select the nominal quantile level $\beta = k\alpha$, 
	\item the effect of the number of trees $T$ on the performance, and 
	\item the performance of the jackknife+ version of our method (QOOB-JP), which corresponds to the OOB-JP style aggregation (equation~\eqref{eq:def-oob-jp}) of quantile-based nonconformity scores.
\end{enumerate}

\section{Numerical comparisons}\label{sec:numerical-experiments}
We compare several methods discussed in this paper using synthetic and real datasets. MATLAB code to execute QOOB and reproduce the experiments in this section is provided at \url{https://github.com/AIgen/QOOB}. Some experiments on synthetic data are discussed in Section~\ref{subsec:synthetic}; the rest of this section discusses results on real datasets. We use the following 6 datasets from the UCI repository: blog feedback, concrete strength, superconductivity, news popularity, kernel performance and protein structure. Metadata and links for these datasets are provided in Appendix~\ref{appsec:additional-info-exps}. Appendix~\ref{appsec:more-exps} also contains experimental results with 5 additional UCI datasets. 

In order to assess the coverage and width properties, we construct 100 versions of each of the datasets, by independently drawing 1000 random data points (uniformly without replacement) from the full dataset. Then we split each such version into two parts: training and testing, with sizes $768$\footnote{the number of training points is divisible by many factors, which is useful for creating a varying number of folds for K-fold methods} and $232$ respectively. Hence corresponding to each dataset, we get 100 different versions with 768 training and 232 testing points. 

For each compared conformal method, we report the following two metrics:
\begin{itemize}
	
	\item \emph{Mean-width}: For a prediction set $C(x)\subseteq \mathcal{Y} = \mathbb{R}$ its width is defined as its Lebesgue measure. For instance, if $C(x)$ is an interval, then the width is its length and if $C(x)$ is a union of two or more disjoint intervals, then the width is the sum of the lengths of these disjoint intervals. We report the average over the mean-width given by
	\begin{equation}\label{eq:mean-width}
	\mbox{Ave-Mean-Width} := \frac{1}{100}\sum_{b = 1}^{100}\, \left(\frac{1}{232}\sum_{i=1}^{232} \mbox{width}({C}_b(X_i^b))\right).
	\end{equation}
	Here ${C}_b(\cdot)$ is a prediction set learnt from the $b$-th version of a dataset. The outer mean is the average over 100 versions of a dataset. The inner mean is the average of the width over the testing points in a particular version of a dataset.
	\item \emph{Mean-coverage}: We have proved finite-sample coverage guarantees for all our methods and to verify (empirically) this property, we also report the average over the mean-coverage given by
	\begin{equation}\label{eq:mean-coverage}
	\mbox{Ave-Mean-Coverage} := \frac{1}{100}\sum_{b=1}^{100}\,\left(\frac{1}{232}\sum_{i=1}^{232} \mathbbm{1}\{Y_i^b \in {C}_b(X_i^b)\}\right).
	\end{equation}
	
\end{itemize} 
In addition to reporting the average over versions of a dataset, we also report the estimated standard deviation of the average (to guage the fluctuations). In the rest of the discussion, the qualification `average' may be skipped for succinctness, but all reports and conclusions are to be understood as comments on the average value for mean-width and mean-coverage. 

Random forest (RF) based regressors perform well across different conformal methods and will be used as the base regressor in our experiments, with varying  $T$, the number of trees. Each tree is trained on an independently drawn bootstrap sample from the training set (containing about $(1-1/e)100\% \approx 63.2\%$ of all training points).
The numerical comparisons will use the following methods:
\begin{enumerate}
	\item SC-$T$: Split conformal~\citep{papadopoulos2002inductive, lei2018distribution} with nested sets $\{[\widehat{\mu}(x) - t, \widehat{\mu}(x) + t]\}_{t \in \mathbb{R}^+}$ and $T$ trees. 
	\item Split-CQR-$T$ ($2\alpha$): Split conformalized quantile regression~\citep{romano2019conformalized} with $T$ trees and nominal quantile level $2\alpha$. This corresponds to the nested sets $
	\{ [\widehat{q}_{2\alpha}^{-i}(x) - t, \widehat{q}_{1-{2\alpha}}^{-i}(x) + t]\}_{t \in \mathbb{R}}$. Quantile conformal methods require the nominal quantile level to be set carefully, as also noted by~\citet{sesia2019comparison}. In our experiments, we observe that Split-CQR-$T$ performs well at the nominal quantile level $2\alpha$. This is discussed more in Section~\ref{subsec:nominal-quantile}. 
	\item 8-fold-CC-$T$: 8-fold cross-conformal~\citep{vovk2015cross, barber2019predictive} with $T$
	trees learnt for every fold and the nested sets $\{[\widehat{\mu}(x) - t, \widehat{\mu}(x) + t]\}_{t \in \mathbb{R}^+}$. Leave-one-out cross-conformal is computationally expensive if $T$ trees are to be trained for each fold, and does not lead to significantly improved performance compared to OOB-CC in our experiments. Hence we did not report a detailed comparison across all datasets. 
	\item OOB-CC-$T$: OOB-cross-conformal~\citep{johansson2014regression, kim2020predictive} with $T$
	trees. This method considers the nested sets $\{[\widehat{\mu}(x) - t, \widehat{\mu}(x) + t]\}_{t \in \mathbb{R}^+}$ where $\widehat{\mu}$ is the average of the mean-predictions for $x$ on out-of-bag trees. 
	\item OOB-NCC-$T$: OOB-normalized-cross-conformal~\citep{bostrom2017accelerating} with $T$ trees. This method considers nested sets $\{[\widehat{\mu}(x) - t\widehat{\sigma}(x), \widehat{\mu}(x) + t\widehat{\sigma}(x)]\}_{t \in \mathbb{R}^+}$ where $\widehat{\sigma}(x)$ is the standard deviation of mean-predictions for $x$ on out-of-bag trees.
	\item QOOB-$T$ ($2\alpha$): OOB-quantile-cross-conformal with $T$ trees and nominal quantile level $\beta = 2\alpha$. This is our proposed method. In our experiments, we observe that QOOB-T performs well at the nominal quantile level $2\alpha$. We elaborate more on the nominal quantile selection in Section~\ref{subsec:nominal-quantile}. %
\end{enumerate}

\begin{table}[!h]
	\caption{Mean-width~\eqref{eq:mean-width} of conformal methods with regression forests ($\alpha = 0.1$). Average values across 100 simulations are reported with the standard deviation in brackets. }
	\centering
	\resizebox{\textwidth}{!}{
		\begin{tabular}{c|c|c|c|c|c|c}
			\hline
			\textbf{Method} & {Blog} & {Protein} & {Concrete}& {News} & {Kernel} & {Superconductivity} \\ \hline \hline 
			SC-100 & 25.54 & 16.88 & 22.29 & 12491.84 & 452.71 & 54.46 \\[-0.15in] 
			& (0.71) & (0.08) & (0.14) & (348.07) & (5.10) & (0.37) \\\hline
			Split-CQR-100 (2$\alpha$) & \textbf{12.22} & 14.20 & 21.45 & \textbf{7468.15} & \textbf{295.49} & 39.59 \\[-0.15in] 
			& (0.35) & (0.09) & (0.12) & (136.93) & (3.09) & (0.27)\\\hline
			8-fold-CC-100 & 24.83 & 16.42 & 19.23 & 12461.40 & 411.81 & 50.30 \\[-0.15in] 
			& (0.44) & (0.05) & (0.04) & (263.54) & (3.4299) & (0.24) \\\hline
			OOB-CC-100 & 24.76 & 16.38 & 18.69 & 12357.58 & 402.97 & 49.31 \\[-0.15in] 
			& (0.50) & (0.04) & (0.03) & (213.72) & (3.13) & (0.24) \\\hline
			OOB-NCC-100 & 20.31 & 14.87 & 18.66 & 11500.22 & 353.35 & 39.55 \\ [-0.15in]
			& (0.42) & (0.05) & (0.06) & (320.91) & (2.95) & (0.22)\\\hline
			QOOB-100 ($2\alpha$) & 14.43 & \textbf{13.74} & \textbf{18.19} & 7941.19 & 300.04 & \textbf{37.04} \\[-0.15in] 
			& (0.38) & (0.05) & (0.05) & (89.21) & (2.70) & (0.18)\\\hline
	\end{tabular}}
	\label{table:overall-comparison-width}
\end{table}

\begin{table}[!h]
	\caption{Mean-coverage~\eqref{eq:mean-coverage} of conformal methods with regression forests ($\alpha = 0.1$).  Average values across 100 simulations are reported. The standard deviation of these average values are zero up to two significant digits. }
	\centering
	\resizebox{\textwidth}{!}{
		\begin{tabular}{c|c|c|c|c|c|c}
			\hline
			\textbf{Method} & {Blog} & {Protein} & {Concrete}& {News} & {Kernel} & {Superconductivity} \\ \hline \hline 
			SC-100 & 0.90 & 0.90 & 0.90 & 0.90 & 0.90 & 0.90 \\ \hline
			Split-CQR-100 ($2\alpha$) & 0.91 & 0.90 & 0.90 & 0.90 & 0.90 & 0.90 \\ \hline
			8-fold-CC-100 & 0.91 & 0.91 & 0.91 & 0.91 & 0.91 & 0.90 \\ \hline
			OOB-CC-100 & 0.90 & 0.91 & 0.91 & 0.91 & 0.91 & 0.90 \\ \hline
			OOB-NCC-100 & 0.92 & 0.91 & 0.91 & 0.92 & 0.93 & 0.91 \\ \hline
			QOOB-100 ($2\alpha$) & 0.92 & 0.91 & 0.92 & 0.91 & 0.93 & 0.91 \\ \hline
	\end{tabular}}
	\label{table:overall-comparison-coverage}
\end{table}
\label{subsec:QOOB-main}

Tables~\ref{table:overall-comparison-width} and~\ref{table:overall-comparison-coverage} report the mean-width and mean-coverage that these conformal methods achieve on 6 datasets. Here, the number of trees $T$ is set to $100$ for all the methods.
We draw the following conclusions: 

\begin{itemize}
	\item Our novel method QOOB achieves the shortest or close to the shortest mean-width compared to other methods while satisfying the $1-\alpha$ coverage guarantee. The closest competitor is Split-CQR. As we further investigate in Section~\ref{subsec:QOOB-trees}, even on datasets where Split-CQR performs better than QOOB, if the number of trees are increased beyond 100, QOOB shows a decrease in mean-width while Split-CQR does not improve. For example, on the kernel dataset, QOOB outperforms Split-CQR at 400 trees.
	\item In Table~\ref{table:overall-comparison-width}, QOOB typically has low values for the standard deviation of the average-mean-width across all methods. This entails more reliability to our method, which may be desirable in some applications. In Sections~\ref{subsec:nominal-quantile} and~\ref{subsec:QOOB-trees}, we observe that this property is true across different number of trees and nominal quantile levels as well. 
	\item On every dataset, QOOB achieves coverage higher than the prescribed value of $1-\alpha$, with a margin of 1-3\%. Surprisingly this is true even if it shortest mean-width among all methods. It may be possible to further improve the performance of QOOB in terms of mean-width by investigating the cause for this over-coverage. 
	\item OOB-CC does better than 8-fold-CC with faster running times. Thus, to develop QOOB, we chose to work with out-of-bag conformal. %
\end{itemize}

Appendix~\ref{appsec:more-exps} contains Table \ref{table:overall-comparison-width} and \ref{table:overall-comparison-coverage} style results for 5 additional UCI datasets. QOOB continues to perform well on these additional datasets. In the rest of this section, we present a more nuanced empirical study of QOOB. Each subsection focuses on a different aspect of QOOB's empirical performance: %
\begin{itemize}
	\item In Section~\ref{subsec:nominal-quantile}, we discuss the significant impact that nominal quantile selection has on the performance of QOOB and Split-CQR. We observe that $2\alpha$ is an appropriate nominal quantile recommendation for both methods. 
	\item In Section~\ref{subsec:QOOB-trees}, we show that increasing the number of trees $T$ leads to decreasing mean-widths for QOOB, while this is not true for its closest competitor Split-CQR. QOOB also outperforms other competing OOB methods across different values for the number of trees $T$. 
	\item In Section~\ref{subsec:QOOB-sample-size}, we compare QOOB and Split-CQR in the small sample size (small $n$) regime where we expect sample splitting methods to lose statistical efficiency. We confirm that QOOB significnatly outperforms Split-CQR on all 6 datasets we have considered for $n \leq 100$. 
	\item In Section~\ref{subsec:cc-jp}, we compare the related methods of cross-conformal and jackknife+ and demonstrate that there exist settings where cross-conformal leads to shorter intervals compared to jackknife+, while having a similar computational cost (as discussed in Section~\ref{subsec:compuation-cross}). This justifies performing OOB conformal with cross-conformal aggregation, instead of jackknife+ style aggregation as suggested by \citet{kim2020predictive}.
	\item In Section~\ref{subsec:synthetic}, we demonstrate that QOOB achieves conditional coverage on a synthetic dataset. We use the data distribution designed by \citet{romano2019conformalized} for demonstrating the conditional coverage of Split-CQR. %
\end{itemize}

\subsection{Nominal quantile selection has a significant effect on QOOB}
\label{subsec:nominal-quantile}
QOOB and Split-CQR both use nominal quantiles $\widehat{q}_\beta$, $\widehat{q}_{1-\beta}$ from a learnt quantile prediction model. In the case of Split-CQR, as observed by \citet{romano2019conformalized} and \citet{sesia2019comparison}, tuning $\beta$ leads to improved performance. %
We perform a comparison of QOOB and Split-CQR at different values of $\beta$. Figure~\ref{fig:width-nominal-level} reports mean-widths for QOOB-100 ($k\alpha$) and Split-CQR-100 ($k\alpha$), with OOB-NCC-100 as a fixed baseline (that does not vary with $k$). We observe that the nominal quantile level significantly affects the performance of Split-CQR and QOOB. Both methods perform well at the nominal quantile of about $2\alpha$. %
We encourage a more detailed study on the theoretical and empirical aspects of nominal quantile selection in future work. We also note that for all values of $k$, QOOB typically has smaller standard deviation of the average-mean-width compared to Split-CQR, implying more reliability in the predictions.

\begin{figure}[t!]
	\includegraphics[width=0.9\textwidth]{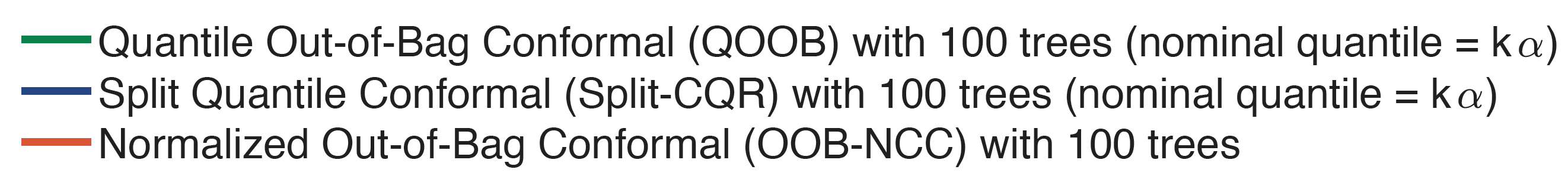}\vspace{-0.8cm}\\
	\centering
	\subfloat[Concrete structure.]{\includegraphics[width=0.33\textwidth]{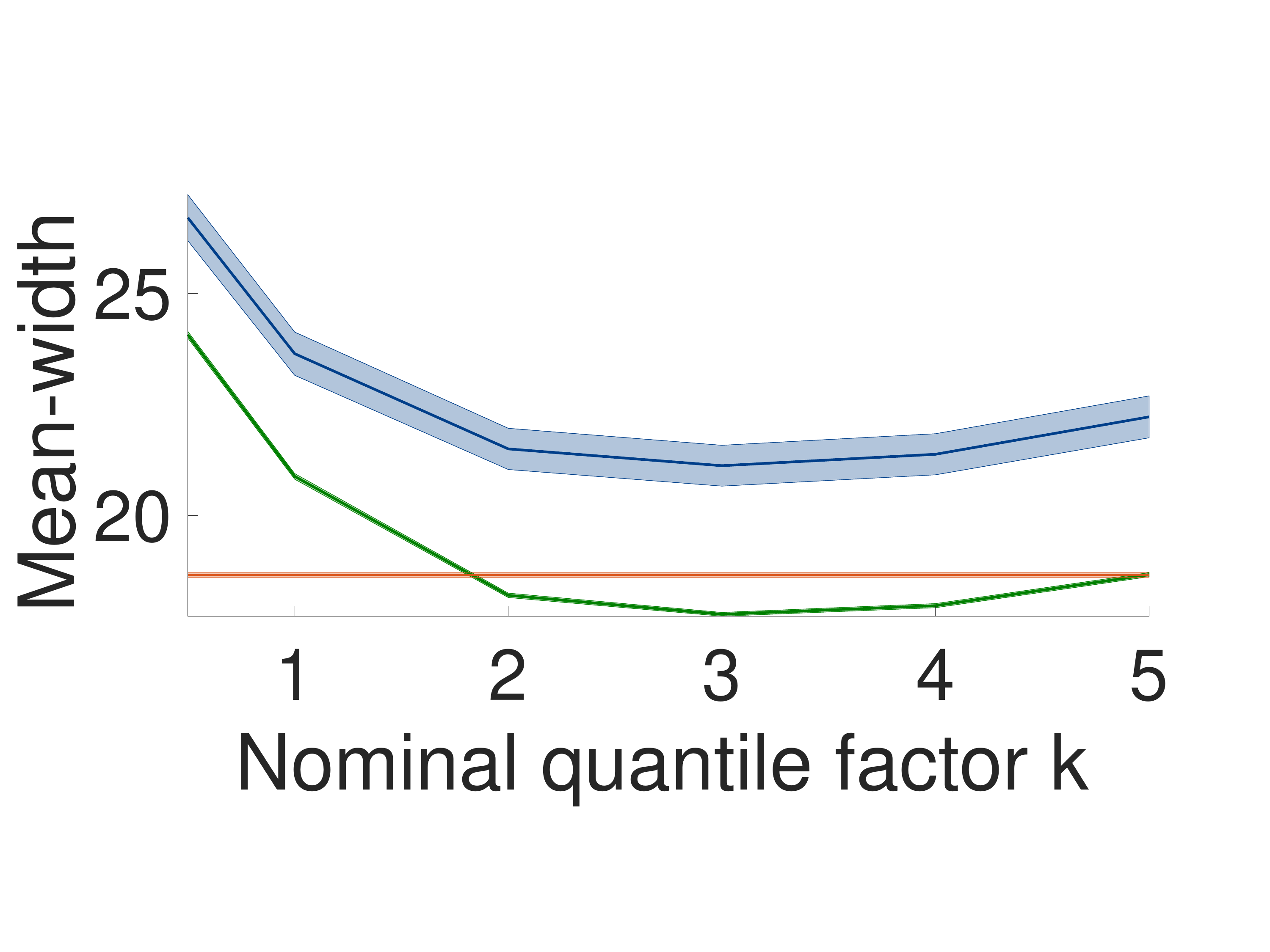}}
	\subfloat[Blog feedback.]{ \includegraphics[width=0.33\textwidth]{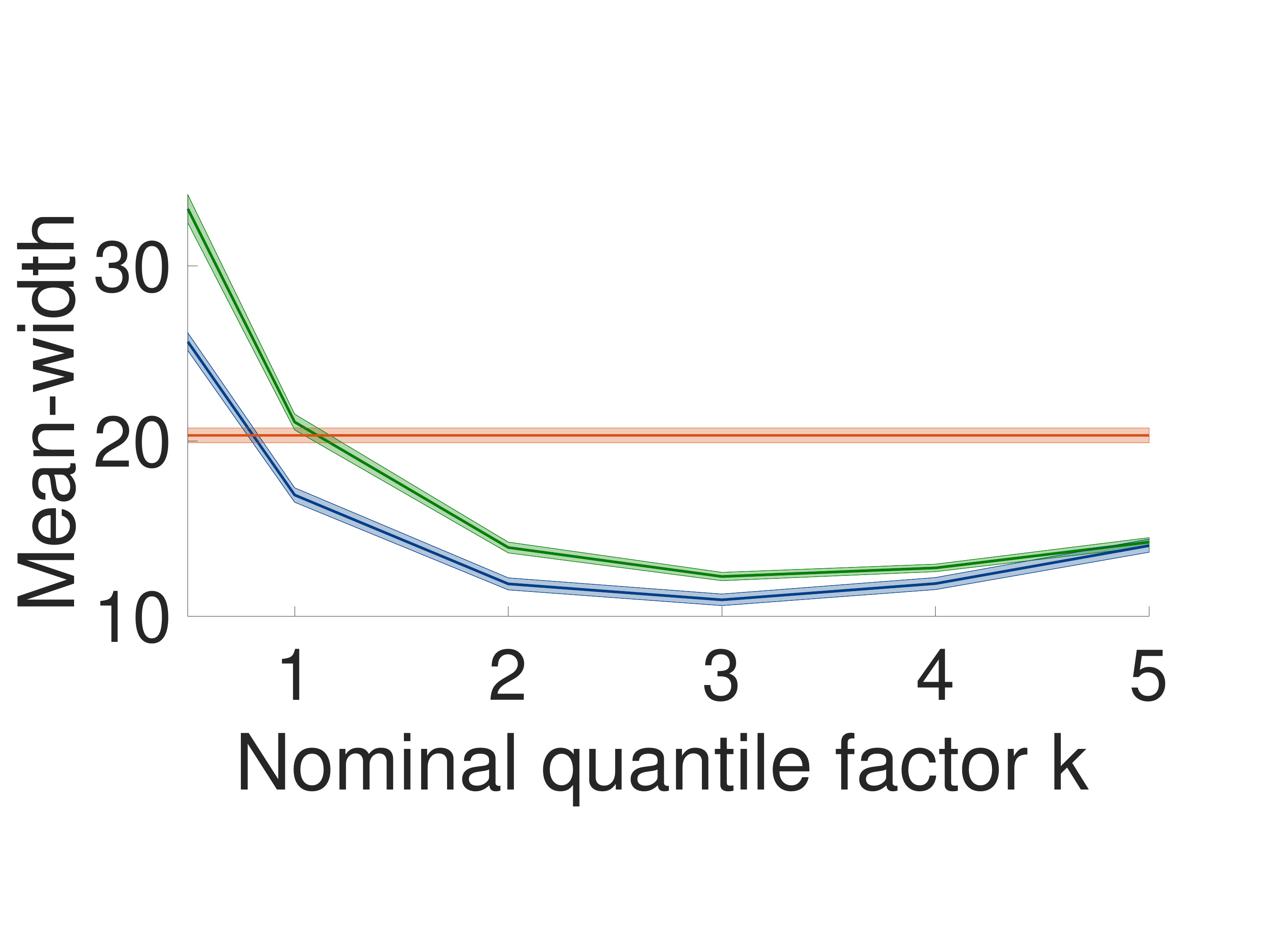}}
	\subfloat[Protein structure.]{\includegraphics[width=0.33\textwidth]{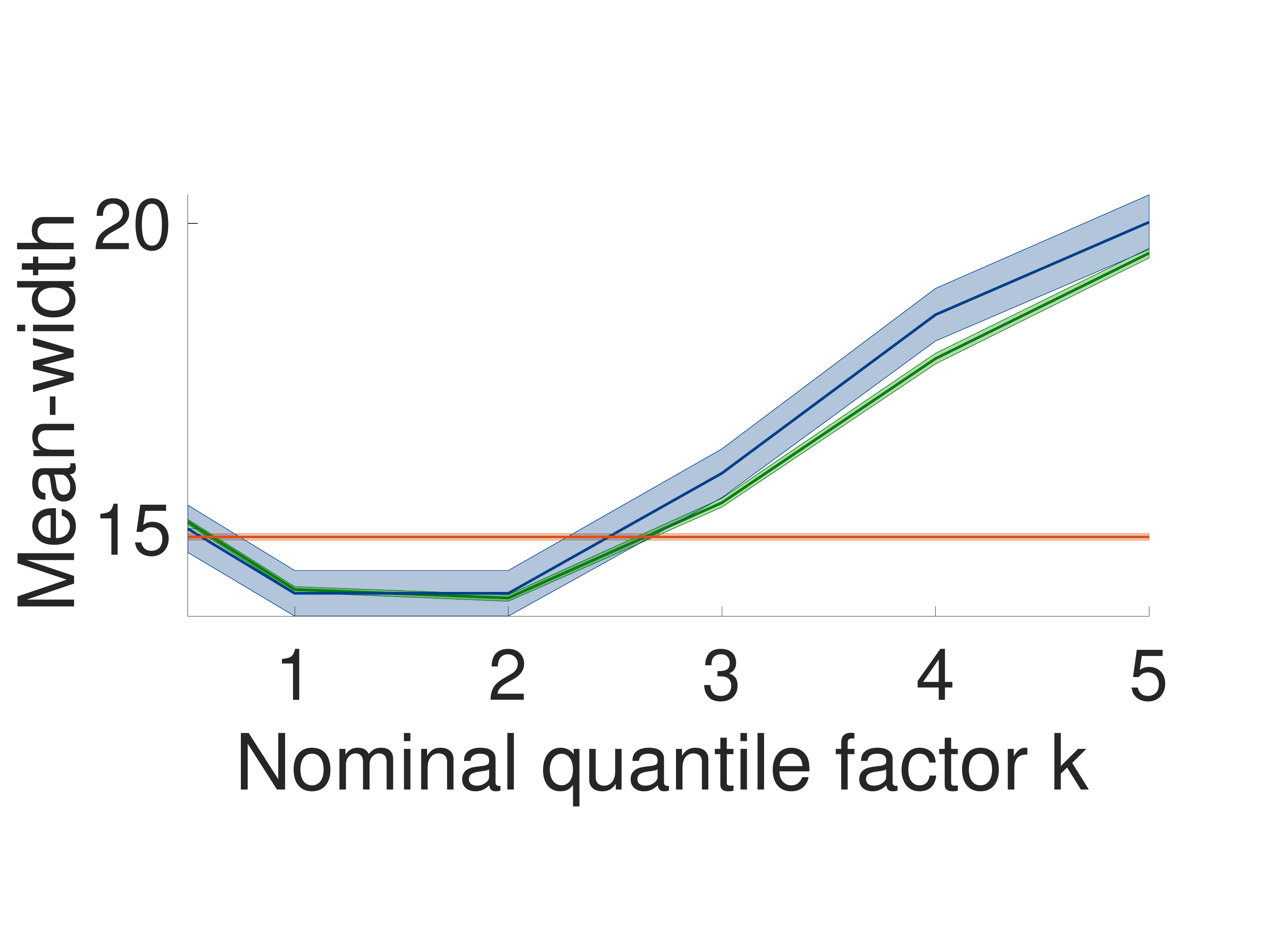}}\\
	\subfloat[Superconductivity.]{\includegraphics[width=0.33\textwidth]{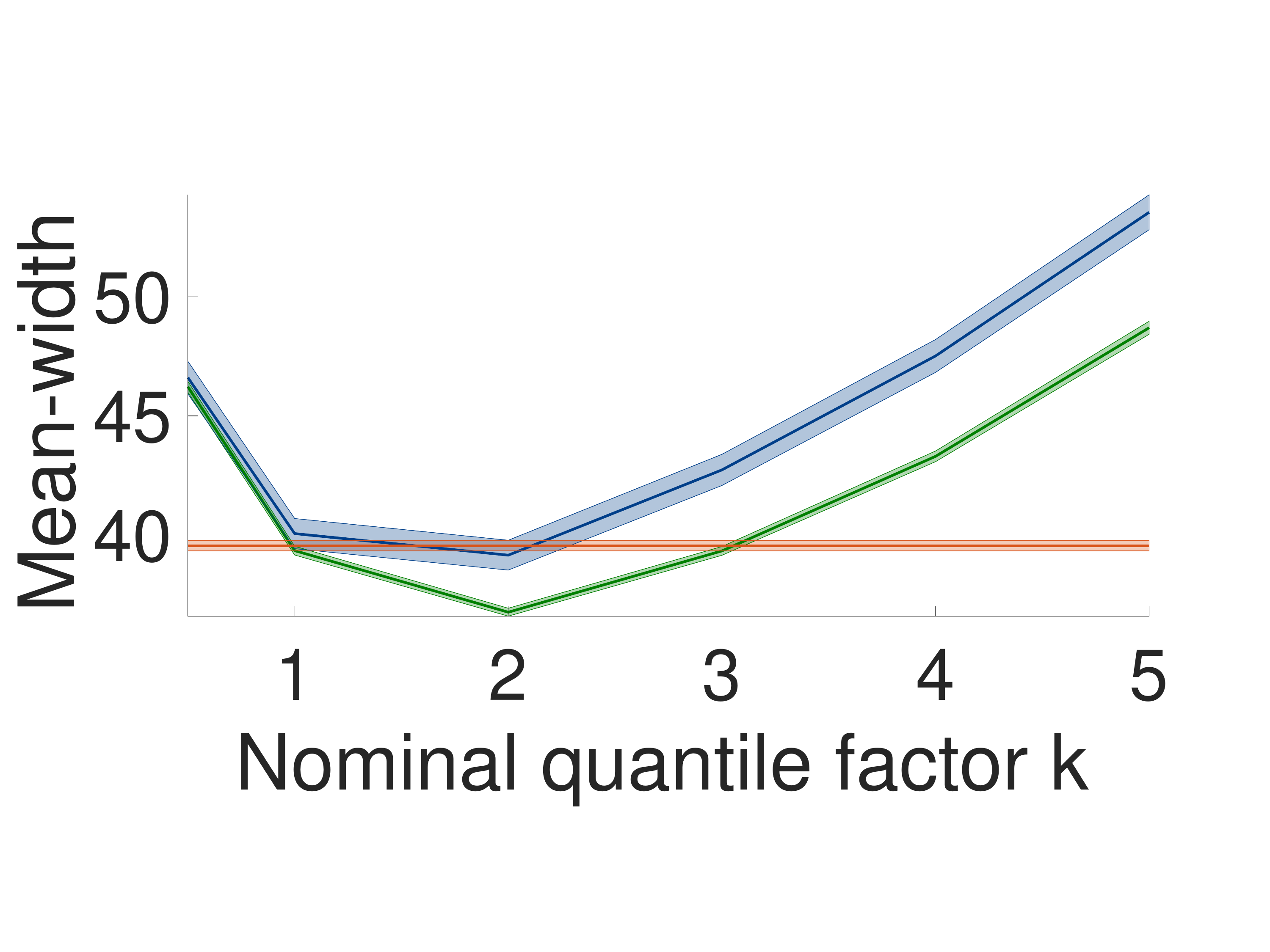}}
	\subfloat[News popularity.]{ \includegraphics[width=0.33\textwidth]{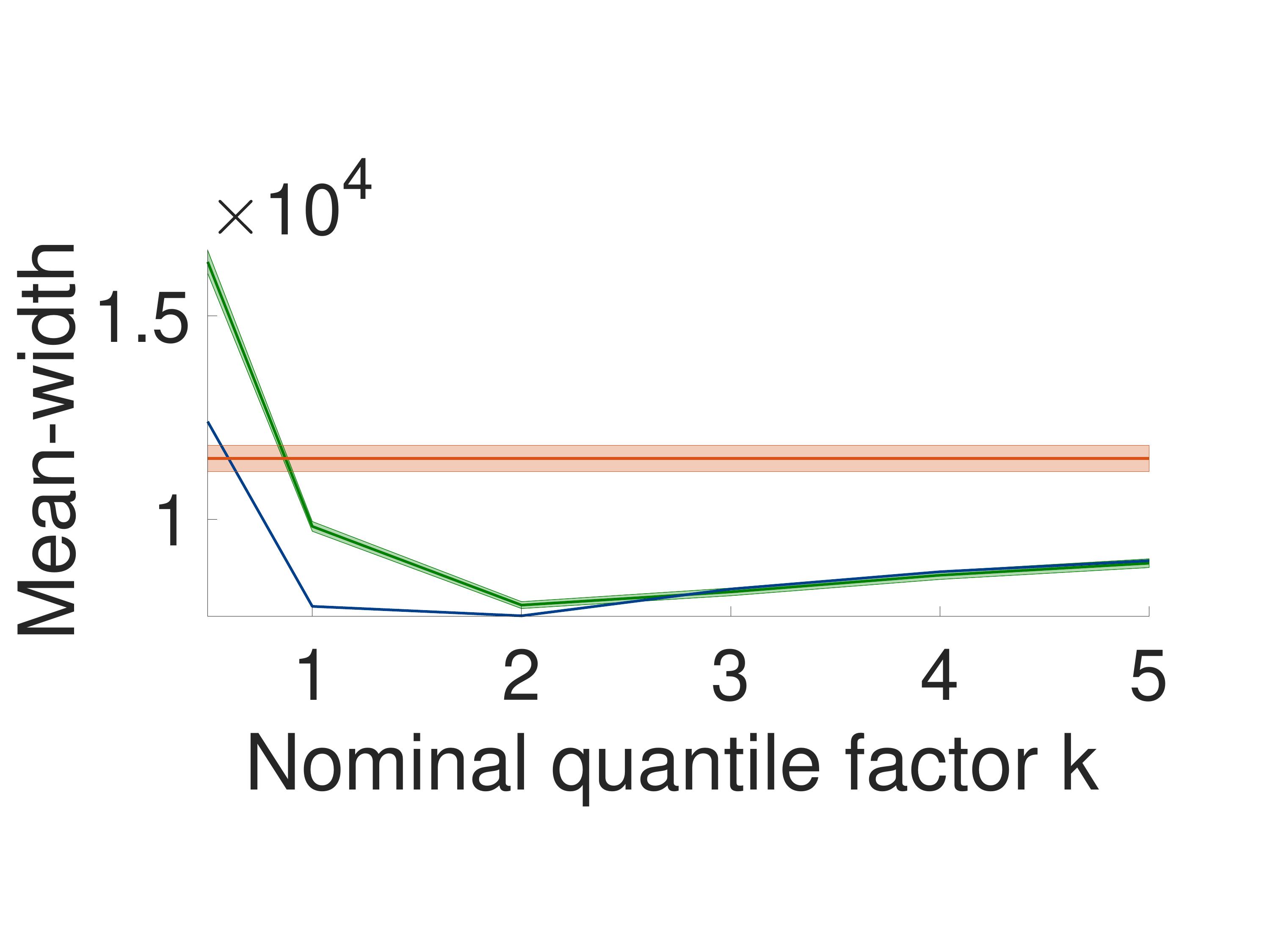}}
	\subfloat[Kernel performance.]{\includegraphics[width=0.33\textwidth]{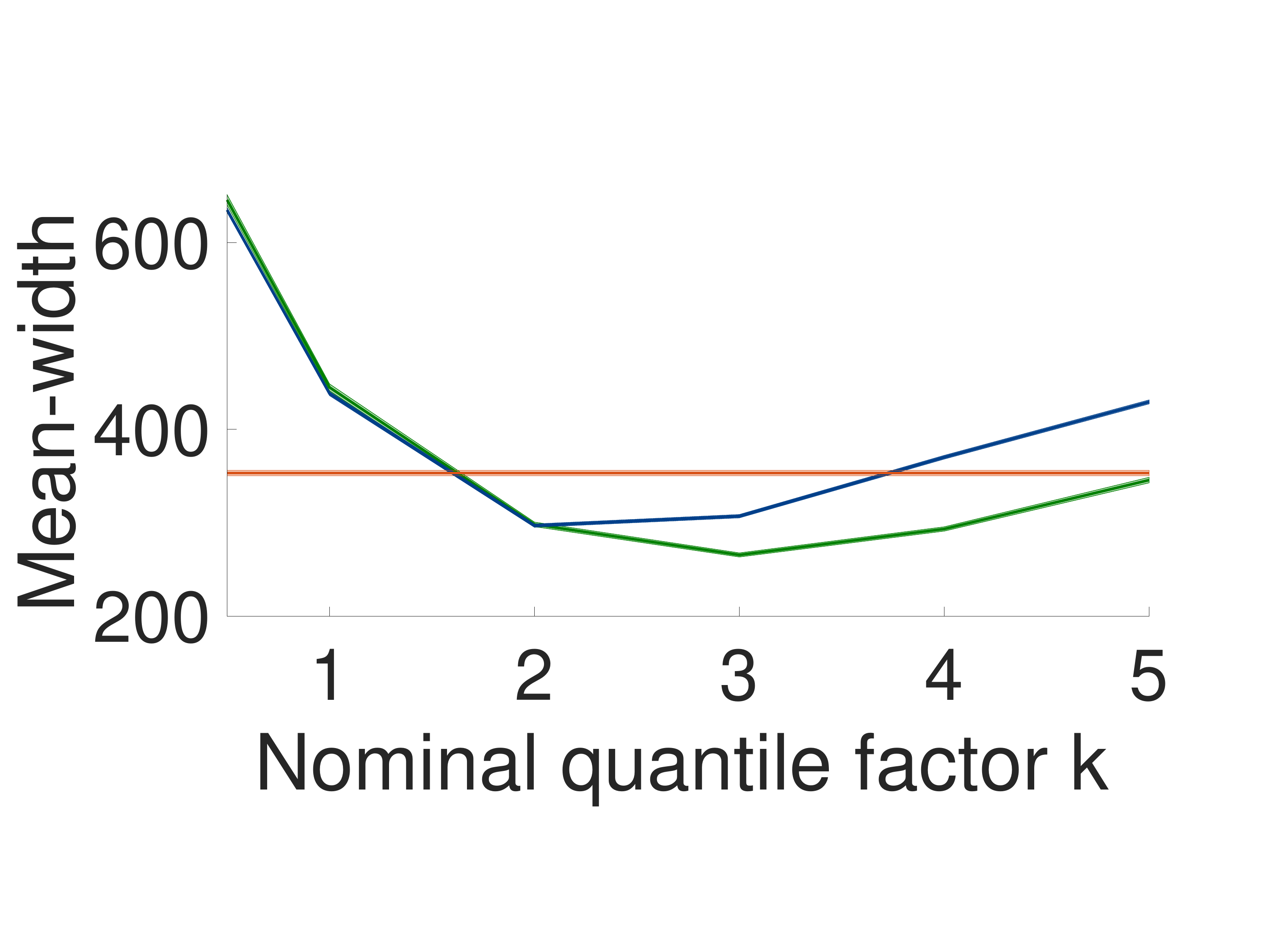}}
	\caption{QOOB and Split-CQR are sensitive to the nominal quantile level $\beta = k\alpha$. At $\beta \approx 2\alpha$, QOOB performs better than OOB-NCC for all datasets (OOB-NCC does not require nominal quantile tuning and is a constant baseline). For the plots above, $\alpha = 0.1$. All methods plotted have empirical mean-coverage at least $1-\alpha$. The mean-width values are averaged over 100 subsamples. The shaded area denotes $\pm 1$ std-dev for the average of mean-width. }
	\label{fig:width-nominal-level}
\end{figure}

\subsection{QOOB has shorter prediction intervals as we increase the number of trees}
\label{subsec:QOOB-trees}
In this experiment, we investigate the performance of the competitive conformal methods from Table~\ref{table:overall-comparison-width} as the number of trees $T$ are varied. For QOOB and Split-CQR, we fix the quantile level to $\beta = 2\alpha$. We also compare with OOB-NCC and another quantile based OOB method described as follows. Like QOOB, suppose we are given a quantile estimator $\widehat{q}_s(\cdot)$. Consider the quantile-based construction of nested sets suggested by~\citet{chernozhukov2019distributional}: 
\[\mathcal{F}_t(x) = [\widehat{q}_t(x), \widehat{q}_{1-t}(x)]_{t \in (0, 1/2)}.
\]
Using these nested sets and the OOB conformal scheme (Section~\ref{sec:OOB-conformal}) leads to the QOOB-D method (for `distributional' conformal prediction as the original authors called it). Since QOOB-D does not require nominal quantile selection, we considered this method as a possible solution to the nominal quantile problem of QOOB and Split-CQR (Section~\ref{subsec:nominal-quantile}). The results are reported in Figure~\ref{fig:width-trees} for $T$ ranging from 50 to 400.

\begin{figure}[h!]
	\centering
	\includegraphics[width=0.9\textwidth]{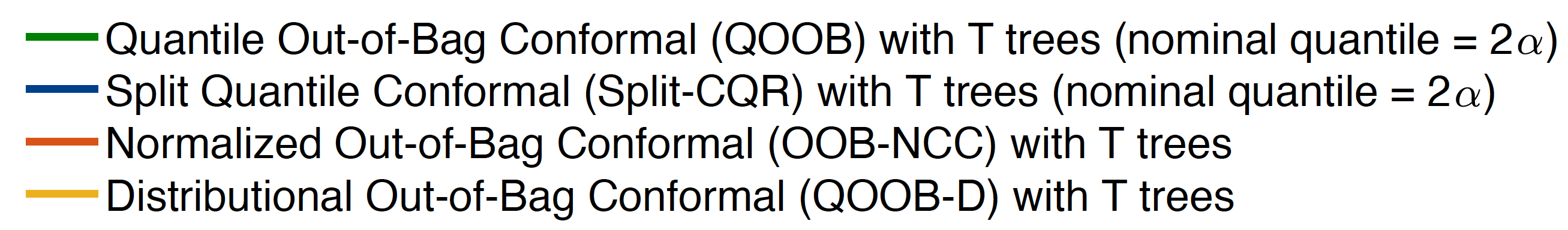}\vspace{-0.6cm}
	\subfloat[Concrete structure.]{\includegraphics[width=0.33\textwidth]{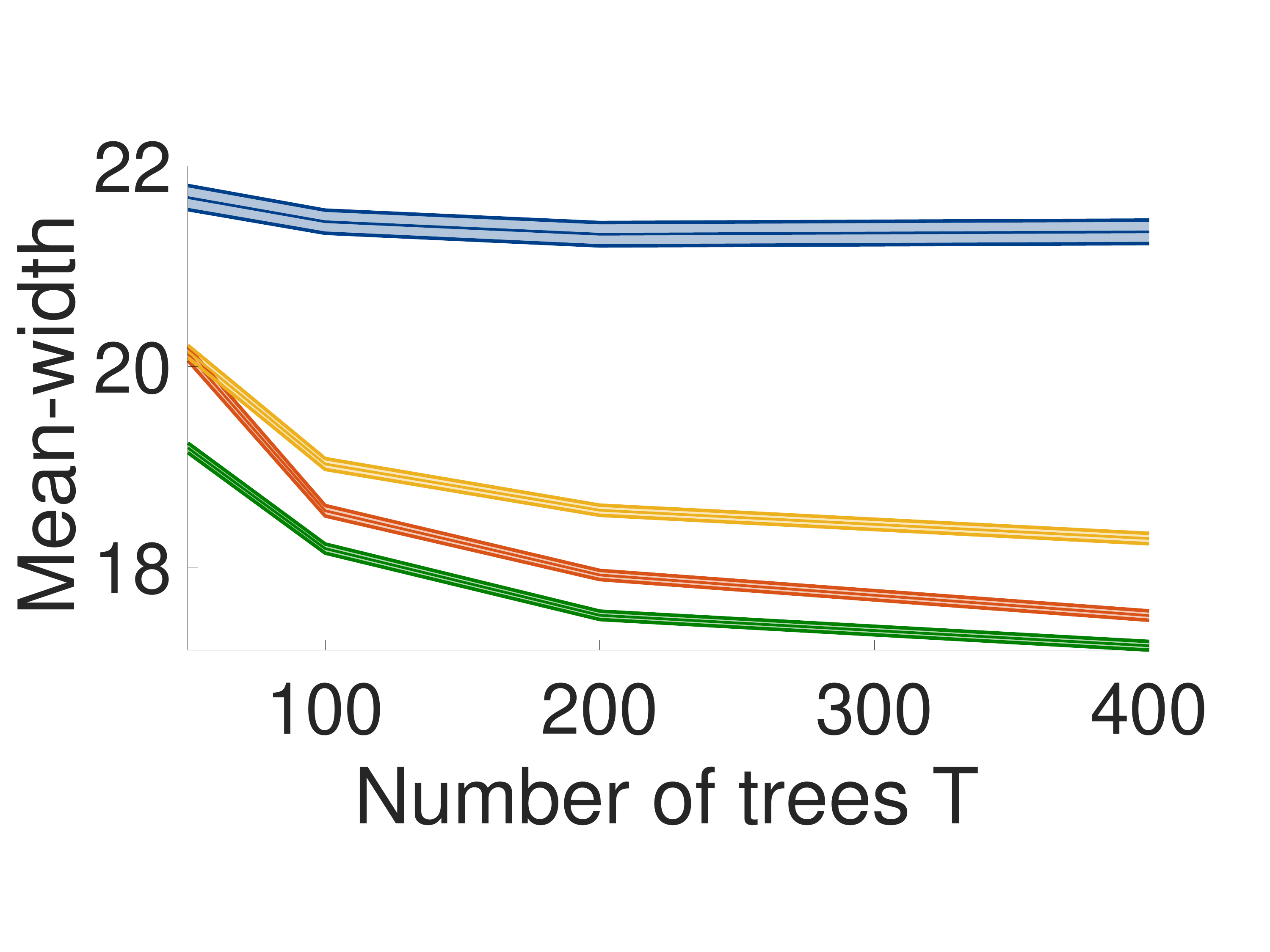}}
	\subfloat[Blog feedback.]{ \includegraphics[width=0.33\textwidth]{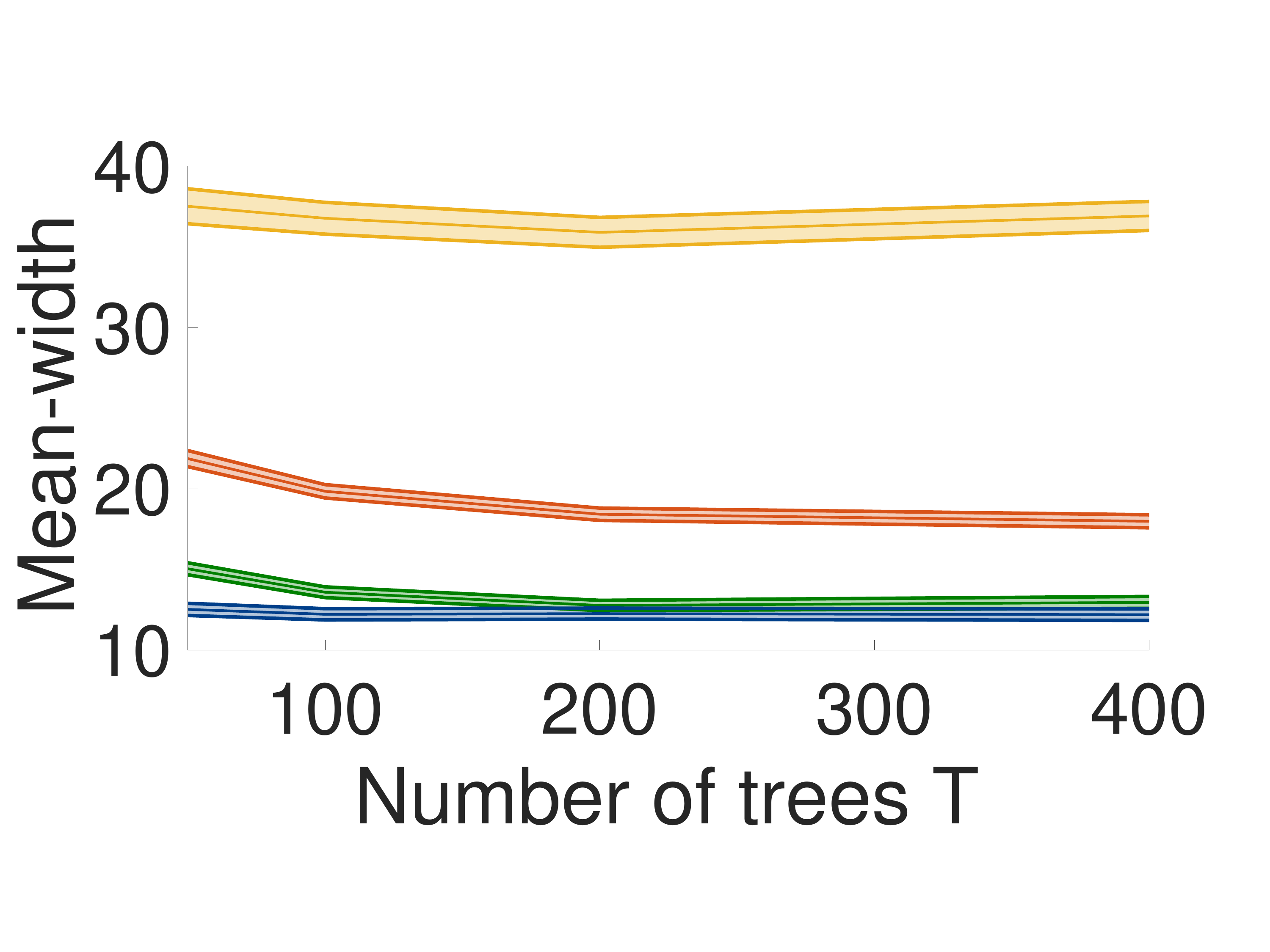}}
	\subfloat[Protein structure.]{\includegraphics[width=0.33\textwidth]{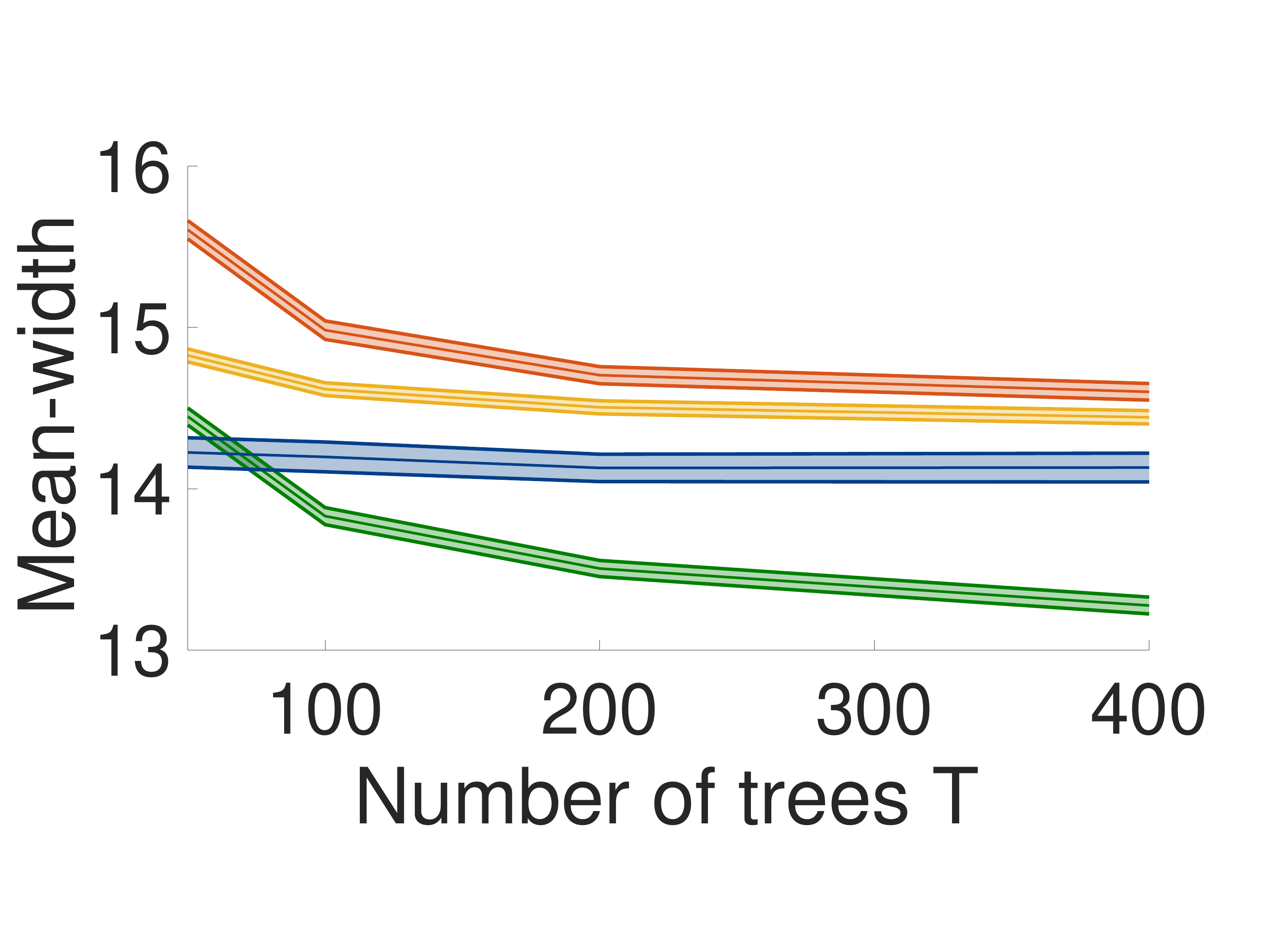}}\\
	\subfloat[Superconductivity.]{\includegraphics[width=0.33\textwidth]{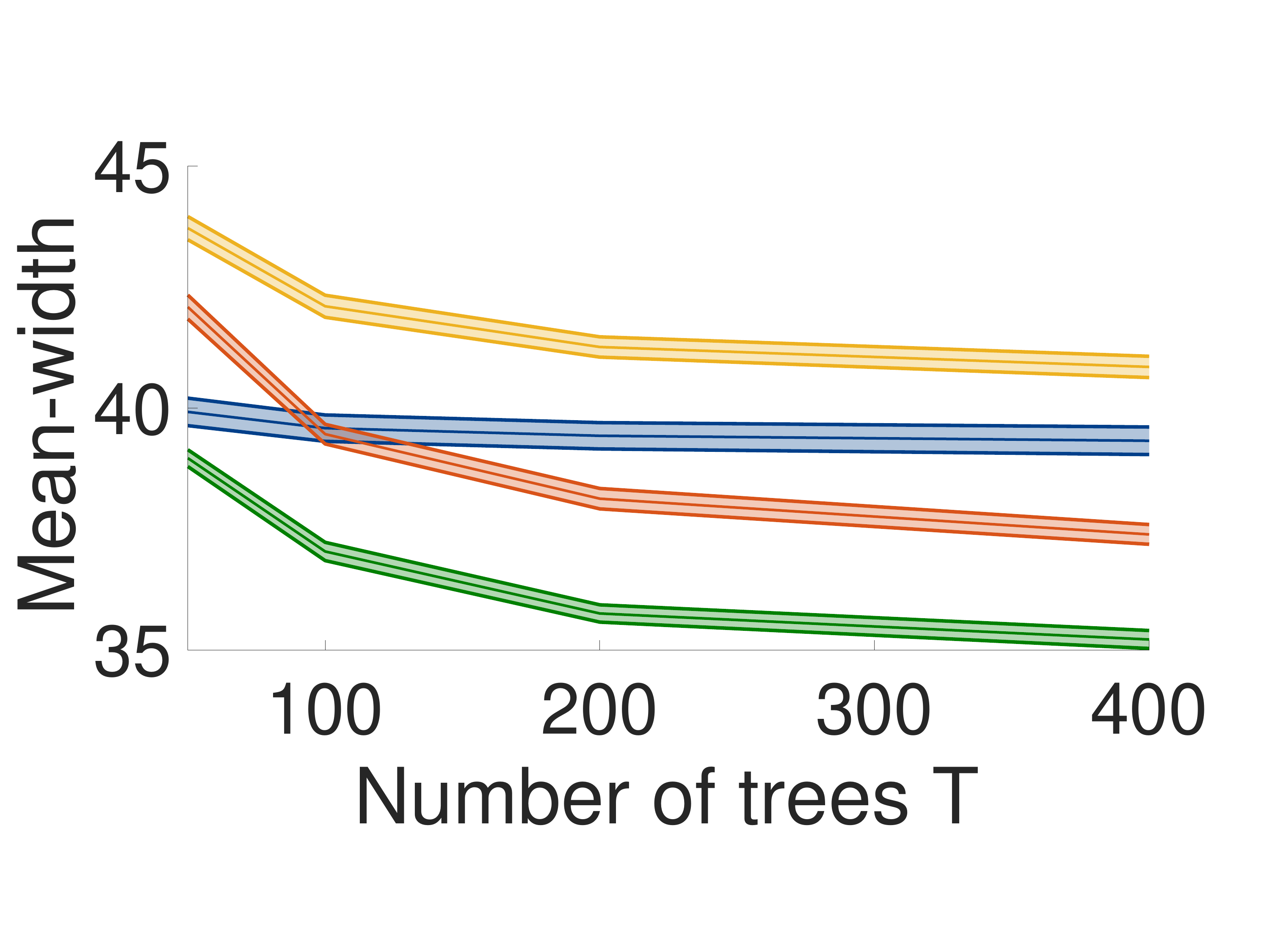}}
	\subfloat[News popularity.]{ \includegraphics[width=0.33\textwidth]{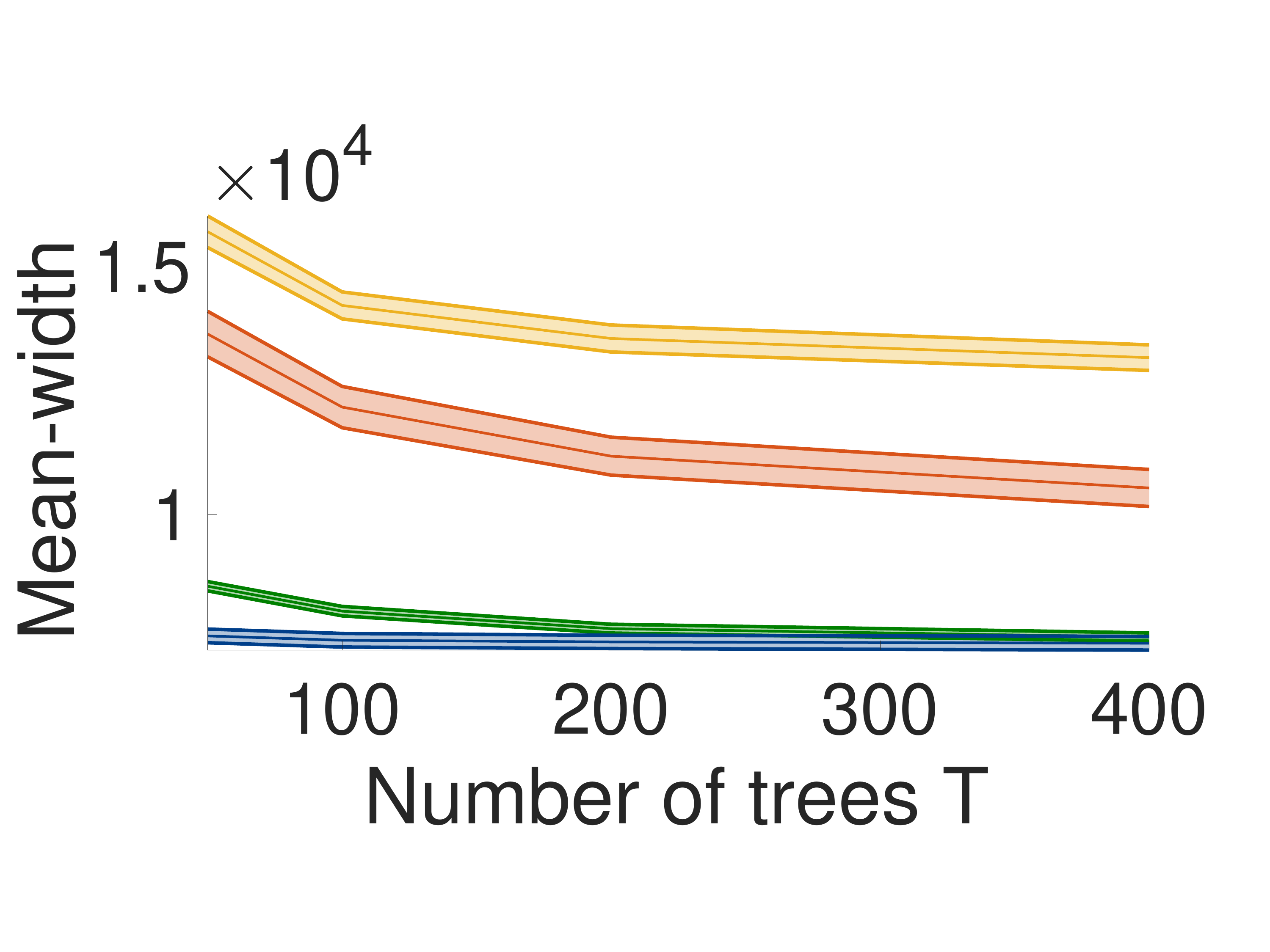}}
	\subfloat[Kernel performance.]{\includegraphics[width=0.33\textwidth]{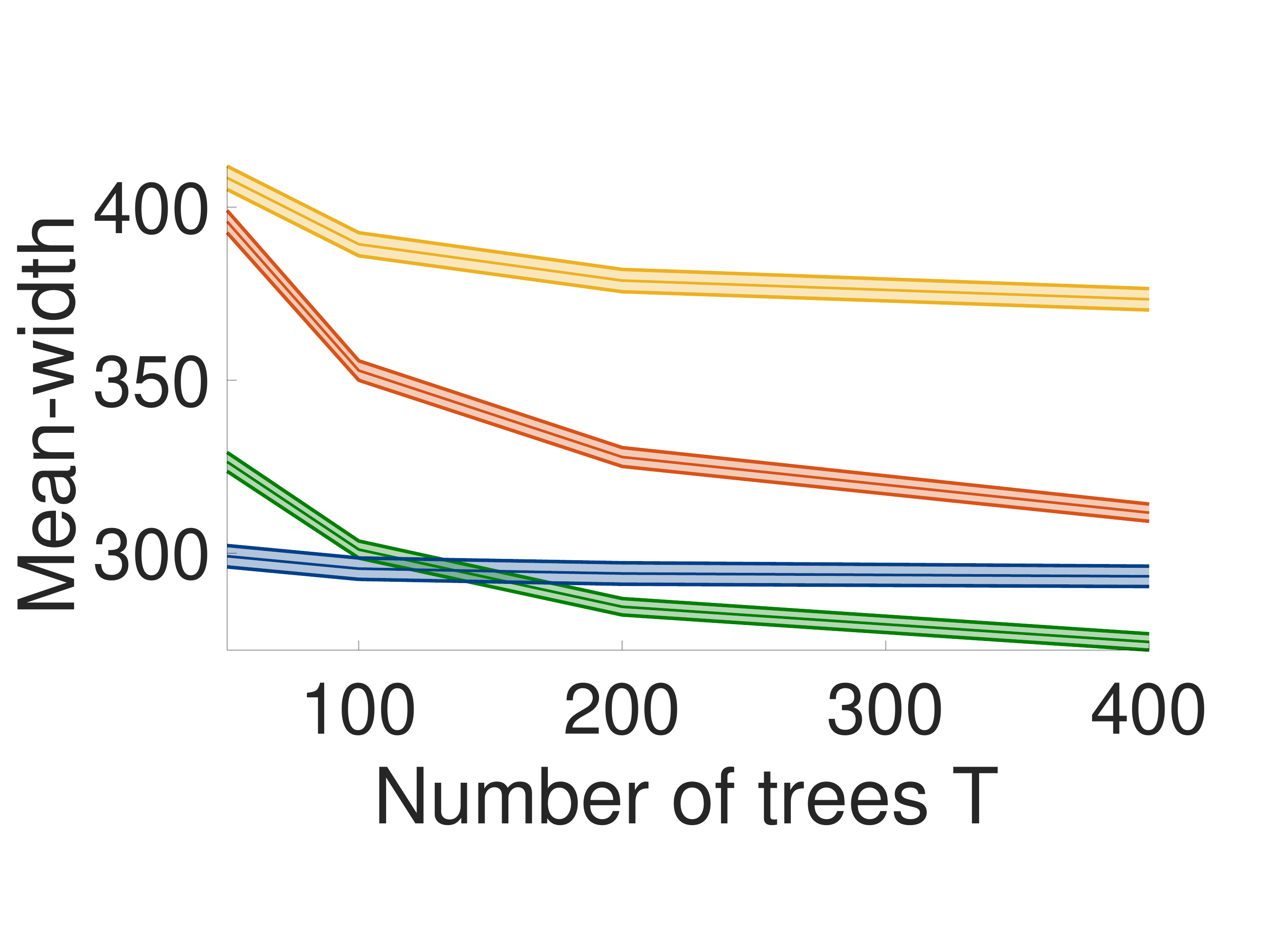}}
	\caption{The performance of QOOB ($2\alpha$) improves with increasing number of trees $T$, while the performance of Split-CQR ($2\alpha$) does not. QOOB ($2\alpha$) beats every other method except Split-CQR ($2\alpha$) for all values of $T$. %
		For the plots above, $\alpha = 0.1$ and all methods plotted have empirical mean-coverage at least $1-\alpha$. The mean-width values are averaged over 100 iterations. The shaded area denotes $\pm 1$ std-dev for the average of mean-width. }
	\label{fig:width-trees}
\end{figure}

We observe that with increasing $T$, QOOB continues to show improving performance in terms of the width of prediction intervals. Notably, this is not true for Split-CQR, which does not show improving performance beyond 100 trees. In the results reported in Table~\ref{table:overall-comparison-width}, we noted that Split-CQR-100 outperformed QOOB-100 on the blog feedback, news popularity and kernel performance datasets. 
However, from Figure~\ref{fig:width-trees} we observe that for $T = 400$, QOOB performs almost the same as Split-CQR on blog feedback and news popularity, and in fact does significantly better than Split-CQR on kernel performance. Further, QOOB shows lower values for the standard deviation of the average-mean-width. The QOOB-D method performs worse than QOOB for every dataset, and hence we did not report it in the other comparisons in this paper. %

\subsection{QOOB outperforms Split-CQR at small sample sizes}\label{subsec:QOOB-sample-size}
QOOB needs $n$ times more computation than Split-CQR to produce prediction intervals, since one needs to make $n$ individual predictions. If fast prediction time is desired, our experiments in Sections~\ref{subsec:nominal-quantile} and \ref{subsec:QOOB-trees} indicate that Split-CQR is a competitive quick alternative. However, here we demonstrate that at the small sample regime, QOOB significantly outperforms Split-CQR on all 6 datasets that we have considered.

\begin{figure}[!t]
	\centering
	\includegraphics[width=0.9\textwidth]{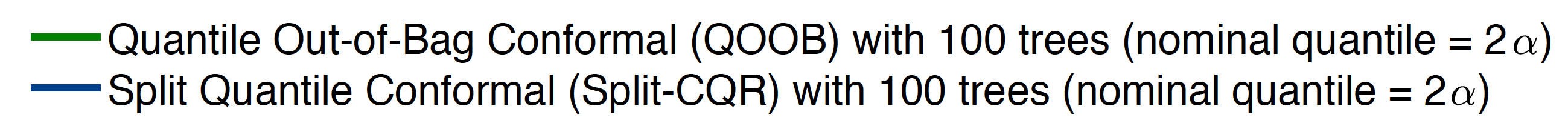}\vspace{-0.6cm}
	\centering
	\subfloat[Concrete structure.]{\includegraphics[width=0.33\textwidth]{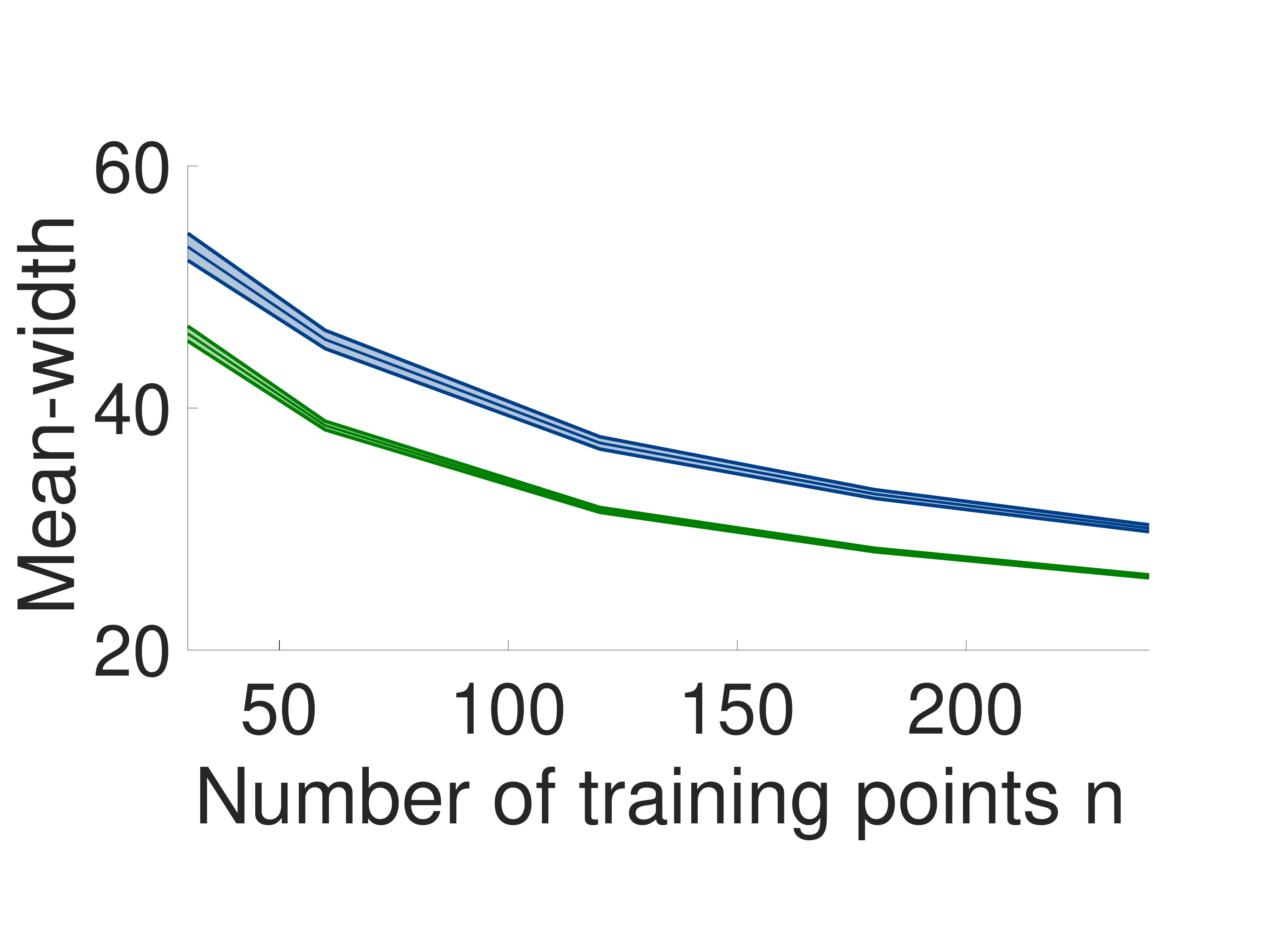}}
	\subfloat[Blog feedback.]{ \includegraphics[width=0.33\textwidth]{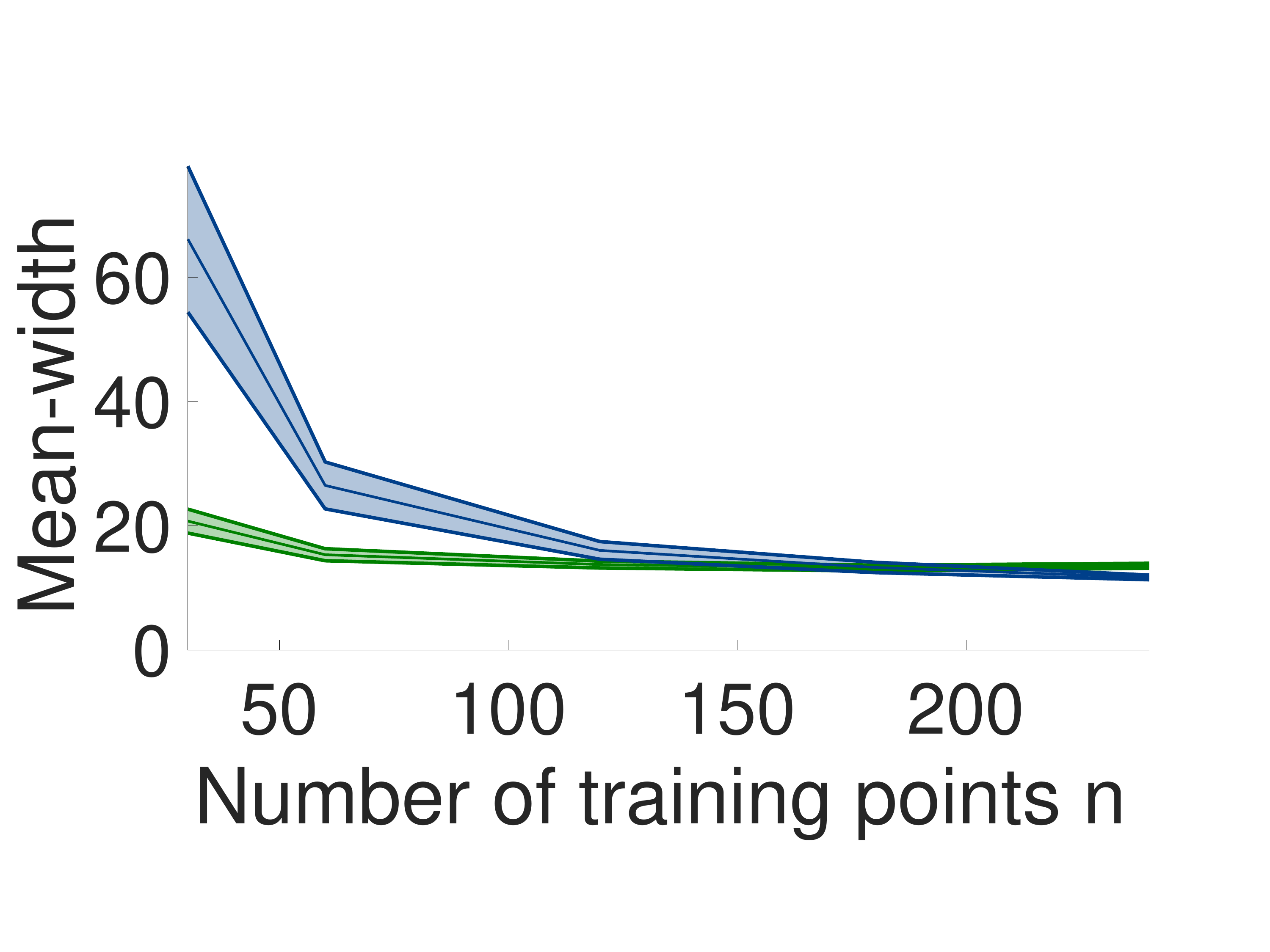}}
	\subfloat[Protein structure.]{\includegraphics[width=0.33\textwidth]{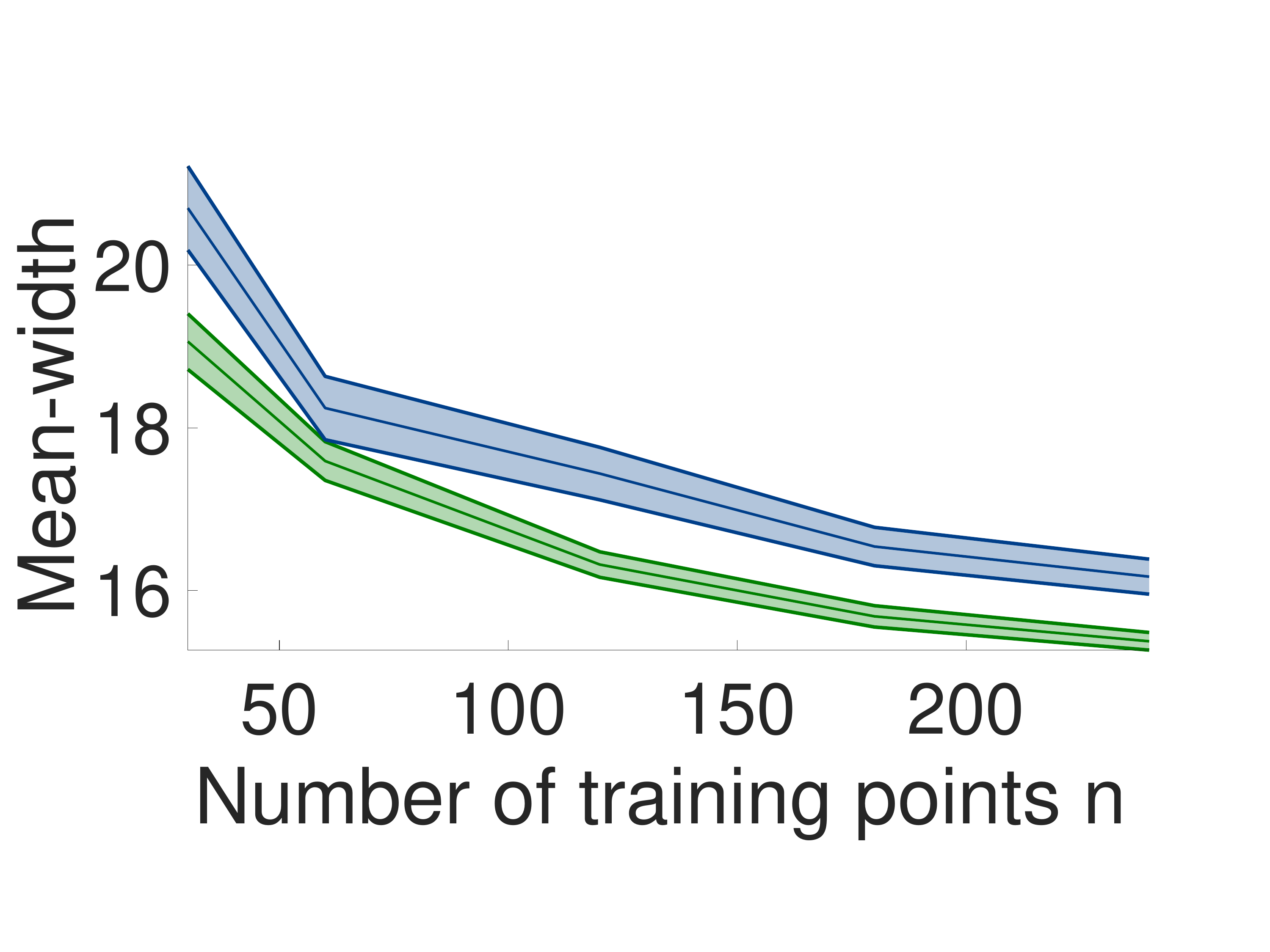}}\\
	\subfloat[Superconductivity.]{\includegraphics[width=0.33\textwidth]{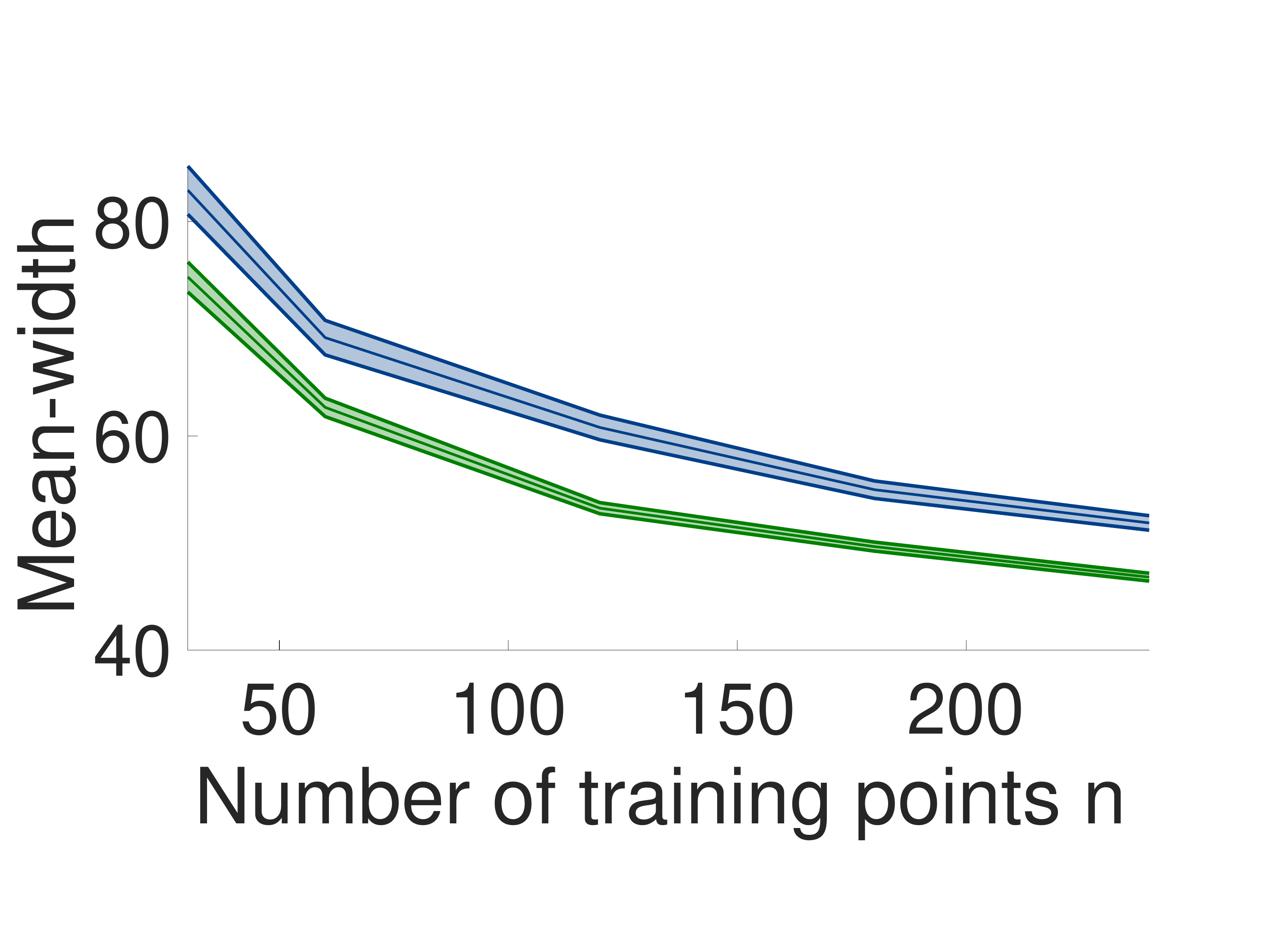}}
	\subfloat[News popularity.]{ \includegraphics[width=0.33\textwidth]{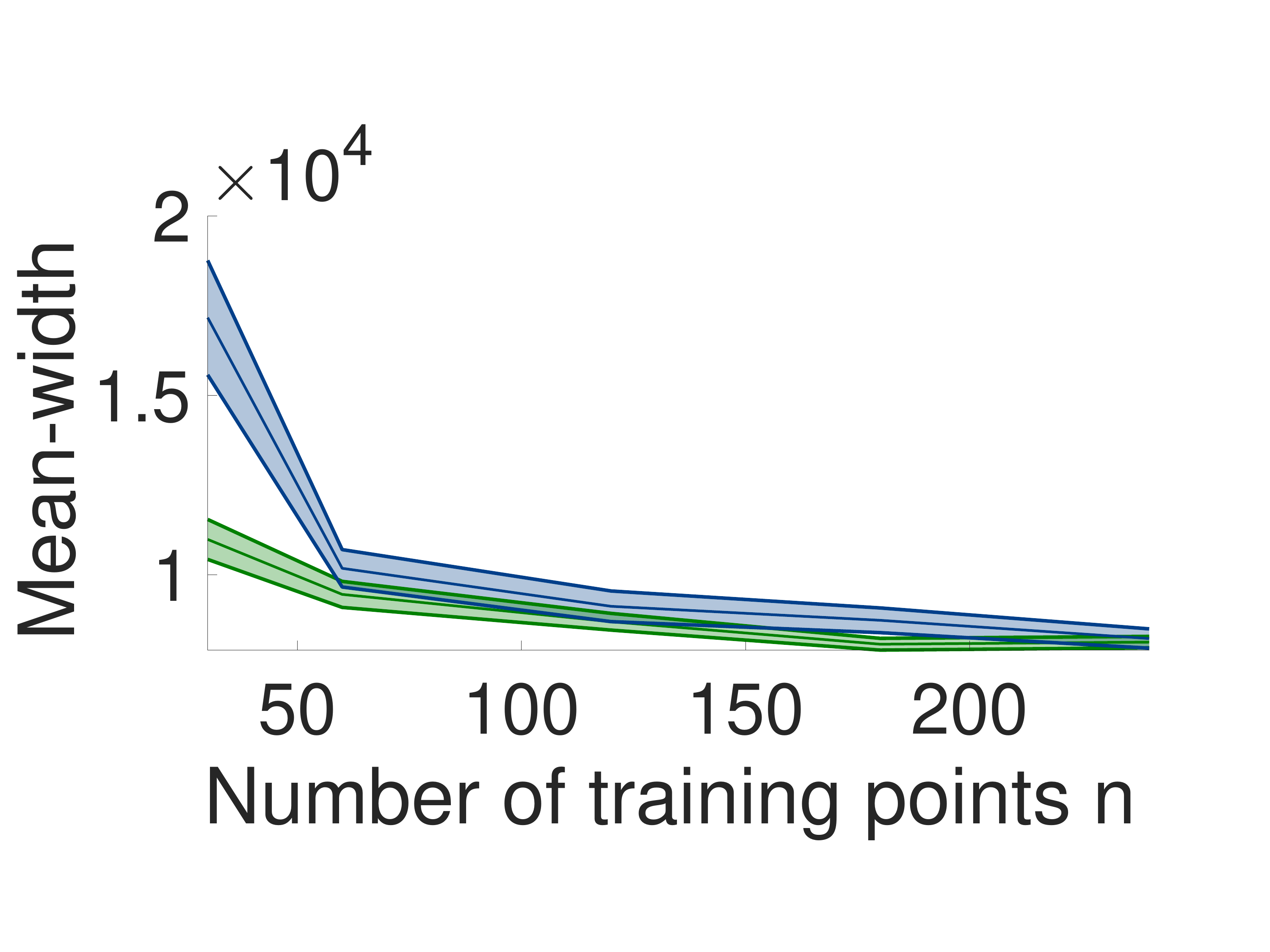}}
	\subfloat[Kernel performance.]{\includegraphics[width=0.33\textwidth]{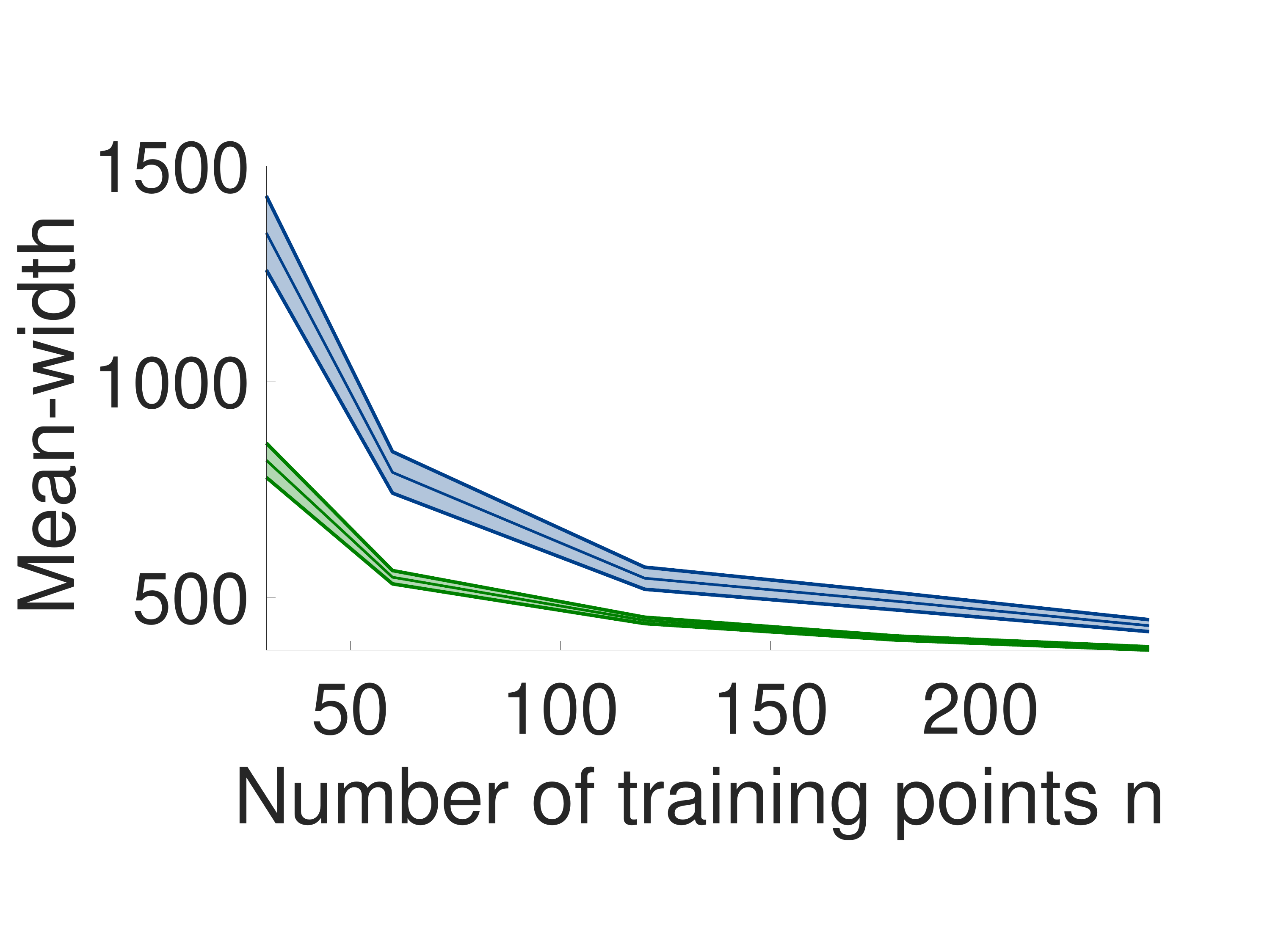}}
	\caption{The performance of QOOB and Split-CQR with varying number of training points $n$. QOOB has shorter mean-width than Split-CQR across datasets for small $n$ and also smaller standard-deviation of the average mean-width. For the plots above, $\alpha = 0.1$ and all methods plotted have empirical mean-coverage at least $1-\alpha$. The mean-width values are averaged over 100 iterations. The shaded area denotes $\pm 1$ std-dev for the average of mean-width. }
	\label{fig:width-n}
\end{figure}
To make this comparison, we subsample the datasets to a smaller sample size and consider the mean-width and mean-coverage properties of QOOB ($2\alpha$) and Split-CQR ($2\alpha$) with $T = 100$. Figure~\ref{fig:width-n} contains the results with $n$ ranging from $30$ to $240$. We observe that at small $n$, QOOB does significantly better than Split-CQR. This behavior is expected since at smaller values of $n$, the statistical loss due to sample splitting is most pronounced. Since the overall computation time decreases as $n$ decreases, QOOB is a significantly better alternative in the small sample regime on all fronts.

\subsection{Cross-conformal outperforms jackknife+}
\label{subsec:cc-jp}
Cross-conformal prediction sets are always smaller than the corresponding jackknife+ prediction sets by construction; see Section~\ref{sec:intervals-LOO} and~\eqref{eq:loo-jp}. %
However, the fact that cross-conformal may not give an interval might be of practical importance. In this subsection, we show that the jackknife+ prediction interval can sometimes be strictly larger than the smallest interval containing the cross-conformal prediction set (this is the convex hull of the cross-conformal prediction set and we call it QOOB-Conv). 

 Table~\ref{table:cc-jp} reports the performance of QOOB, QOOB-JP, and QOOB-Conv on the blog feedback dataset. Here QOOB-JP refers to the OOB-JP version~\eqref{eq:def-oob-jp}. %
 For each of these, we set the nominal quantile level to $0.5$ instead of $2\alpha$ as suggested earlier (this led to the most pronounced difference in mean-widths). 
\begin{table}[H]
	\caption{Mean-width of $C^{\texttt{OOB}}(x), \mathrm{Conv}(C^{\texttt{OOB}}(x)),$ and $C^{\texttt{OOB-JP}}(x)$ for the blog feedback dataset with QOOB method. The base quantile estimator is quantile regression forests, and $\alpha = 0.1$. Average values across 100 simulations are reported with the standard deviation in brackets . }
	\centering
	\begin{tabular}{c|c|c}
		\hline
		\textbf{Method} &\textbf{Mean-width} & \textbf{Mean-coverage} \\ \hline\hline
		QOOB-100 ($\beta = $0.5) & \textbf{14.67 (0.246)} & 0.908  (0.002)  \\
		\hline
		QOOB-Conv-100 ($\beta = $0.5) & 14.73 (0.249) & 0.908 (0.002)  \\
		\hline
		QOOB-JP-100 ($\beta = $0.5) & 15.36 (0.248) & 0.911 (0.002)  \\
		\hline
	\end{tabular}
	\label{table:cc-jp}
\end{table}

\begin{figure}[!t]
  \includegraphics[width=0.16\textwidth]{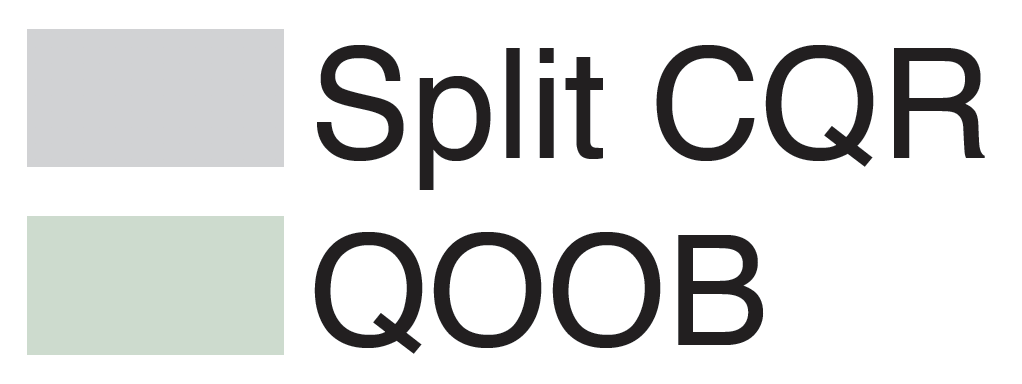}\hspace{2.6cm}\includegraphics[width=0.14\textwidth]{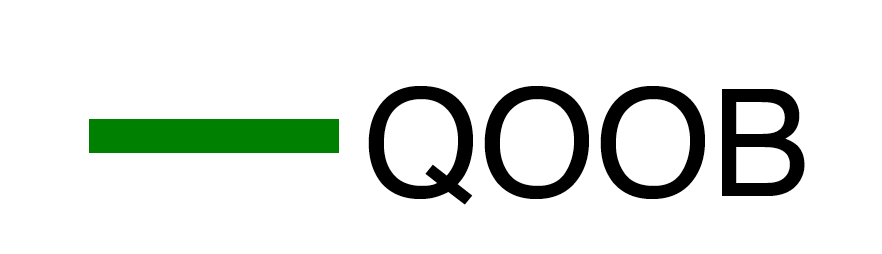}\hspace{2.5cm}\includegraphics[width=0.16\textwidth]{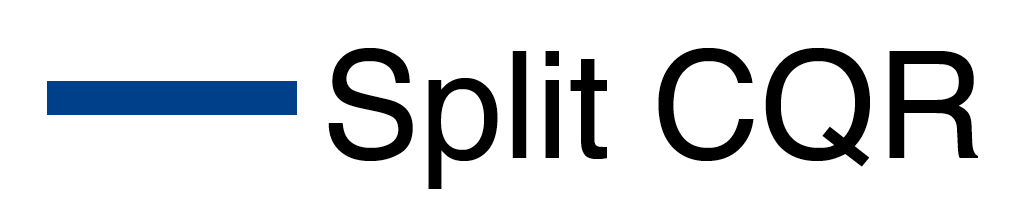}
  \vspace{-1cm}
  \centering
  \subfloat[$n = 100$. (QOOB) MW = 2.16, MC = 0.91. (Split-CQR) MW = 2.23, MC = 0.92.]{\includegraphics[width=0.33\textwidth]{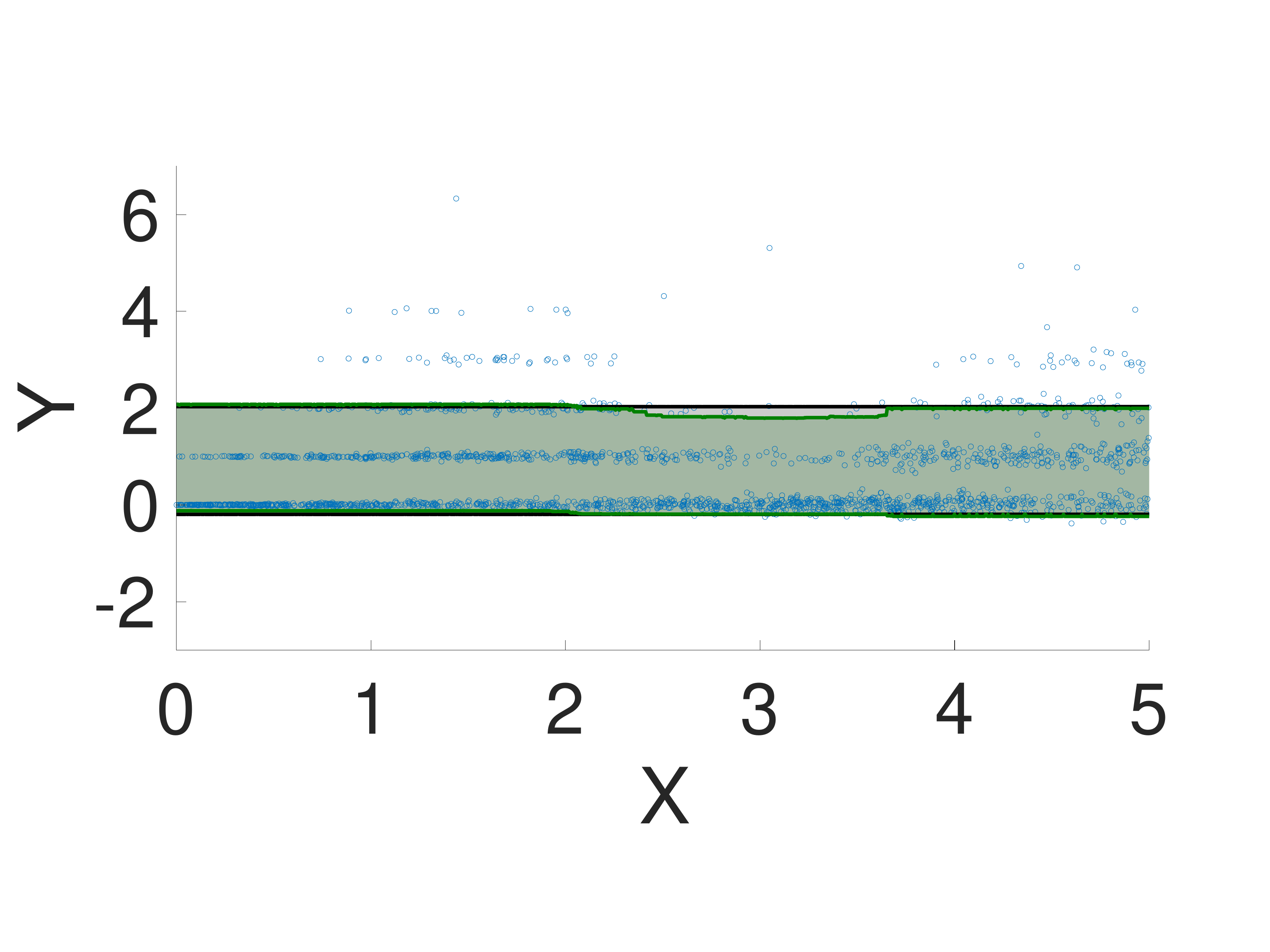}\includegraphics[width=0.33\textwidth]{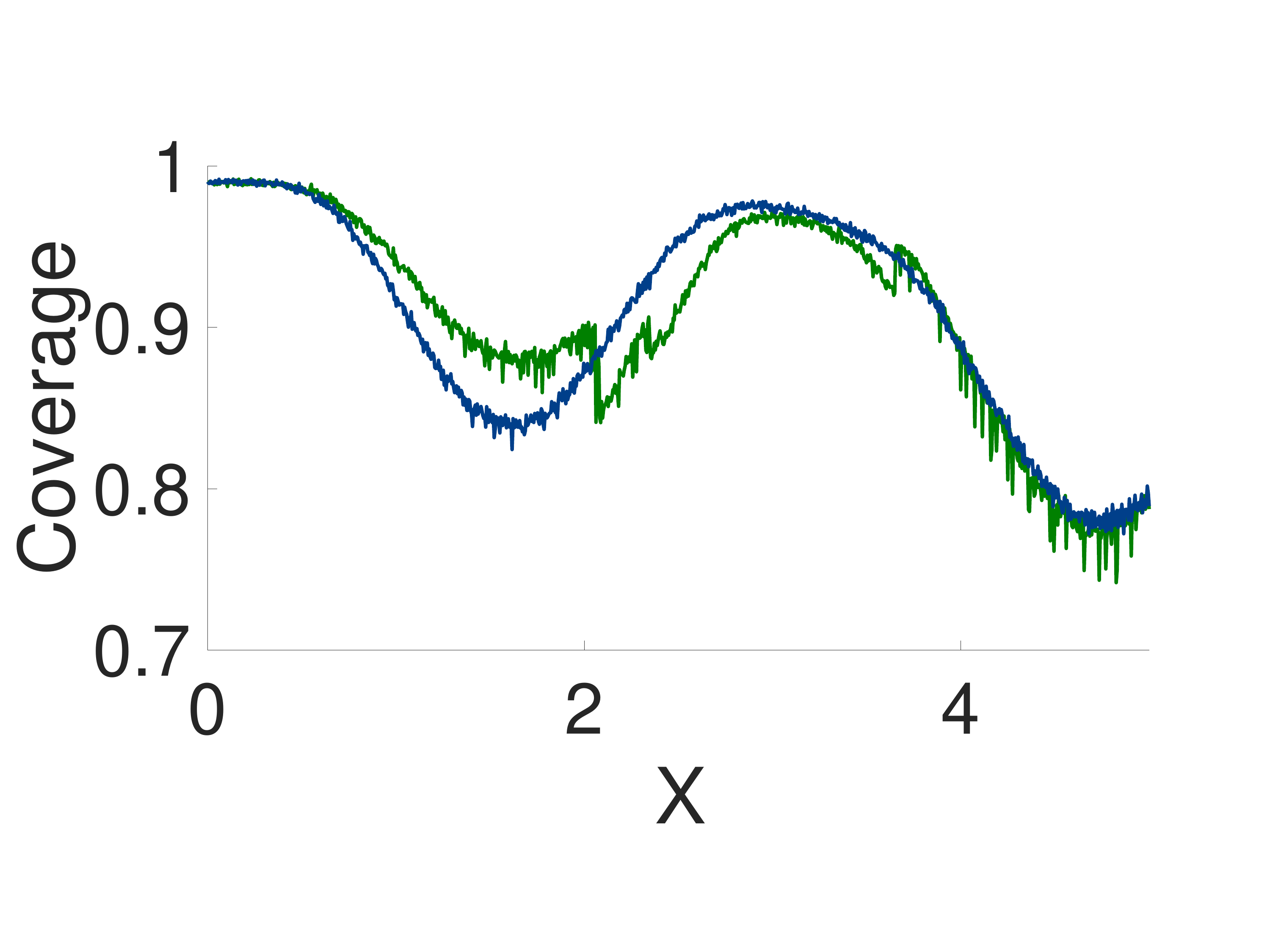}\includegraphics[width=0.33\textwidth]{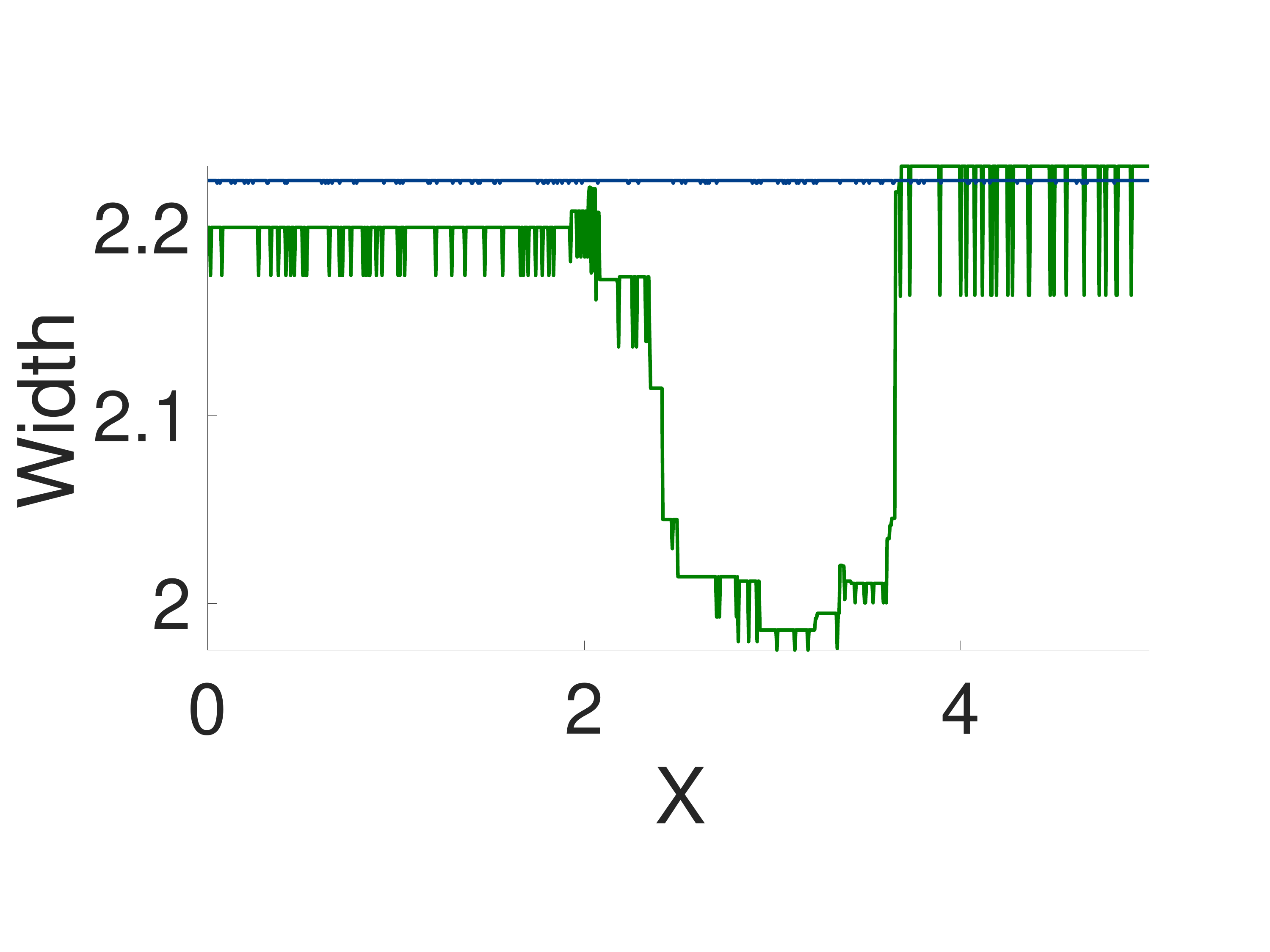}}\\\vspace{-0.5cm}
  \subfloat[$n = 200$. (QOOB) MW = 1.99, MC = 0.92. (Split-CQR) MW = 2.18, MC = 0.91.]{\includegraphics[width=0.33\textwidth]{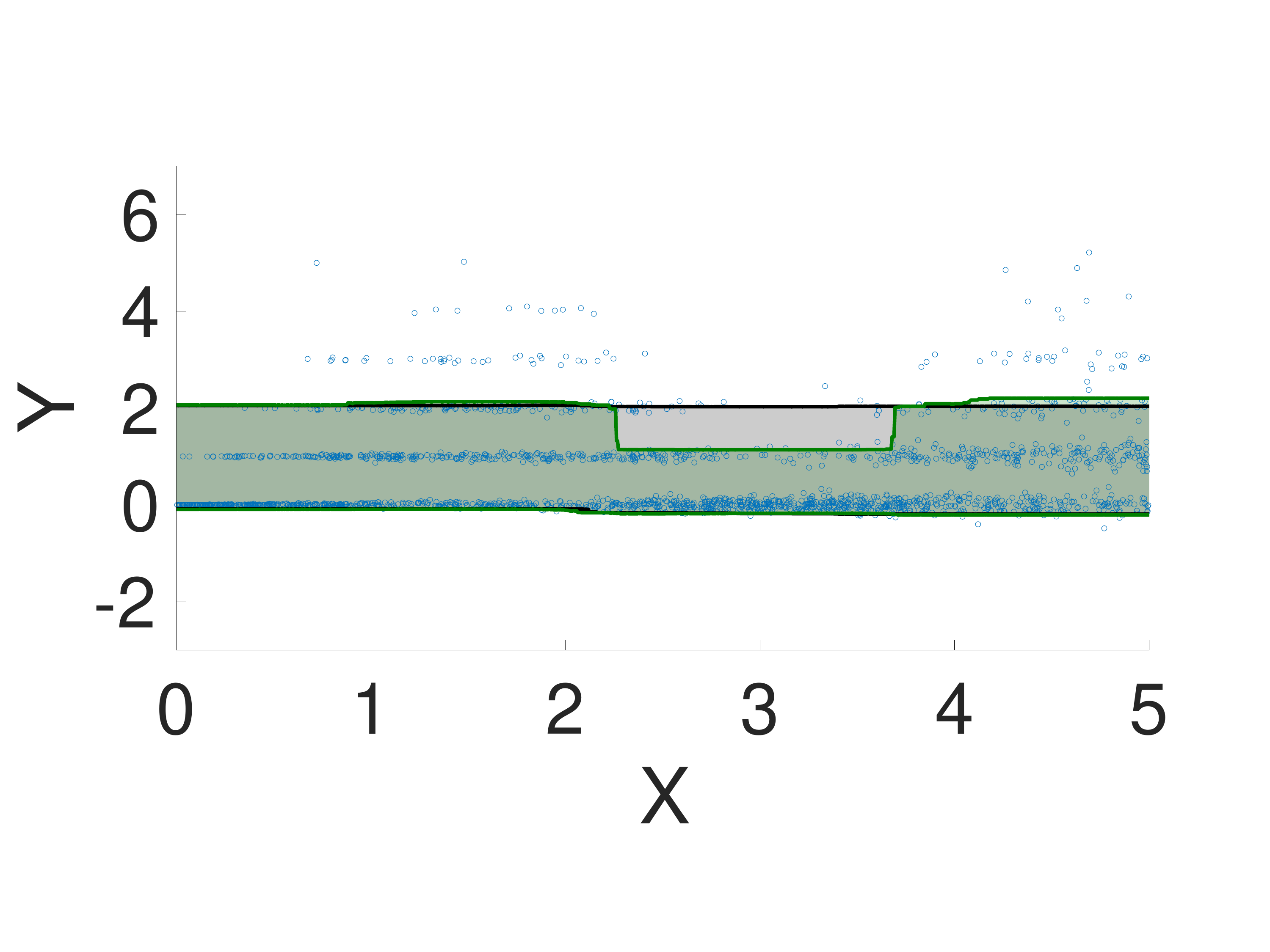}\includegraphics[width=0.33\textwidth]{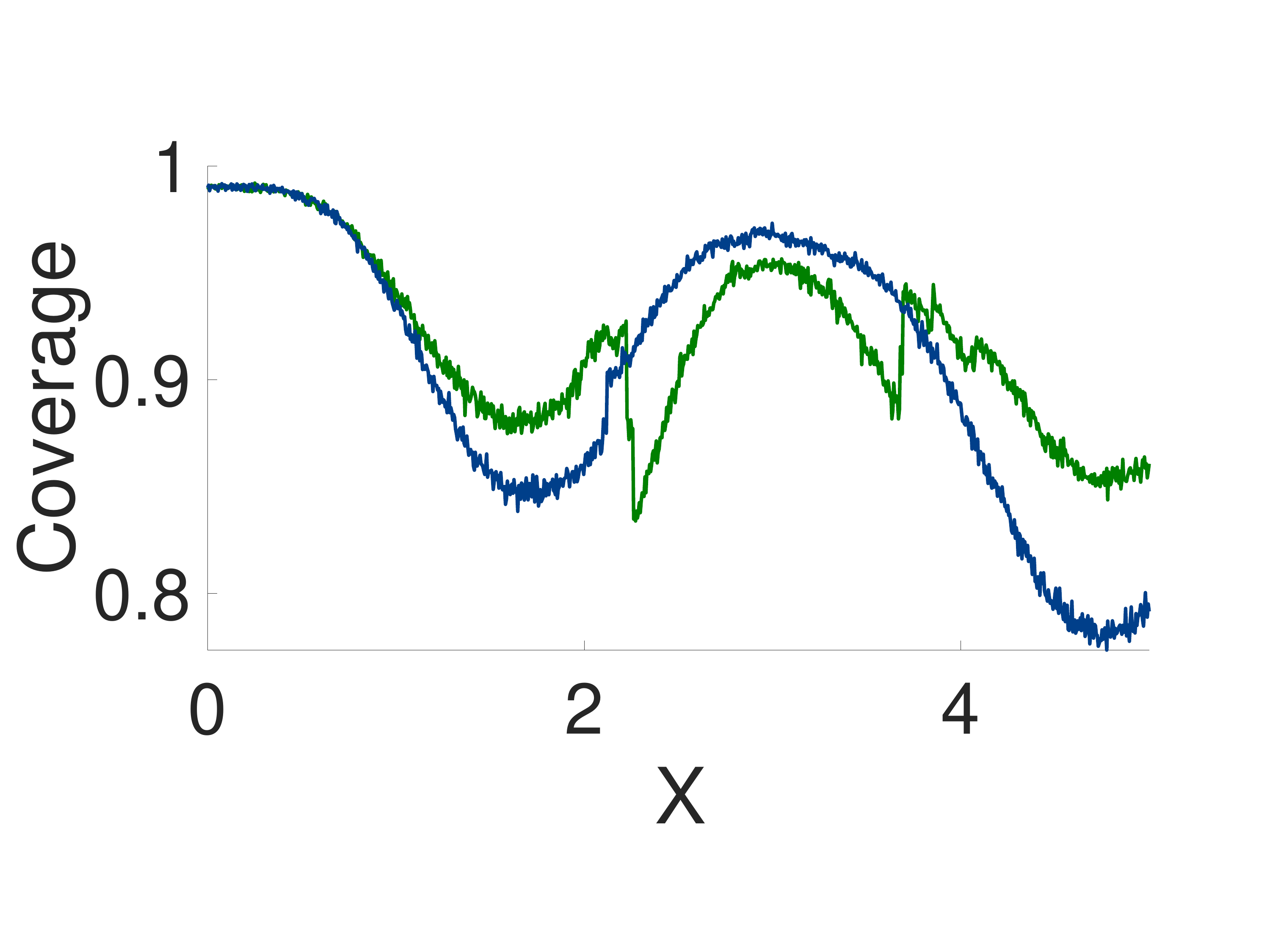}\includegraphics[width=0.33\textwidth]{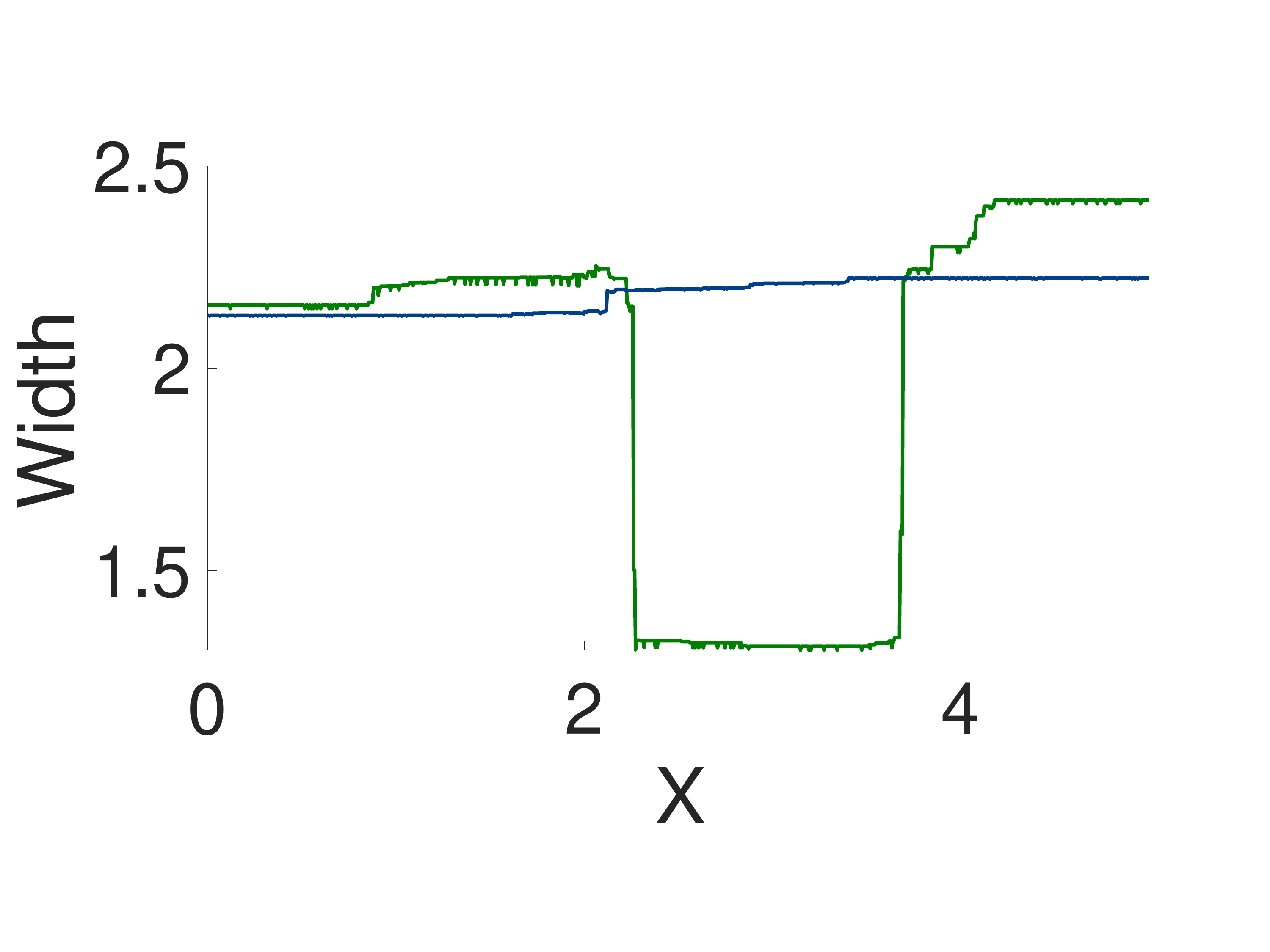}}\\\vspace{-0.5cm}
  \subfloat[$n = 300$. (QOOB) MW = 1.86, MC = 0.89. (Split-CQR) MW = 1.94, MC = 0.89.]{\includegraphics[width=0.33\textwidth]{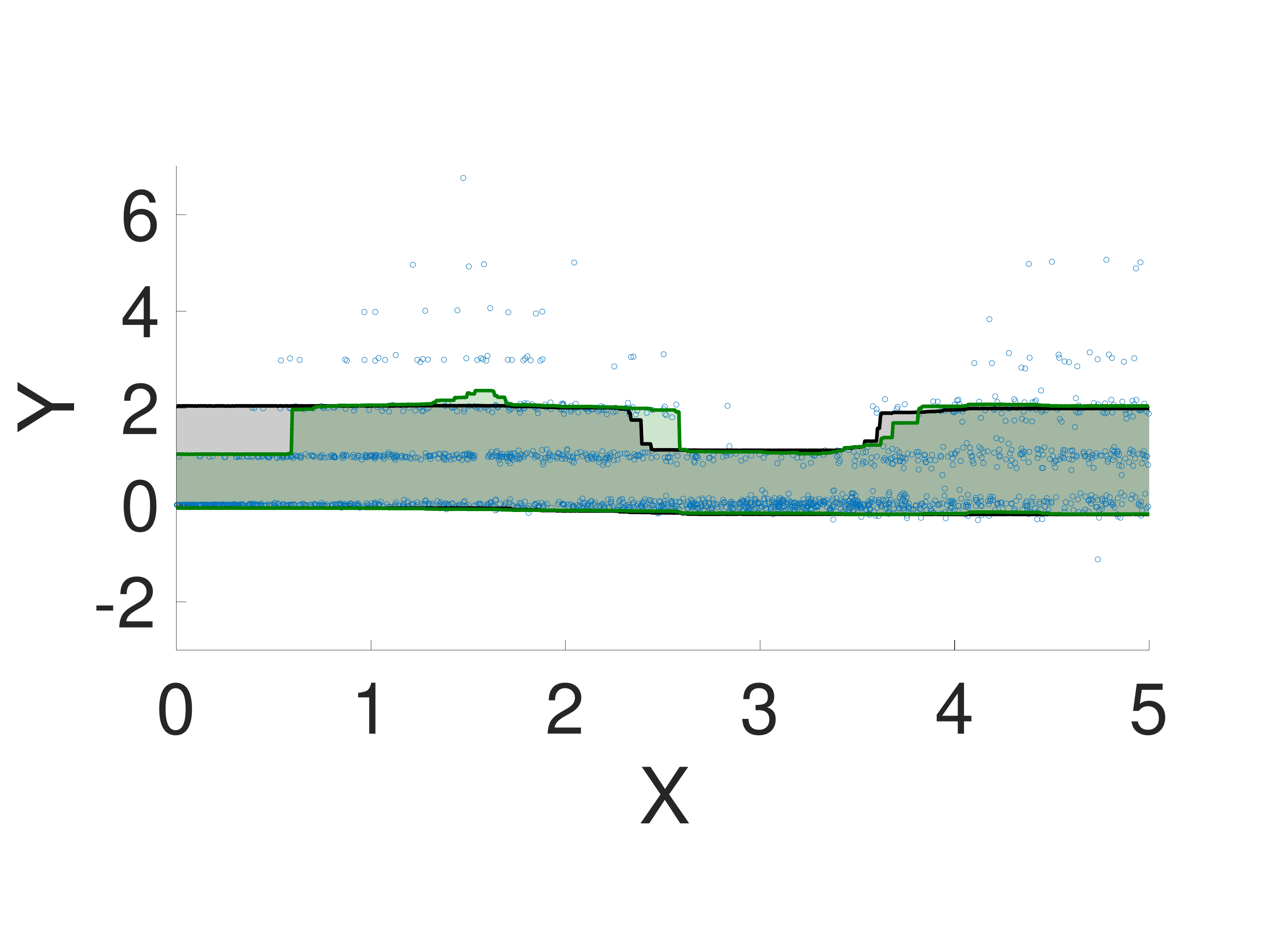}\includegraphics[width=0.33\textwidth]{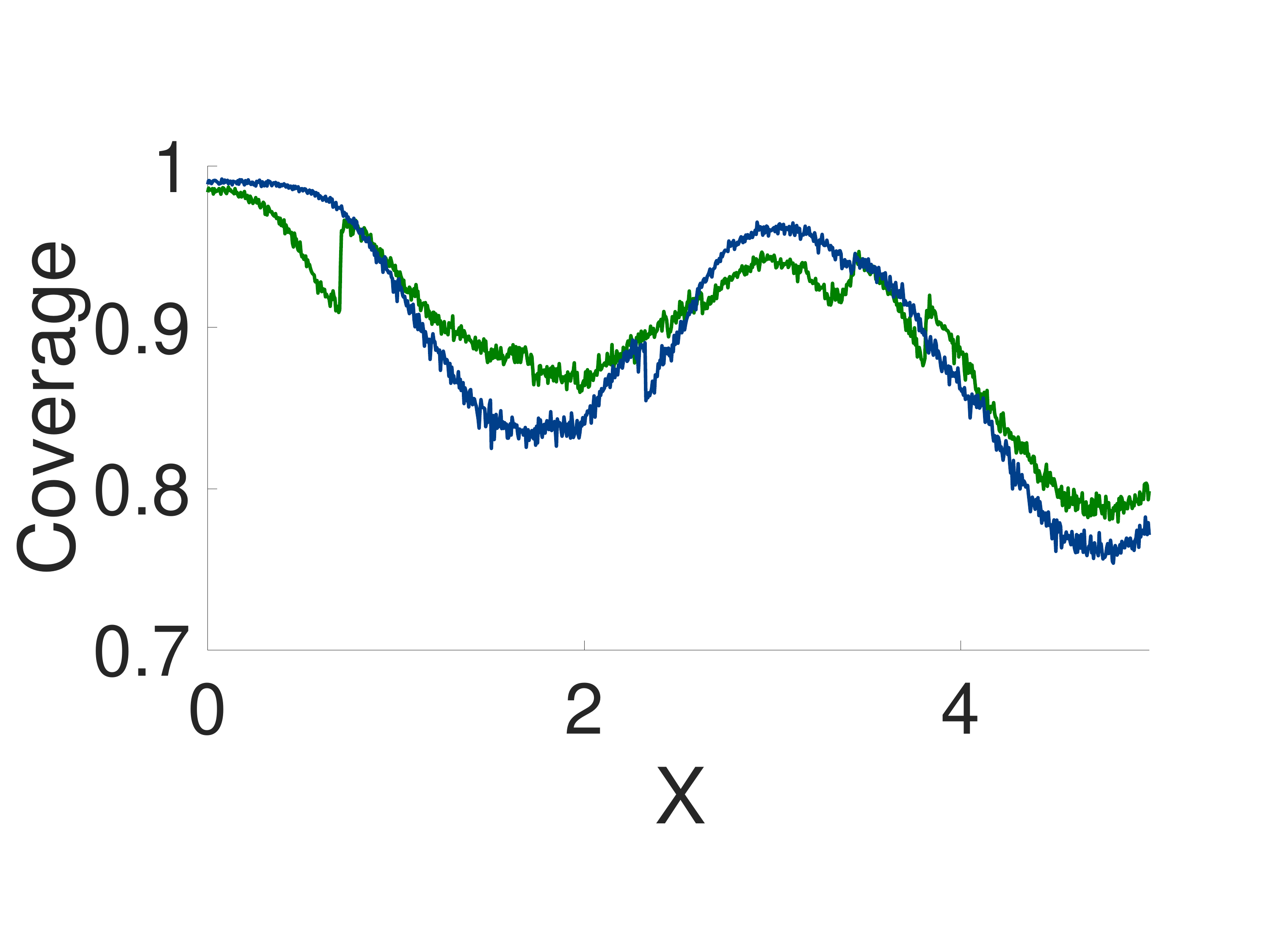}\includegraphics[width=0.33\textwidth]{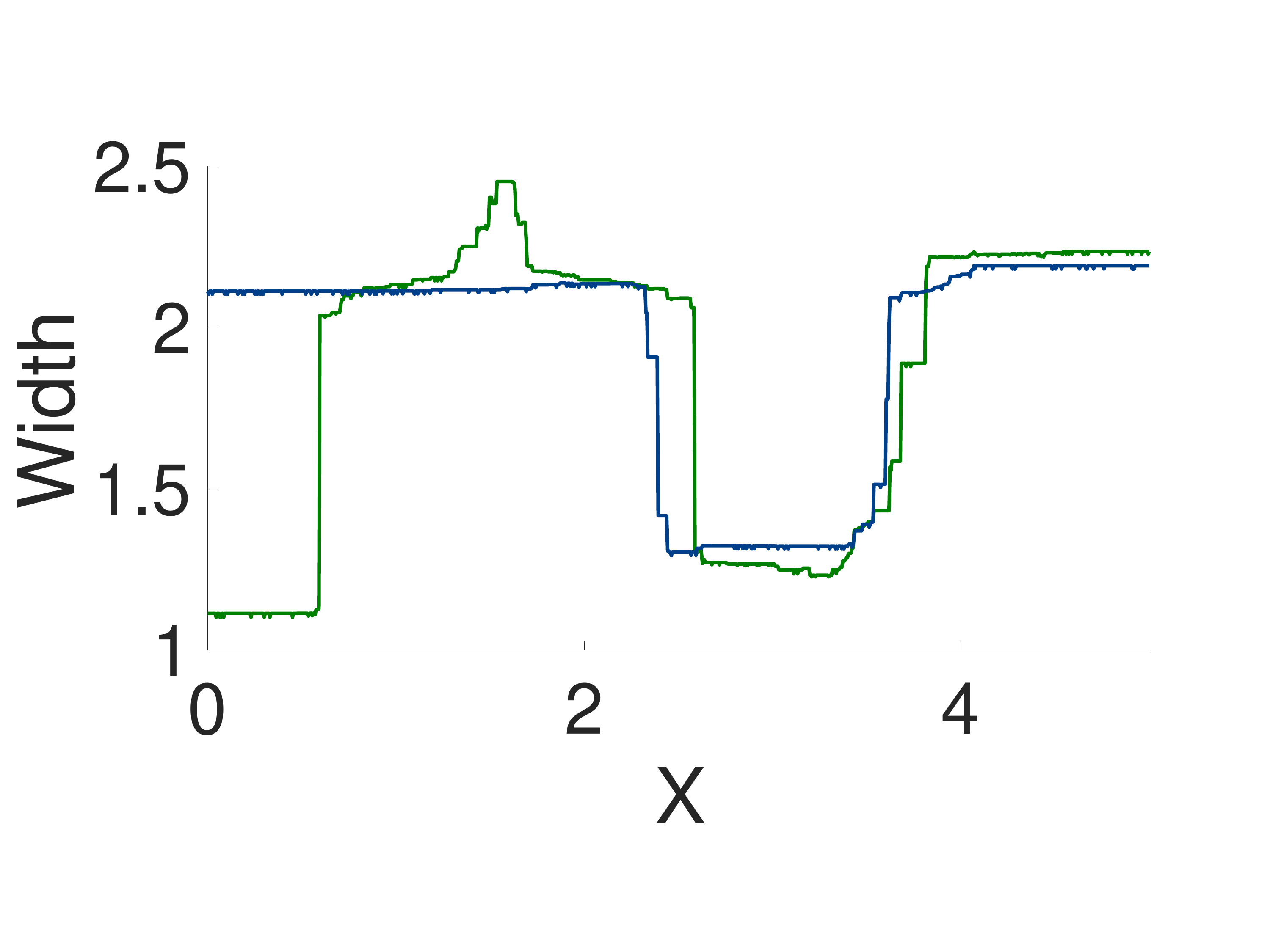}}\\\vspace{-0.5cm}
  \subfloat[$n = 400$. (QOOB) MW = 2.10, MC = 0.91. (Split-CQR) MW = 2.09, MC = 0.90.]{\includegraphics[width=0.33\textwidth]{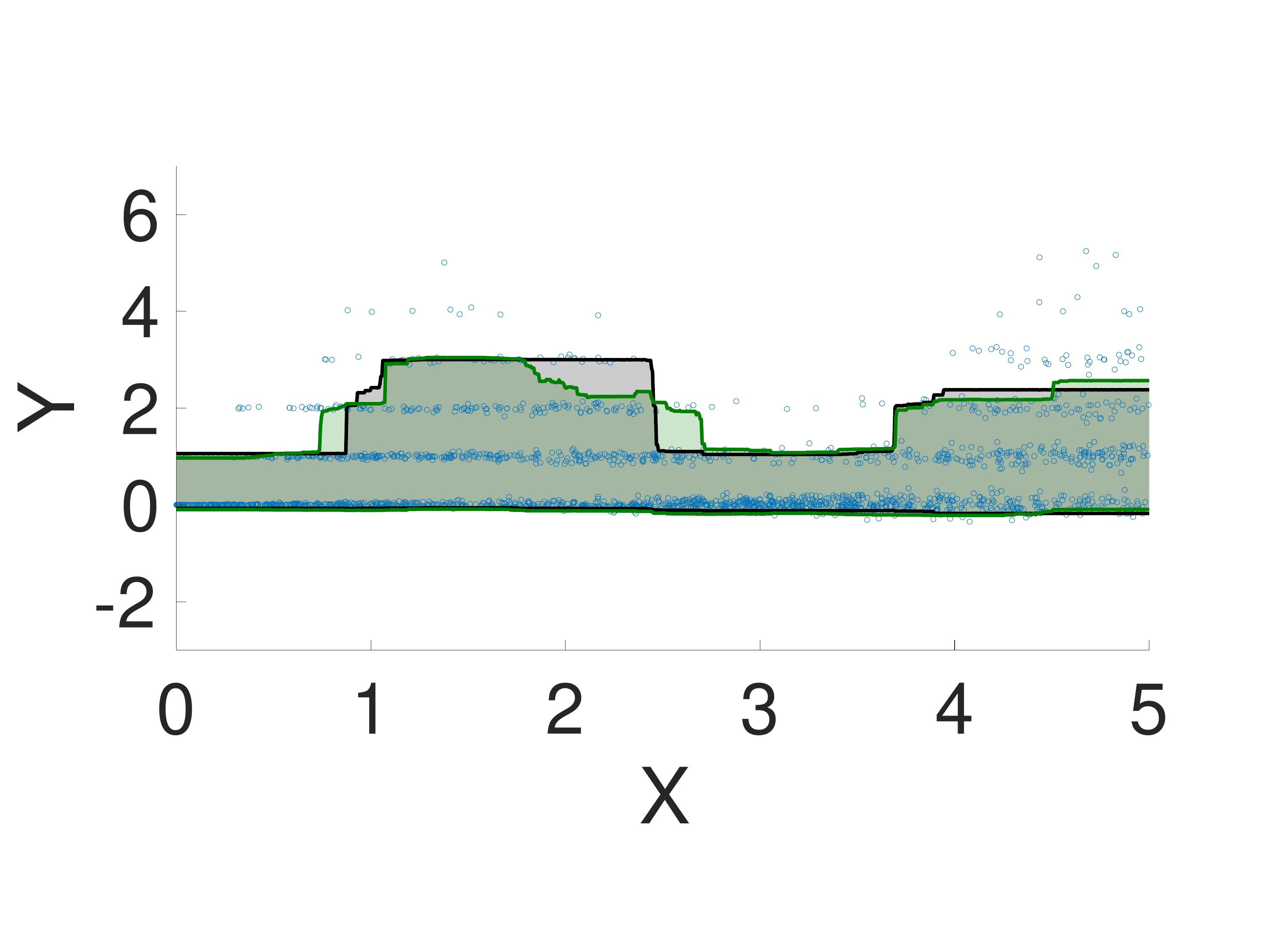}\includegraphics[width=0.33\textwidth]{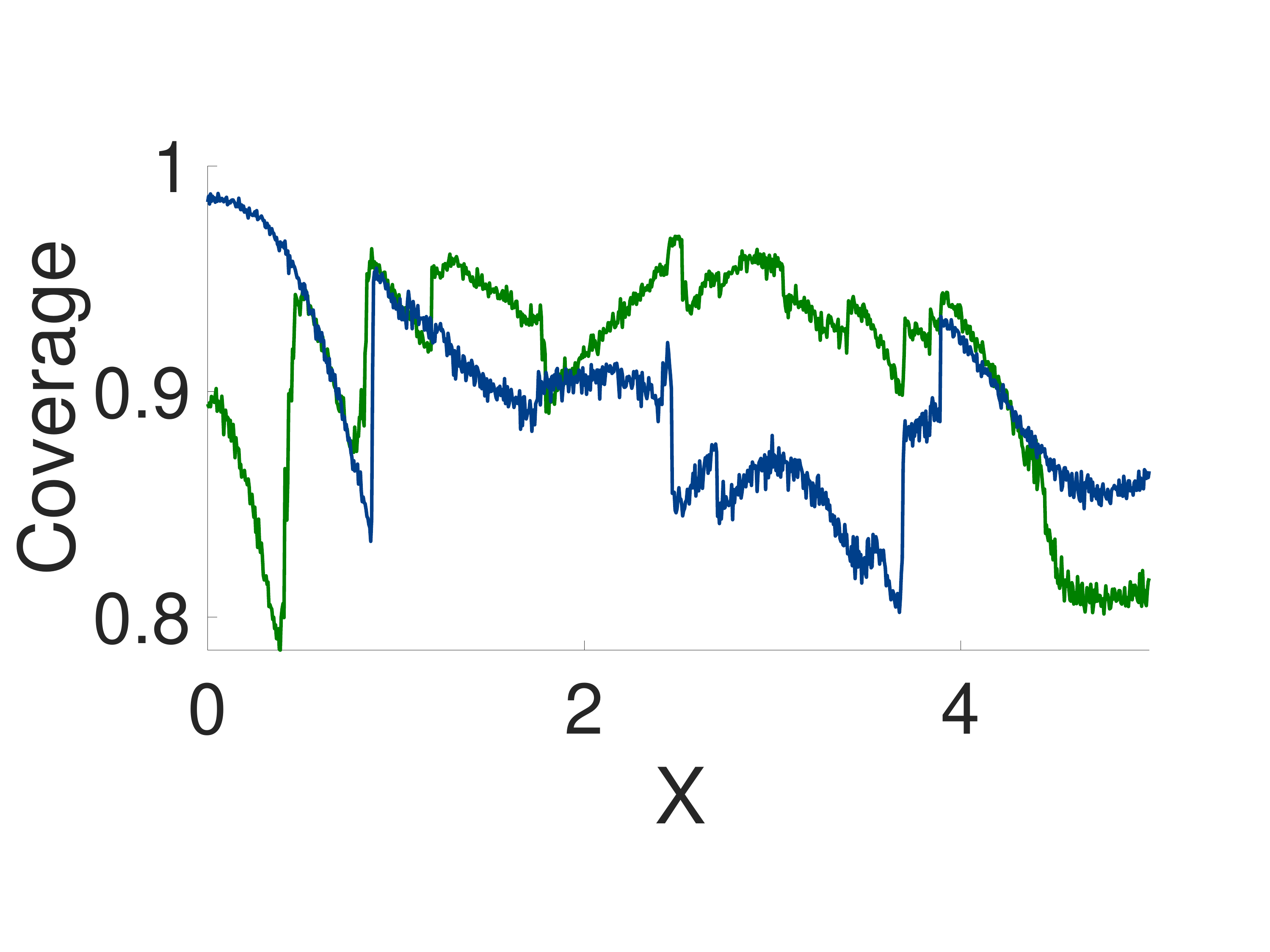}\includegraphics[width=0.33\textwidth]{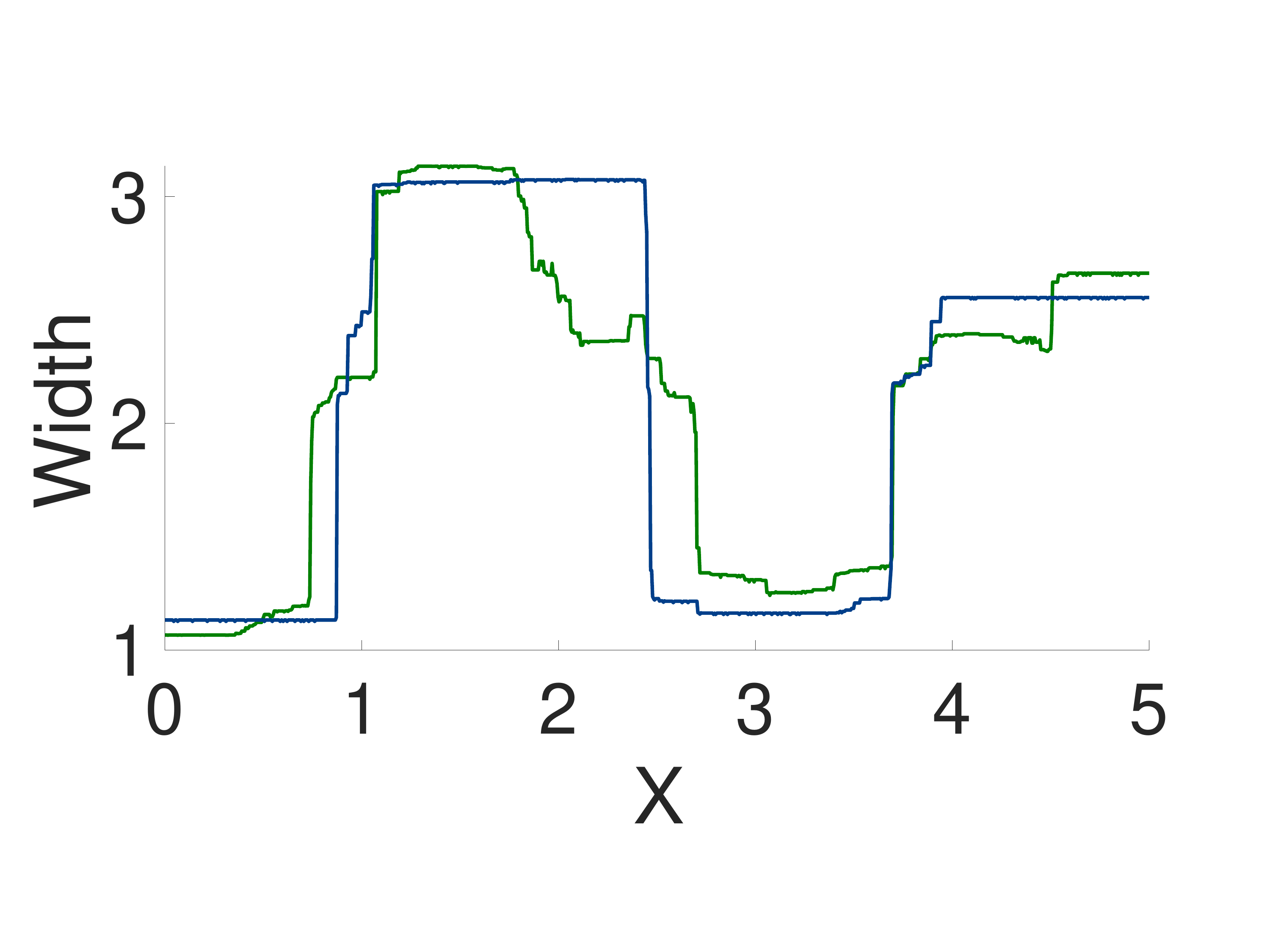}}

  \caption{The performance of QOOB-100 and Split-CQR-100 on synthetic data with varying number of training points $n$ ($\alpha = 0.1$). MW refers to mean-width and MC refers to mean-coverage. QOOB shows conditional coverage at smaller values of $n$ than Split-CQR. Section~\ref{subsec:synthetic} contains more experimental details. } 
  \label{fig:synthetic}
\end{figure}

While this is a specific setting, our goal is to provide a proof of existence. In other settings, cross-conformal style aggregation and jackknife+ style aggregation may have identical prediction sets. However, because the cross-conformal prediction set as well as its convex hull can be computed in nearly the same time (see Section~\ref{subsec:compuation-cross}) and have the same marginal validity guarantee, one should always prefer cross-conformal over jackknife+.

\subsection{QOOB demonstrates conditional coverage empirically }
\label{subsec:synthetic}
To demonstrate that Split-CQR exhibits conditional coverage, \citet[Appendix B]{romano2019conformalized} designed the following data-generating distribution for $P_{XY}$:
\begin{align*}
	\epsilon_1 &\sim N(0, 1), 
	\epsilon_2 \sim N(0, 1), 
	u \sim \text{Unif}[0, 1], \quad\mbox{and}\quad X \sim \text{Unif}[0, 1], \\
	Y &\sim \text{Pois}(\sin^2(X) + 0.1) + 0.03 X \epsilon_1 + 25 \mathbf{1}\{u < 0.01\} \epsilon_2.
\end{align*}
We use the same distribution to demonstrate the conditional coverage of QOOB. Additionally, we performed the experiments at a small sample size ($n \leq 300$) to understand the effect of sample size on both methods (the original experiments had $n = 2000$, for which QOOB and Split-CQR perform identically). Figure~\ref{fig:synthetic} summarizes the results. 

For this experiment, the number of trees $T$ are set to $100$ for both methods. To choose the nominal quantile level, we first ran the python notebook at \url{https://github.com/yromano/cqr} to reproduce the original experiments performed by \citet{romano2019conformalized}. Their code first learns a nominal quantile level for Split-CQR by cross-validating. On executing their code, we typically observed values near $0.1$ for $\alpha = 0.1$ and hence we picked this nominal quantile level for our experiments as well (for both Split-CQR and QOOB). For our simple 1-dimensional distribution, deeper trees lead to wide prediction sets. This was also observed in the original Split-CQR experiments. To rectify this, the  minimum number of training data-points in the tree leaves was set to $40$; we do this in our experiments as well.

\section{Conclusion}\label{sec:summary}
We introduced an alternative framework to score-based conformal prediction which is based on a sequence of nested prediction sets. We argued that the nested conformal prediction framework is more natural and intuitive. 
We demonstrated how to translate a variety of existing nonconformity scores into nested prediction sets. %
We showed how cross-conformal prediction, the jackknife+, and out-of-bag conformal can be described in our nested framework. The interpretation provided by nested conformal opens up new procedures
to practitioners. We propose one such procedure --- QOOB --- which uses quantile regression forests to perform out-of-bag conformal prediction. \rev{We proposed an efficient cross-conformalization algorithm (Algorithm~\ref{alg:efficientCC}) that makes cross-conformal as efficient as jackknife+. QOOB relies on this efficient cross-conformalization procedure.} We demonstrated empirically that QOOB achieves state-of-the-art performance on multiple real-world datasets. 

\rev{QOOB has certain limitations. First, without making additional assumptions, we can only guarantee $1 - 2\alpha$ coverage for QOOB. On the other hand, we observe that in practice QOOB has coverage slightly larger than $1-\alpha$. Second, QOOB is designed specifically for real-valued responses; extending it to classification and other settings would be interesting. Third, QOOB is computationally intensive compared to the competitive alternative Split-CQR (at least when sample-sizes are high; see Section~\ref{subsec:QOOB-sample-size} for more details). We leave the resolution of these limitations to future work. }

\subsection*{Acknowledgments}
This paper has been revised as per suggestions by the reviewers and editor of Pattern Recognition Special Issue on Conformal Prediction. We would like to thank Robin Dunn, Jaehyeok Shin and Aleksandr Podkopaev for comments on an initial version of this paper. AR is indebted to Rina Barber, Emmanuel Cand{\`e}s and Ryan Tibshirani for several insightful conversations on conformal prediction. This work used the Extreme Science and Engineering Discovery Environment (XSEDE)~\citep{towns2014xsede}, which is supported by National Science Foundation grant number ACI-1548562. Specifically, it used the Bridges system~\citep{Nystrom:2015:BUF:2792745.2792775}, which is supported by NSF award number ACI-1445606, at the Pittsburgh Supercomputing Center (PSC). CG thanks the PSC consultant Roberto Gomez for support with MATLAB issues on the cluster.

\bibliographystyle{apalike}
\bibliography{NestedConformal}

 \newpage
\appendix
\begin{center}
\bf\large APPENDIX
\end{center}
\rev{
\section{Equivalence between score-based conformal prediction and nested conformal}
\label{appsec:Equivalence}}
We show that every instance of score-based conformal prediction can be cast in terms of nested conformal prediction, and vice versa. Appendix~\ref{subsec:equivalence-non-transductive} argues this fact for non-transductive conformal methods which are the focus of the main paper. Appendix~\ref{subsec:equivalence-transductive} describes full transductive conformal prediction using nested sets and argues the equivalence in that setting as well.

\subsection{Equivalence for non-transductive conformal methods}
\label{subsec:equivalence-non-transductive}
Non-transductive conformal methods define a score function $r : \mathcal{X} \times \mathcal{Y} \to \mathbb{R}$ using some part of the training data (a split, a leave-one-out set, a fold, a subsample, etc). Thus unlike full (tranductive) conformal, the $r$ does not depend on the point $(x,y)$ that it is applied to (of course the value $r(x, y)$ depends on $(x, y)$ but not $r$ itself). 

Given a nested family, $\{\mathcal{F}_t(\cdot)\}_{t \in \mathbb{R}}$ learnt on some part of the data,  a nonconformity score can be constructed per equation~\eqref{eq:Score-definition}. We now argue the other direction. Given any nonconformity score $r$ and an $x \in \mathcal{X}$, consider the family of nested sets $\{\mathcal{F}_t(x)\}_{t \in \mathbb{R}}$ defined as: 
\[
\mathcal{F}_t(x) := \{y \in \mathcal{Y}: r(x, y) \leq t\}.
\]
Clearly, $y \in \mathcal{F}_t(x)$ if and only if $r(x, y) \leq t$. Hence,
\[
\inf\left\{t\in\mathcal{T}:\,y\in\mathcal{F}_t(x)\right\} = \inf\left\{t\in\mathcal{T}:\,r(x, y) \leq t\right\} = r(x, y).
\]
Thus, for any nonconformity score $r$, there exists a  family of nested sets that recovers it. 

The randomness in $r$ (through the data on which it is learnt) is included in the conformal validity guarantees of split conformal, cross-conformal, OOB conformal, etc. Similarly, the guarantees based on nested conformal implicitly include the randomness in $\{\mathcal{F}_t(\cdot)\}_{t \in \mathcal{T}}$.

\subsection{An equivalent formulation of full transductive conformal using the language of nested sets}
\label{subsec:equivalence-transductive}
For simpler exposition, in this subsection, we skip the qualification `transductive' and just say `conformal'. Each instance of conformal refers to transductive conformal. We follow the description of a conformal prediction set as defined by \citet{balasubramanian2014conformal}.  Conformal can be defined for any space $\mathcal{Z}$; in the predictive inference setting of this paper, one can think of $z=(x,y) \in \mathcal{Z} = \mathcal{X} \times \mathcal{Y}$. 

A nonconformity $N$-measure is a measurable function $A$ that assigns every sequence $(z_1,\ldots,z_N)$ of $N$ examples to a corresponding sequence $(\alpha_1,\ldots,\alpha_N)$ of $N$ real numbers that is equivariant with respect to permutations, meaning that for any permutation $\pi:[N]\to[N]$,
\[
(\alpha_1,\ldots,\alpha_N) = A(z_1,\ldots,z_N)\quad\Rightarrow\quad (\alpha_{\pi(1)}, \ldots, \alpha_{\pi(N)}) = A(z_{\pi(1)}, \ldots, z_{\pi(N)}).
\]

For some training set $z_1, z_2, \ldots, z_n$ and a candidate point $z$, define the nonconformity score as 
\[
(\alpha_1^z, \ldots, \alpha_{n+1}^z) := A(z_1, \ldots, z_n, z).
\]
For each $z$, define 
\[
p^z := \frac{|\{i \in [n+1]:\,\alpha_i^z \ge \alpha_{n+1}^z\}|}{n + 1}.
\]
Then the \emph{conformal prediction set} determined by $A$ as a nonconformity measure is defined by
\begin{equation}\label{eq:Gamma-alpha}
\Gamma^{\alpha}(z_1, \ldots, z_n) := \{z:\,p^z > \alpha\}.
\end{equation}
In predictive inference, we have a fixed $x$ and wish to learn a prediction set for $y$. This set takes the form 
\[
\{y:\,p^{(x, y)} > \alpha\}.
\]
If the training and test-data are exchangeable, it can be shown that the above prediction set is marginally valid at level $\alpha$ (Proposition 1.2 \citep{balasubramanian2014conformal}). 

The nested (transductive) conformal predictor starts with a sequence of nested sets. For any $N\ge 1$ and $(z_1, \ldots, z_N) \in \mathcal{Z}^N$, let $\{\mathcal{F}_t(z_1, \ldots, z_N)\}_{t\in\mathcal{T}}$ be a collection of sets, with each $\mathcal{F}_t : \mathcal{Z}^N \to 2^{\mathcal{Z}}$ being invariant to permutations of indice (the codomain of $\mathcal{F}_t$ corresponds to all subsets of $\mathcal{Z}$). The nested condition reads that for every $t_1 \leq t_2 \in \mathcal{T}$, $\mathcal{F}_{t_1} \subseteq \mathcal{F}_{t_2}$. 

Unlike the non-transductive setting considered in the main paper, the codomain of $\mathcal{F}_t$ is not a function from $\mathcal{X} \mapsto 2^{\mathcal{Y}}$ but a subset of $\mathcal{X} \times \mathcal{Y}$. The transductive version is more general and includes the non-transductive version; every $g : \mathcal{X} \to 2^\mathcal{Y}$ can be represented by $\{(x, y) : y \in g(x)\} \subseteq \mathcal{X} \times \mathcal{Y}$. 
For observations $Z_1, \ldots, Z_n$ and a possible future $z$, define the scores
\begin{align*}
\alpha_i^z &:= \inf\{t\in\mathcal{T}:\,Z_i\in\mathcal{F}_t(Z_1, \ldots, Z_n, z)\},\quad i\in[n],\\
\alpha_{n+1}^z &:= \inf\{t\in\mathcal{T}:\,z\in\mathcal{F}_t(Z_1,\ldots,Z_n, z)\}.
\end{align*}

Once nonconformity scores are produced, nested conformal reduces to the usual transductive conformal inference.The final prediction set is produced in the usual way:
\[
C_{\alpha}(Z_1, \ldots, Z_n) := \{z:\,p^z > \alpha\},\quad\mbox{where}\quad p^z := \frac{|\{i\in[n+1]:\,\alpha_i^z \ge \alpha_{n+1}^z\}|}{n+1}.
\]

Below, we prove that the converse also holds, that is every instance of standard conformal inference can be viewed as an instance of nested conformal.

\begin{prop}\label{prop:any-conformal-is-nested}
Suppose $A$ is a given nonconformity measure. There exists a nested sequence such that the nested conformal set matches with the conformal prediction set.
\end{prop}
\begin{proof}[Proof of Proposition~\ref{prop:any-conformal-is-nested}]

Consider any $(z_1, \ldots, z_N) \in \mathcal{Z}^{N}$ and let $(\alpha_1, \ldots, \alpha_N) = A(z_1, \ldots, z_N)$. Let $\mathcal{T} =  \mathbb{R}$. If $t = \alpha_i$ for some $i$, define $\mathcal{F}_t(z_1, \ldots, z_N) := \{z_i : \alpha_i \leq t\}$. Also, define $\mathcal{F}_\infty = \mathcal{Z}$ and $\mathcal{F}_{-\infty} = \emptyset$. For $t \notin \{\alpha_1, \ldots, \alpha_N\}$, let $\widehat{t} = \max\{t' \in \{-\infty, \alpha_1, \ldots, \alpha_N, \infty\} : t' < t\}$. Then, set $\mathcal{F}_t(z_1, \ldots, z_N) = \mathcal{F}_{\widehat{t}}(z_1, \ldots, z_N)$. The nested condition is true by construction. It is easily verified that the conformal prediction set induced by this nested sequence is equivalent to the one produced by the original $A$: we have for $i \in [N]$, 
\begin{align*}
\alpha_i^{\text{nested}} &= \inf\{t \in \mathcal{T} : z_i \in \mathcal{F}_t(z_1, \ldots, z_N)\} = \alpha_i.
\end{align*}
In other words, the nonconformity scores are recovered, and consequently the conformal prediction set will be identical. 
\end{proof}

Proposition 1.3 of~\cite{balasubramanian2014conformal} shows that conformal prediction is universal in a particular sense (informally, any valid scheme for producing assumption-free confidence sets can be replaced by a conformal prediction scheme that is at least as efficient). Since everything that can be accomplished via nested conformal prediction can also be done via conformal prediction and vice versa, nested conformal prediction is also universal in the same sense.

\section{K-fold cross-conformal and CV+ using nested sets}
\label{app:k-fold}
In Section~\ref{sec:Jackknife-plus} we rephrased leave-one-out cross-conformal and jackknife+ in the nested framework. In this section, we will now describe their K-fold versions. 
 \subsection{Extending K-fold cross-conformal using nested sets}\label{subsec:Cross-conformal}

 Suppose $S_1, \ldots, S_K$ denotes a disjoint partition of $\{1,2,\ldots,n\}$ such that $|S_1| = |S_2| = \cdots = |S_K|$.
 For exchangeability, this equality of sizes is very important.
 Let $m = n/K$ (assume $m$ is an integer).
 Let $\{\mathcal{F}_t^{-S_k}(x)\}_{t\in\mathcal{T}}$ be a sequence of nested sets computed based on $\{1,2,\ldots,n\}\setminus S_k$.
 Define the score
 \[
 r_i(x, y) ~:=~ \inf\left\{t\in\mathcal{T}:\,y\in\mathcal{F}_t^{-S_{k(i)}}(x)\right\},
 \] 
 where $k(i)\in[K]$ is such that $i\in S_{k(i)}$.
 The cross-conformal prediction set is now defined as
 \[
 C^{\texttt{cross}}_K(x) ~:=~ \left\{y:\,\sum_{i=1}^n \mathbbm{1}\{r_i(X_i, Y_i) < r_i(x, y)\} < (1-\alpha)(n+1)\right\}.
 \]
 It is clear that if $K = n$ then $C_K^{\texttt{cross}}(x) = C^{\texttt{LOO}}(x)$ for every $x$. The following result proves the validity of $C^{\texttt{cross}}_K(\cdot)$ as an extension of Theorem 4 of~\citet{barber2019predictive}. This clearly reduces to Theorem~\ref{thm:validity-nested-conformal-jackknife-plus} if $K = n$. 
 \begin{thm}\label{thm:K-fold-validity}
 If $(X_i, Y_i), i\in[n]\cup\{n+1\}$ are exchangeable and sets $\mathcal{F}_t^{-S_k}(x)$ constructed based on $\{(X_i,Y_i):\,i\in[n]\setminus S_k\}$ are invariant to their ordering, then 
 \[
 \mathbb{P}(Y_{n+1}\in  C^{\texttt{cross}}_K(X_{n+1})) ~\ge~ 1- 2\alpha - \min\left\{\frac{1-K/n}{K+1}, \frac{2(K-1)(1-\alpha)}{n+K}\right\}.
 \]
 \end{thm}
 See Appendix~\ref{appsec:Jackknife-plus} for a proof.
Although the construction of $C^{\texttt{cross}}_K(x)$ is based on a $K$ fold split of the data. The form is exactly the same as that of $C^{\texttt{LOO}}(x)$ in~\eqref{eq:LOO}. Hence the computation of $C^{\texttt{cross}}_K(x)$ can be done based on the discussion in Section~\ref{subsec:compuation-cross}. In particular, if each of the nested sets $\mathcal{F}_t(x)$ are either intervals or empty sets, the $C^{\texttt{cross}}_K(x)$ aggregation step (after computing the residuals) can be performed in time $O(n\log n)$.

 \subsection{Extending CV+ using nested sets}
 The prediction sets ${C}^{\texttt{cross}}(x)$ and $C^{\texttt{cross}}_K(x)$ are defined implicitly and are in general not intervals. The sets ${C}^{\texttt{cross}}(x)$ and $C^{\texttt{cross}}_K(x)$ can be written in terms of nested sets as
 \[
 C^{\texttt{cross}}_K(x) ~:=~ \left\{y:\,\sum_{i=1}^n \mathbbm{1}\big\{y \notin \mathcal{F}_{r_i(X_i, Y_i)}^{-S_{k(i)}}(x)\big\} < (1-\alpha)(n+1)\right\}.
 \]
 In this subsection, we show that there exists an explicit interval that always contains $C_{K}^{\texttt{cross}}(x)$ whenever $\{\mathcal{F}_t(x)\}_{t\in\mathcal{T}}$ is a collection of nested intervals (instead of just nested sets). This is a generalization of the CV+ interval defined by~\citet{barber2019predictive}. The discussion of this subsection can be extended to the case whenever the nested sets are either intervals or the empty set just like we did for leave-one-out cross-conformal and jackknife+ in Appendix~\ref{app:empty-case}. 

 If each $\mathcal{F}_t(x)$ is an interval, we can write $\mathcal{F}_{r_i(X_i,Y_i)}^{-S_{k(i)}}(x) = [\ell_i(x), u_i(x)]$ for some $\ell_i(\cdot), u_i(\cdot)$. Using this notation, we can write $C^{\texttt{cross}}_K(x)$ as
 \[
 C^{\texttt{cross}}_K(x) ~:=~ \left\{y:\,\sum_{i=1}^n \mathbbm{1}\{y\notin [\ell_i(x), u_i(x)]\} < (1-\alpha)(n+1)\right\}.
 \] 

Following the same arguments as in Section~\ref{subsec:jackknife+} we can define the CV+ prediction 
 \begin{equation}
C^{\texttt{CV+}}_K(x) 	  ~:=~  [q^{-}_{n,\alpha}(\ell_i(x)), q^{+}_{n,\alpha}(u_i(x))] ~\supseteq~ C^{\texttt{cross}}_K(x)  ,
 \end{equation}
 where $q^{-}_{n,\alpha}(\ell_i(x))$ denotes the $\lfloor\alpha(n+1)\rfloor$-th smallest value of  $\{\ell_i(x)\}_{i=1}^n$. and $q^{+}_{n,\alpha}(u_i(x))$ denotes the $\lceil(1-\alpha)(n+1)\rceil$-th smallest value of  $\{u_i(x)\}_{i=1}^n$.
 For $K = n$, $C^{\texttt{CV+}}_K(x)$ reduces to the jackknife+ prediction interval  $C^{\texttt{JP}}(x)$. %
$C^{\texttt{CV+}}(x)$ and $C^{\texttt{JP}}(x)$ are always non-empty intervals if each of the $\mathcal{F}_t(x)$ are non-empty intervals. Because $C^{\texttt{cross}}_K(x) \subseteq C^{\texttt{CV+}}_K(x)$ for all $x$, we obtain a validity guarantee from Theorem~\ref{thm:K-fold-validity}: for all $2\le K\le n$,
 \[
 \mathbb{P}\left(Y_{n+1}\in C^{\texttt{CV+}}_K(X_{n+1})\right) \ge 1 - 2\alpha - \min\left\{\frac{1-K/n}{K + 1},\,\frac{2(K-1)(1-\alpha)}{n + K}\right\}.
 \]
Note that the convex hull of the K-fold cross-conformal prediction set is also an interval smaller than the CV+ interval
 \begin{equation*}
C^{\texttt{cross}}_K(x)  \subseteq \text{Conv}(C^{\texttt{cross}}_K(x)) \subseteq C^{\texttt{CV+}}_K(x) 	.
\end{equation*}
In some cases, the containment above can be strict, and hence we recommend the K-fold cross conformal or its convex hull over CV+. 

\section{Nested conformal based on multiple repetitions of splits}\label{sec:sampling-nested}
In Section~\ref{sec:split-conformal}, we described the nested conformal version of split conformal which is based on one particular split of the data into two parts, and in Appendix~\ref{app:k-fold} we discussed partitions of the data into $2 \leq K \leq n$ parts. In practice, however, to reduce the additional variance due to randomization, one might wish to consider several (say $M$) different splits of data into two parts combine these predictions. \citet{lei2018distribution} discuss a combination of $M$ split conformal prediction sets based on Bonferroni correction and in this section, we consider an alternative combination method that we call \emph{subsampling conformal}. The same idea can also be used for cross-conformal version where the partition of the data into $K$ folds can be repeatedly performed $M$ times.
The methods to be discussed are related to those proposed by~\citet{carlsson2014aggregated}, \citet{vovk2015cross}, and~\citet{linusson2017calibration, linusson2019efficient}, but these papers do not provide validity results for their methods.

\subsection{Subsampling conformal based on nested prediction sets}
Fix a number $K\ge1$ of subsamples. Let $M_1$, $M_2$, $\ldots$, $M_K$ denote independent and identically distributed random sets drawn uniformly from $\{M: M\subset[n]\}$; one can also restrict to $\{M:M\subset[n], |M| = m\}$ for some $m\ge1$. For each set $M_k$, define the $p$-value for the new prediction $y$ at $X_{n+1}$ as
\[
p_k^y(x) := \frac{|\{i\in M_k^c:\,r_k(x, y) \le r_k(X_i, Y_i)\}| + 1}{|M_k^c| + 1},
\]
where the scores $r_k(X_i, Y_i)$ and $r_k(x, y)$ are computed as
\begin{align*}
r_k(X_i, Y_i) &:= \inf\{t\in\mathcal{T}:\,Y_i \in \mathcal{F}_t^{M_k}(X_i)\},\quad i\in M_k^c,\\
r_k(x, y) &:= \inf\{t\in\mathcal{T}:\,y\in\mathcal{F}_t^{M_k}(x)\},
\end{align*}
based on nested sets $\{\mathcal{F}_t^{M_k}(x)\}_{t\in\mathcal{T}}$ computed based on observations in $M_k$.
Define the prediction set as
\[
C_{K}^{\texttt{subsamp}}(x) := \left\{y:\frac{1}{K}\sum_{k=1}^K \frac{|\{i\in M_k^c:\,r_k(x, y) \le r_k(X_i, Y_i)\}| + 1}{|M_k^c| +1} > \alpha\right\}.
\]
It is clear that for $K = 1$, $C^{\texttt{subsamp}}_K(x)$ is same as the split conformal prediction set discussed in Section~\ref{sec:split-conformal}.
The following results proves the validity of $C_K^{\texttt{subsamp}}(\cdot)$.
\begin{thm}\label{thm:nested-subsampling}
	If $(X_i,Y_i),\,i\in[n]\cup\{n+1\}$ are exchangeable, then for any $\alpha\in[0,1]$ and $K\ge1$,
	$\mathbb{P}\left(Y_{n+1}\notin C_{K}^{\texttt{subsamp}}(X_{n+1})\right) \le \min\{2,K\}\alpha.$
\end{thm}
See Appendix~\ref{appsec:sampling-nested} for a proof.
Note that we can write $p_k^y(x)$ as $p^y(x; M_k)$ by adding the argument for observations used in computing the nested sets. Using this notation, we can write for $K$ large
\begin{equation}\label{eq:approx-average}
\frac{1}{K}\sum_{k=1}^K p^y(x; M_k) ~\approx~ \mathbb{E}_{M}[p^y(x; M)],
\end{equation}
where the expectation is taken with respect to the random ``variable'' $M$ drawn uniformly from a collection of subsets of $[n]$ such as $\{S:\,S\subset[n]\}$ or $\{S:\,S\subset[n], |S| = m\}$ for some $m\ge1$. Because any uniformly drawn element in $\{S:\,S\subset[n]\}$ can be obtained by sampling from $[n]$ without replacement (subsampling), the above combination of prediction intervals can be thought as subbagging introduced in~\citep{buhlmann2002analyzing}.

\citet[Section 2.3]{lei2018distribution} combine the $p$-values $p_k^y(x)$ by taking the minimum. They define the set
\[
C_K^{\texttt{split}}(x) := \bigcap_{k=1}^K \{y:\,p_k^y(x) > \alpha/K\} ~=~ \left\{y:\,K\min_{1\le k\le K}p_k^y(x) > \alpha\right\}.
\]
Because $p_1^y(x), p_2^y(x), \ldots, p_K^y(x)$ are independent and identically distributed (conditional on the data), averaging is a natural stabilizer than the minimum; all the $p$-values should get equal contribution towards the stabilizer but the minimum places all its weight on one $p$-value.  

\citet[Appendix B]{vovk2015cross} describes a version of $C_{K}^{\texttt{subsamp}}(\cdot)$ using bootstrap samples instead of subsamples and this corresponds to bagging. We consider this version in the following subsection. 

\subsection{Bootstrap conformal based on nested prediction sets}\label{sec:Bootstrap-nested}
The subsampling prediction set $C_K^{\texttt{subsamp}}(x)$ is based on sets $M_k$ obtained by sampling without replacement. Statistically a more popular alternative is to form sets $M_k$ by sampling with replacement, which corresponds to bootstrap. 

Let $B_1, \ldots, B_K$ denote independent and identically distributed bags (of size $m$) obtained by random sampling with replacement from $[n] = \{1,2,\ldots,n\}$. For each $1\le k\le K$, consider scores
\begin{align*}
r_{k}(X_i, Y_i) &:= \inf\{t\in\mathcal{T}:\,Y_i\in\mathcal{F}_t^{B_k}(X_i)\},\quad i\in B_k^c,\\
r_k(x, y) &:= \inf\{t\in\mathcal{T}:\,y\in \mathcal{F}_t^{B_k}(x)\},
\end{align*}
based on nested sets $\{\mathcal{F}_t^{B_k}(x)\}_{t\in\mathcal{T}}$ computed based on observations $(X_i, Y_i), i\in B_k$; $B_k$ should be thought of as a bag rather than a set of observations because of repititions of indices. Consider the prediction interval
\[
C_{\alpha,K}^{\texttt{boot}}(x) := \left\{y:\,\frac{1}{K}\sum_{k=1}^K \frac{|\{i\in [n]\setminus B_k:\,r_k(x, y) \le r_k(X_i, Y_i)\}| + 1}{|[n]\setminus B_k| + 1} > \alpha\right\}.
\]
This combination of prediction interval based on bootstrap sampling is a version of bagging and was considered in~\citet[Appendix B]{vovk2015cross}. The following result proves a validity bound for $C_{\alpha,K}^{\texttt{boot}}(X_{n+1})$. %
\begin{thm}\label{thm:bootstrap-conformal}
	If $(X_i, Y_i), i\in[n]\cup\{n+1\}$ are exchangeable, then for any $\alpha\in[0,1]$ and $K \ge 1$, $\mathbb{P}(Y_{n+1}\notin C_{\alpha, K}^{\texttt{boot}}(X_{n+1})) \le \min\{2,K\}\alpha$.
\end{thm}
	See Appendix~\ref{appsec:sampling-nested} for a proof.
\citet[Proposition 1]{carlsson2014aggregated} provide a similar result in the context of aggregated conformal prediction but require an additional consistent sampling assumption. 
The computation of the subsampling and the bootstrap conformal prediction sets is no different from that of cross-conformal and the techniques discussed in susbection~\ref{subsec:compuation-cross} are still applicable. 

\rev{\citet{linusson2019efficient} demonstrated that aggregated conformal methods tend to be conservative. We also observed this in our simulations. Because of this, we did not present these methods in our experiments. }

\section{Cross-conformal and Jackknife+ if $\mathcal{F}_t(x)$ could be empty}\label{app:empty-case}
The definition of cross-conformal~\eqref{eq:LOO} is agnostic to the interval interpretation through $\mathcal{F}_{r_i(X_i, Y_i)}^{-i}(x)$ since $r_i(X_i, Y_i)$ and $r_i(x, y)$ are well-defined irrespective of whether $\mathcal{F}_t(x)$ is an interval or not. However, the discussion in Sections~\ref{subsec:jackknife+} and \ref{subsec:compuation-cross} indicates that the interval interpretation is useful for interpretability as well as to be able to compute $C^\texttt{LOO}(x)$ efficiently. %
In these subsections, we assumed that $\mathcal{F}_t(x)$ is always an interval. However there exist realistic scenarios in which $\mathcal{F}_t(x)$ is always either an interval or an empty set. Fortunately, it turns out that the discussion about jackknife+ and efficient cross-conformal computation can be generalized to this scenario as well. 

\subsection{When can $\mathcal{F}_t(x)$ be empty?}
Consider the quantile estimate based set entailed by the CQR formulation of \citet{romano2019conformalized}: $\mathcal{F}_t(x) = [\widehat{q}_{\alpha/2}(x) - t, \widehat{q}_{1-\alpha/2}(x) + t]$. $\mathcal{F}_t(x)$ is implicitly defined as the empty set if $t < 0.5(\widehat{q}_{\alpha/2}(x) - \widehat{q}_{1-\alpha/2}(x))$. Notice that since $t$ can be negative, the problem we are considering is different from the quantile crossing problem which has separately been discussed by \citet[Section 6]{romano2019conformalized}, and may occur even if the quantile estimates satisfy $\widehat{q}_{\alpha/2}(x) \leq \widehat{q}_{1-\alpha/2}(x)$ for every $x$.  In the cross-conformal or jackknife+ setting, $\mathcal{F}_{r_i(X_i, Y_i)}^{-i}(x) $ is empty for a test point $x$ if 
\[
r_i(X_i, Y_i) < 0.5(\widehat{q}^{-i}_{\alpha/2}(x) - \widehat{q}^{-i}_{1-\alpha/2}(x)),
\]
where $\widehat{q}^{-i}$ are the quantile estimates learnt leaving $(X_i, Y_i)$ out. If the above is true, it implies that $r_i(X_i, Y_i) < r_i(x, y)$ for every possible $y$. From the conformal perspective, the interpretation is that $(x, y)$ is more `non-conforming' than $(X_i, Y_i)$ for every $y \in \mathcal{Y}$. In our experiments, we observe this does occur occasionally for cross-conformal (or out-of-bag conformal described in Section~\ref{sec:OOB-conformal}) with quantile-based nested sets. In hindsight, it seems reasonable that this would happen at least once across multiple training and test points. 

\subsection{Jackknife+ and efficiently computing $C^\texttt{LOO}(x)$ in the presence of empty sets}
Suppose $\mathcal{F}_{t}(x)$ is an interval whenever it is non-empty. Define \[\Lambda_x := \{i : \mathcal{F}_{r_i(X_i, Y_i)}^{-i}(x) \text{ is not empty}\},\]
Equivalently we can write $\Lambda_x := \{i : \exists y, y \in [\ell_i(x), u_i(x)]\}$. The key observation of this section is that for jackknife+ and the computation, only the points in $\Lambda_x$ need to be considered. To see this, we re-write the interval definition of the cross-conformal prediction~\eqref{eq:loo-interval-definition}:
\begin{align}
C^{\texttt{LOO}}(x) &=  \left\{y:\,  \alpha(n+1) - 1 < \sum_{i =1}^n \mathbbm{1}\{y \in [\ell_i(x), r_i(x)]\} \right\} \nonumber \\
&=  \left\{y:\,  \alpha(n+1) - 1 < \sum_{i \in \Lambda_x} \mathbbm{1}\{y \in [\ell_i(x), r_i(x)]\} \right\}, \label{eq:loo-interval-definition-reinterpret}
\end{align}
since if $i \notin \Lambda_x$, no $y$ satisfies $y \in [\ell_i(x), u_i(x)]$.
Following the same discussion as in Section~\ref{subsec:jackknife+}, we can define the jackknife+ prediction interval as
\[
C^{\texttt{JP}}(x) ~:=~ [q_{n,\alpha}^{-}(\{\ell_i(x)\}_{i \in \Lambda_x}),\, -q_{n,\alpha}^{-}(\{-u_i(x)\}_{i \in \Lambda_x})]\]
where $q_{n,\alpha}^{-}(\{\ell_i(x)\}_{i \in \Lambda_x})$ denotes the 
$\lfloor\alpha(n+1)\rfloor$-th smallest value of  $\{\ell_i(x)\}_{i \in \Lambda_x}$ and $q_{n,\alpha}^{-}(\{-u_i(x)\}_{i \in \Lambda_x})$ denotes the $\lfloor\alpha(n+1)\rfloor$-th smallest value of  $\{-u_i(x)\}_{i \in \Lambda_x}$ (if $|\Lambda_x| = n$ the above definition can be verified to be exactly the same as the one provided in equation~\eqref{eq:JP-definition}). It may be possible that $\lfloor\alpha(n+1)\rfloor > |\Lambda_x|$ in which case the jackknife+ interval (and the cross-conformal prediction set) should be defined to be empty. The $1-2\alpha$ coverage guarantee continues to hold marginally even though we may sometimes return empty sets. 

To understand the computational aspect, 
we note that the discussion of Section~\ref{subsec:compuation-cross} continues to hold with %
$\mathcal{Y}^x$ (equation~\eqref{eq:y-k-main}) redefined to only include the intervals end-points for intervals which are defined:
\begin{equation}
\mathcal{Y}^x := \bigcup_{i\in \Lambda_x} \Lbag \ell_i(x), u_i(x)\Rbag,
\end{equation}
and $\mathcal{S}^x$ defined only for these points. This also follows from the re-definition of $C^\texttt{LOO}(x)$ in~\eqref{eq:loo-interval-definition-reinterpret}. With the definition above, Algorithm~\ref{alg:efficientCC} %
works generally for the case where $\mathcal{F}_t(x)$ could be an interval or an empty set.

\section{Proofs}\label{appsec:proofs}

\subsection{Proofs of results in Section~\ref{sec:split-conformal}}\label{appsec:split-conformal}

\begin{proof}[Proof of Proposition~\ref{thm:Conformal-main-result}]
Set $r_{n+1} := r(X_{n+1}, Y_{n+1})$. By the construction of the prediction interval, we have
\[
Y_{n+1}\in C(X_{n+1})\quad\mbox{if and only if}\quad r_{n+1}\le Q_{1-\alpha}(r, \mathcal{I}_2).
\]
Hence
\[
\mathbb{P}(Y_{n+1}\in C(X_{n+1})\big|\{(X_i, Y_i):i\in\mathcal{I}_1\}) = \mathbb{P}(r_{n+1}\le Q_{1-\alpha}(r,\mathcal{I}_2)\big|\{(X_i,Y_i):i\in\mathcal{I}_1\}).
\]
Exchangeability of $(X_i, Y_i), i\in[n]\cup\{n+1\}$ implies the exchangeability of $(X_i, Y_i), i\in\mathcal{I}_2\cup\{n+1\}$ conditional on $(X_i, Y_i), i\in\mathcal{I}_1$. This in turn implies that $r_i, i\in\mathcal{I}_2\cup\{n+1\}$ are also exchangeable (conditional on the first split of the training data) and thus Lemma 2 of~\citet{romano2019conformalized} yields
\[
\mathbb{P}(r_{n+1} \le Q_{1-\alpha}(r, \mathcal{I}_2)|\{(X_i, Y_i):\,i\in\mathcal{I}_1\}) \ge 1-\alpha,
\]
and the assumption of almost sure distinctness of $r_1,\ldots,r_n$ implies~\eqref{eq:Upper-bound}. 
\end{proof}

\subsection{Proof of Theorem~\ref{thm:validity-nested-conformal-jackknife-plus}}\label{appsec:Jackknife-plus}
Define the matrix $D\in\mathbb{R}^{(n+1)\times(n+1)}$ with entries
\[
D_{i,j} := \begin{cases}+\infty,&\mbox{if }i= j,\\r_{(i,j)}(X_i, Y_i),&\mbox{if }i\neq j,\end{cases}
\]
where
$r_{(i,j)}(x, y) := \inf\{t\in\mathcal{T}:\,y\in\mathcal{F}_t^{-(i,j)}(x)\},$
with $\mathcal{F}_t^{-(i,j)}(x)$ defined analogues to $\mathcal{F}_t^{-i}(x)$ computed based on $\{(X_k, Y_k):\,k\in[n+1]\setminus\{i,j\}\}$. It is clear that $D_{i,n+1} = r_{i}(X_i, Y_i)$, and $D_{n+1,i} = r_i(X_{n+1}, Y_{n+1})$. Therefore,
\[
Y_{n+1}\notin{C}^{\texttt{LOO}}(X_{n+1})\quad\mbox{if and only if}\quad(1-\alpha)(n+1)\le \sum_{i=1}^{n+1} \mathbbm{1}\{D_{i,n+1} < D_{n+1,i}\},
\]
which holds if and only if $n+1\in\mathcal{I}(D)$, with $\mathcal{I}(D)$ is defined as in~\eqref{eq:Definition-I-of-A-set}. Hence from Theorem~\ref{thm:fundamental-jackknife-result} the result is proved, if its assumption is verified. This assumption follows from the fact that $\mathcal{F}_t^{-(i,j)}(x)$ treats its training data symmetrically.
\qed 

\subsection{Proof of Theorem~\ref{thm:OOB-cross-validity}}\label{appsec:OOB-conformal}
	The OOB-conformal procedure treats the training data $\{(X_i, Y_i)\}_{i\in [n]}$ exchangeably but not the test point $(X_{n+1}, Y_{n+1})$. To prove a validity guarantee, we first lift the OOB-conformal procedure to one that treats all points $\{(X_i, Y_i)\}_{i\in [n+1]}$ exchangeably. This is the reason we require a random value of $K$, as will be evident shortly. 
	
	The lifted OOB-conformal method is described as follows. Construct a collection of sets $\{\widetilde{M}_i\}_{i=1}^{\widetilde{K}}$, where each $\widetilde{M}_i$ is independently drawn using bagging or subsampling $m$ samples from $[n+1]$ (instead of $[n]$). Following this, for every $(i, j) \in [n+1]\times [n+1]$ with $i \neq j$ define $\{\mathcal{F}_t^{-(i, j)}\}_{t \in \mathcal{T}}$ as the sequence of nested sets learnt by ensembling samples $M_{-(i,j)} := \{M_k: i, j \notin M_k\}$. The nested sets $\{\mathcal{F}_t^{-(i, j)}\}_{t \in \mathcal{T}}$ are then used to compute residuals on $(X_i, Y_i)$ and $(X_j, Y_j)$: we define a matrix $D\in\mathbb{R}^{(n+1)\times(n+1)}$ with entries
	\[
	D_{i,j} := \begin{cases}+\infty,&\mbox{if }i= j,\\r_{(i,j)}(X_i, Y_i),&\mbox{if }i\neq j,\end{cases}
	\]
	where
	$r_{(i,j)}(x, y) := \inf\{t\in\mathcal{T}:\,y\in\mathcal{F}_t^{-(i,j)}(x)\}.$
	
	 We will now invoke Theorem~\ref{thm:fundamental-jackknife-result} for the exchangeable random variables $\{Z_i = (X_i, Y_i)\}_{i=1}^{n+1}$. The assumption of Theorem~\ref{thm:fundamental-jackknife-result} holds for the elements $D_{i,j}$ since the ensemble method we use to learn $\{\mathcal{F}_t^{-(i,j)}\}_{t\in\mathcal{T}}$ treats all random variables apart from $Z_i, Z_j$ symmetrically. Thus for every $j \in [n+1]$, 
	\[
	\Pr(j \notin \mathcal{I}(D)) \leq 2\alpha - \frac{1}{n+1} \leq 2\alpha,
	\]
	for $\mathcal{I}(D)$ defined in~\eqref{eq:Definition-I-of-A-set}. We will now argue that for $i \in [n]$, $D_{i,n+1} = r_{i}(X_i, Y_i)$, and $D_{n+1,i} = r_i(X_{n+1}, Y_{n+1})$. Notice that for every $j \in [\widetilde{K}]$,
	 \begin{align*}
	 	\text{if we use bagging: }	& \Pr(n+1 \notin \widetilde{M}_j) = \left(1 - \frac{1}{n+1}\right)^m ;\\ 
	 	\text{if we use subsampling: }& \Pr(n+1 \notin \widetilde{M}_j) = \left(1 - \frac{m}{n+1}\right).
	 \end{align*}
	 and so $K = |\{j: n+1 \notin \widetilde{M}_j\}| \sim \text{Bin}(\widetilde{K}, (1 - \frac{1}{n+1})^m)$ for bagging and $K  \sim \text{Bin}(\widetilde{K}, 1 - \frac{m}{n+1})$ for subsampling. Evidently, we can conclude that conditioned on the set $\{j: n+1 \notin \widetilde{M}_j\}$, $\{M_j\}_{j=1}^K \overset{d}{=} \{\widetilde{M}_j: n+1 \notin \widetilde{M}_j\}  $. In other words, the OOB-conformal bagging or subsampling procedure is embedded in its lifted version. Therefore,
	 \begin{align*}
	 	\Pr(Y_{n+1}\notin{C}^{\texttt{OOB}}(X_{n+1})) &= \Pr((1-\alpha)(n+1)\le \sum_{i=1}^{n+1} \mathbbm{1}\{D_{i,n+1} < D_{n+1,i}\}) \\
	 	& = \Pr(n+1\in\mathcal{I}(D)) \qquad\qquad\qquad\text{(per definition~\eqref{eq:Definition-I-of-A-set}),}
	 \end{align*}
	 which as we have shown happens with probability at most $2\alpha$. This completes the proof.
\qed 

\subsection{Proof of Proposition~\ref{prop:efficientCC-correctness}}
\label{appsec:alg-proof}
To see why Algorithm~\ref{alg:efficientCC} works, we describe it step by step along with variable definitions. To simplify understanding, assume that $\mathcal{Y}^x$ does not contain repeated elements (Algorithm~\ref{alg:efficientCC} continues to correctly compute the cross-conformal prediction set even if this is not true; the requirement mentioned on the ordering of elements in $\mathcal{Y}^x$ before definition~\eqref{eq:s-k-main} is crucial for Algorithm~\ref{alg:efficientCC} to remain correct with repeated elements). As we make a single pass over $\mathcal{Y}^x$ in sorted order, at every iteration $i$, when we are on line \ref{line:condition-2} or line \ref{line:condition-3}, the variable $count$ stores the number of training points that are more nonconforming than $(x, y^x_i)$; $count$ is increased by $1$ whenever a left end-point is seen (line \ref{line:increase-count}) and is decreased by $1$ after a right end-point is seen (line \ref{line:decrease-count}). Thus $count$ correctly computes the left hand side of condition~\eqref{eq:y-condition} for the current value of $y \in \mathcal{Y}^x$. The rest of the algorithm compares the value in $count$ to the right hand side of condition~\eqref{eq:y-condition}, which is stored in $threshold$, to compute the prediction set $C^x$. %

If $count$ is strictly larger than $\alpha(n+1)-1$, then by \eqref{eq:y-condition}, $y^x_i \in C^\texttt{LOO}(x)$. If this were not true for $y^x_{i-1}$ (as checked in line \ref{line:condition-2}), then for every $y \in (y^x_{i-1}, y^x_{i})$, we have $\ y \notin C^\texttt{LOO}(x)$. Hence $y^x_i$ is a left end-point for one of the intervals in $C^\texttt{LOO}(x)$. We store this value of $y$ in $left\_{endpoint}$ until the right end-point is discovered (line \ref{line:update-left-endpoint}). Next, if we are at a right end-point, if the current value of $count$ is larger than $\alpha(n+1)-1$, and if $count$ is at most $\alpha(n+1)-1$ after the interval ends (condition on line \ref{line:condition-3}), then $y _i^x\in C^\texttt{LOO}(x)$ and for every $y \in (y_i^x, y_{i+1}^x)$, $y \notin C^\texttt{LOO}(x)$. %
Thus $y _i^x$ is a right end-point for some interval in $C^x$, with the left end-point given by the current value of $left\_{endpoint}$. We update $C^x$ accordingly in line \ref{line:update-prediction-set}. %
\qed

\subsection{Auxiliary lemmas used in Appendix~\ref{appsec:Jackknife-plus}}\label{appsec:Auxiliary-results}
For any matrix $A\in\mathbb{R}^{N\times N}$ and $\alpha\in[0,1]$, define
\begin{equation}\label{eq:Definition-I-of-A-set}
\mathcal{I}(A) := \left\{i\in[N]:\,\sum_{j=1}^N \mathbbm{1}\{A_{ji} < A_{ij}\} \ge (1-\alpha)N\right\}.
\end{equation}
\begin{lem}[Section 5.3 of~\citet{barber2019predictive}]\label{lem:Cardinality}
For any matrix $A$,
\[
\frac{|\mathcal{I}(A)|}{N} \le 2\alpha - \frac{1}{N}.
\]
\end{lem}
\begin{lem}[Section 5.2 of~\citet{barber2019predictive}]\label{lem:Exchangebility-implication}
If $A$ is a matrix of random variables such that for any permutation matrix $\Pi$, $A \overset{d}{=} \Pi A \Pi^{\top},$
then for all $1\le j\le N$,
\[
\mathbb{P}(j\in \mathcal{I}(A)) = \frac{\mathbb{E}[|\mathcal{I}(A)|]}{N} \le 2\alpha - \frac{1}{N}. 
\]
\end{lem}
\begin{rem}
The condition $A \overset{d}{=} \Pi A \Pi^{\top}$ (for any permutation matrix $\Pi$) is equivalent to $(A_{i,j}) \overset{d}{=} (A_{\pi(i),\pi(j)})$ for any permutation $\pi:[N]\to[N]$
\end{rem}
Consider the following form of matrices:
\begin{equation}\label{eq:Matrix-definition}
A_{i,j} = \begin{cases}+\infty,&\mbox{if }i=j,\\
\mathbb{G}(Z_i, \{Z_1,\ldots,Z_N\}\setminus\{Z_i, Z_j\}),&\mbox{if }i\neq j,
\end{cases}
\end{equation}
for exchangeable random variables $Z_1,\ldots,Z_N$.
\begin{lem}\label{lem:Exchangeable-proof}
If $Z_1,\ldots,Z_N$ are exchangeable and $\mathbb{G}(\cdot, \cdot)$ treats the elements of its second argument symmetrically, then the matrix $A$ defined by~\eqref{eq:Matrix-definition} satisfies
\[
A \overset{d}{=} \Pi A\Pi^{\top}, 
\]
for any permutation matrix $\Pi$.
\end{lem}
\begin{proof}
Observe that for any $i$, $A_{i,i} = A_{\pi(i),\pi(i)}$ deterministically. For any $i\neq j$, and $\pi(i) = k\neq \pi(j) = \ell$,
\begin{align*}
A_{i,j} &:= \mathbb{G}(Z_i, \{Z_1,\ldots,Z_N\}\setminus\{Z_i, Z_j\}),\\
A_{k,\ell} &:= \mathbb{G}(Z_{k}, \{Z_1,\ldots,Z_N\}\setminus\{Z_k, Z_{\ell}\}).
\end{align*}
Exchangeability of $Z_1,\ldots,Z_N$ implies that for any permutation $\pi:[N]\to[N]$ and any function $F$ that depends symmetrically on $Z_1,\ldots,Z_N$
\[
F(Z_i, Z_j) \overset{d}{=} F(Z_{\pi(i)}, Z_{\pi(j)}).
\]
The result follows by taking $F(Z_i, Z_j) := \mathbb{G}(Z_i, \{Z_1,\ldots,Z_N\}\setminus\{Z_i, Z_j\})$.
\end{proof}
\begin{thm}\label{thm:fundamental-jackknife-result}
If $\mathbb{G}(\cdot, \cdot)$ is a function that treats the elements of its second argument symmetrically, then for any set of exchangeable random variables $Z_1,\ldots,Z_{N}$, and matrix $A$ defined via~\eqref{eq:Matrix-definition}, we have
\[
\mathbb{P}(j\in\mathcal{I}(A)) \le 2\alpha - \frac{1}{N}\quad\mbox{for all}\quad j\in[N].
\]
\end{thm}
\begin{proof}
The proof follows by combining Lemmas~\ref{lem:Cardinality},~\ref{lem:Exchangebility-implication},~\ref{lem:Exchangeable-proof}.
\end{proof}
\subsection{Proofs of results in Appendix~\ref{sec:sampling-nested}}\label{appsec:sampling-nested}
\begin{proof}[Proof of Theorem~\ref{thm:nested-subsampling}]
Conditional on the randomness of $M_k$, $p_k^{Y_{n+1}}(X_{n+1})$ is a valid $p$-value and hence conditional on $M_1,\ldots,M_K$, $(2/K)\sum_{k=1}^K p_k^{Y_{n+1}}(X_{n+1})$ is a valid $p$-value, by Proposition 18 of~\citet{vovkcombining}. Therefore, for any $\alpha\in[0,1]$,
\begin{equation}\label{eq:p-value-average}
\mathbb{P}\left(\frac{1}{K}\sum_{k=1}^K p_k^{Y_{n+1}}(X_{n+1}) \le \alpha\right) \le \min\{2,K\}\alpha,
\end{equation}
which implies the result. (The factor $2$ follows from~\citep{vovkcombining} for $K \ge 2$ and for $K = 1$ the factor $2$ is not necessary because $p_1^{Y_{n+1}}(X_{n+1})$ is a valid $p$-value.)
\end{proof}
\begin{proof}[Proof of Theorem~\ref{thm:bootstrap-conformal}]
Fix $1\le k\le K$. Conditional on $B_k$, the scores $r_{k}(X_i, Y_i), i\in B_k^c, r_k(X_{n+1}, Y_{n+1})$ are exchangeable because $\{(X_i, Y_i):i\in B_k^c\cup\{n+1\}\}$ are exchangeable by the assumption. Therefore, $p_k^{Y_{n+1}}(X_{n+1})$ is a valid $p$-value conditional on the bag $(X_i, Y_i), i\in B_k$, where
\[
p_k^{y}(x) := \frac{|\{i\in [n]\setminus B_k:\,r_k(x, y) \le r_{k}(X_i, Y_i)\}| + 1}{|[n]\setminus B_k| + 1}.
\]
This is similar to the conclusion of Proposition~\ref{thm:Conformal-main-result} where we only need exchangeability of $(X_i, Y_i), i\in\mathcal{I}_2\cup\{n+1\}$. The result now follows from Proposition 18 of~\citet{vovkcombining}.
\end{proof}

\rev{\section{Imitating the optimal conditionally-valid prediction set}
 \label{appsec:optimal-prediction-set}}
   Prediction sets that are intervals may not be suitable (or well-defined) unless $\mathcal{Y}$ is a totally ordered set. For example, $\mathcal{Y}$ is not totally ordered for classification problems. Furthermore, even if $\mathcal{Y}$ is ordered, it is well known (see introduction of \citet{lei2013distribution}) that the optimal conditionally-valid prediction regions are level sets of conditional densities (with respect to an appropriate underlying measure), which need not be intervals. Formally, suppose $P_{Y|X}$ has a density $p(y|x)$ with respect to some measure $\mu$ on $\mathcal{Y}$. For a given miscoverage level $\alpha \in [0, 1]$ we wish to identify an optimal set $C \subseteq \mathcal{Y}$ that satisfies $1-\alpha$ coverage for $Y \mid X = x$, ie 
   \[
   \int_{y \in C} p(y|x)dy \ge 1 - \alpha.
   \]
   For simplicity, suppose that the conditional density $p(\cdot|x)$ is injective. Then, it is easy to see that the smallest set (with respect to the measure $\mu$) that satisfies coverage $1-\alpha$ must correspond to an upper level set $\{y\in\mathcal{Y}:\,p(y|x) \ge t\}$
 for some $t \geq 0$.
   In particular, the appropriate value of $t$ depends on $x$ and $\alpha$, and is given by  %
    \begin{equation}\label{eq:t-alpha}
    t_{\alpha}(x) := \sup\left\{t\ge 0 : \int_{p(y|x) \ge t} p(y|x)dy \ge 1 - \alpha\right\}.
    \end{equation}
   Clearly the set $\{{y \in \mathcal{Y}: p(y|x) \geq t_\alpha(x)}\}$ satisfies $1-\alpha$ coverage and is the smallest set to do so. 
   If an oracle provides us access to $p(y|x)$ for every $x,y$, we may thus compute the optimal prediction set at level $\alpha$ as 
    \begin{equation}\label{eq:optimal-pred-set}
   C_{\alpha}^{\texttt{oracle}}(x) ~:=~ \{y\in\mathcal{Y}:\,p(y|x) \ge t_{\alpha}(x)\}.
   \end{equation}
   Note that the prediction regions 
   $
   \left\{C_{\delta}^{\texttt{oracle}}(x)\right\}_{\delta\in[0,1]}
   $
   form a sequence of nested sets. 
   This motivates us to imitate/approximate $C^{\texttt{oracle}}_{\alpha}(x)$ through the nested framework as follows.  
  Let $\widehat{p}(\cdot|x)$ be any estimator of the conditional density, and let $\widehat{g}_x(t)$ be defined as
  \[
  \widehat{g}_x(t) ~:=~ \int_{y:\,\widehat{p}(y|x) \ge t}\widehat{p}(y|x)dy,
  \]
  which represents the estimated conditional probability of a level set with threshold $t$.
  Now, in our effort to mimic \eqref{eq:t-alpha}, for any $\delta \in [0,1]$, define $\widehat t_{\delta}(x)$ as the plugin threshold estimator: 
  \begin{equation}\label{eq:plugin-t}
  \widehat{t}_{\delta}(x) ~:=~ \sup\left\{t\ge0:\, \widehat{g}_x(t) \ge 1 - \delta\right\}.
  \end{equation}
   Last, define the nested sets $\{\mathcal{F}_{\delta}(x)\}_{\delta\in[0,1]}$ as
   \[
   \mathcal{F}_{\delta}(x) := \{y:\,\widehat{p}(y|x) \ge \widehat{t}_{\delta}(x)\}.
   \]
   It is clear that for any $x$, the function $\delta\mapsto \widehat{t}_{\delta}(x)$ is monotonically
   decreasing and hence for any $x$, the sets $\{\mathcal{F}_{\delta}(x)\}_{\delta\in[0, 1]}$ are nested. 
   It is also clear that if $\widehat{p}(\cdot|x) = p(\cdot|x)$, then $\mathcal{F}_{\alpha}(x) = C^{\texttt{oracle}}_{\alpha}(x)$. 
   Following this, we conjecture that
   if $\widehat{p}(\cdot|x)$ is consistent for $p(\cdot|x)$ in supremum norm, then $\mu(\mathcal{F}_{\alpha}(x)\Delta C_{\alpha}^{\texttt{oracle}}(x)) \to 0$ as $n \to \infty$. (The notation $A \Delta B$  represents symmetric difference.)

   An important distinguishing fact about nested sets~$\mathcal{F}_{\delta}(x)$, in comparison with the examples in Section~\ref{sec:split-conformal}, is that these are not intervals in general, and are also useful for prediction in the context of classification. 

   Applying nested split-conformal method from \eqref{eq:Nested-prediction-set} to the nested sets above yields the prediction set
   \[
   \widehat{C}_{\alpha}^{\texttt{oracle}}(x) ~:=~ \mathcal{F}_{\widehat{\delta}(\alpha)}(x),
   \] 
   where $\widehat{\delta}(\alpha)$ is obtained from the second split $\mathcal{I}_2$ of the data.
   Unlike the definition in Section~\ref{sec:split-conformal}, $\widehat{\delta}(\alpha)$ here is given by the equation
   \[
   1 - \widehat{\delta}(\alpha) = \lceil(1-\alpha)(1 + 1/|\mathcal{I}_2|)\rceil-\mbox{th quantile of }1 - \delta(X_i, Y_i), i\in\mathcal{I}_2,
   \]
   where $\delta(X_i, Y_i) := \sup\{\delta\in[0,1]:\,\widehat{p}(Y_i|X_i) \ge \widehat{t}_{\delta}(X_i)\}.$ This difference is because the nested sets here are decreasing in $\delta$ instead of increasing as in Section~\ref{sec:split-conformal}.
   Proposition~\ref{thm:Conformal-main-result} readily yields the validity guarantee
   \[
   \mathbb{P}(Y_{n+1}\in\widehat{C}_{\alpha}^{\texttt{oracle}}(X_{n+1})) \ge 1 - \alpha.
   \]
   It is important to realize that the prediction set $\widehat{C}_{\alpha}^{\texttt{oracle}}(X_{n+1})$ is only marginally valid, although it is motivated through the optimal conditionally valid prediction regions. 

   \cite{izbicki2019distribution} propose a prediction set similar to $\widehat{C}_{\alpha}^{\texttt{oracle}}$ but use an alternative estimate of $t_{\delta}(x)$ via a notion of ``profile distance'' between points, which they defined as
   \[
   d(x, x') := \int_0^{\infty} (\widehat{g}_x(t) - \widehat{g}_{x'}(t))^2dt.
   \]
   For every $x$, let $\mathcal{N}_m(x; d)$ represent the indices of the first $m$ nearest neighbors of $x$ with respect to the profile distance $d$. The prediction set described by \citet{izbicki2019distribution} is 
   \begin{equation}\label{eq:Izbicki-estimator}
   \widecheck{t}_{\delta}(x) := \mbox{$(1-\delta)$-th quantile of $\{\widehat{p}(y_i|x_i) : i\in \mathcal{N}_m(x; d)$\}}.
   \end{equation}
  Formulating their procedure in terms of the nested sets
   \[
   \mathcal{F}_{\delta}^{\texttt{ISS}}(x) ~:=~ \{y\in\mathcal{Y}:\,\widehat{p}(y|x) \ge \widecheck{t}_{\delta}(x)\},
   \]
   Proposition~\ref{thm:Conformal-main-result} readily yields marginal validity.

   We conjecture an improvement over the proposal of \citet{izbicki2019distribution}. Because the optimal prediction set~\eqref{eq:optimal-pred-set} depends directly on $t_{\delta}(x)$, it is more important to combine information from those $x_i$'s which are close to $x$ in terms of ${t}_{\delta}(x)$.  For example, we may define a ``revised'' profile distance as
   \[\textstyle
   \widetilde{d}(x, x') := \int_0^1 (\widehat t_{\delta}(x) - \widehat t_{\delta}(x'))^2w(\delta)d\delta,
   \]  
   where the function $w(\cdot)$ provides more weight on values close to zero.
	(We use such $w$ because often one is interested in coverage levels close to 1 or equivalently small values of the miscoverage level $\alpha$.) Using this new distance, we can use various alternative estimators of the threshold, for example using $\widetilde d$ in \eqref{eq:Izbicki-estimator}, or kernel-smoothed variants of \eqref{eq:plugin-t} and \eqref{eq:Izbicki-estimator}.

 In this paper, we focus on the problem of valid prediction region in the context of regression in which case one often wants to report an interval. For this reason, we leave the discussion of optimal prediction regions discussed herein at this stage, although a more detailed enquiry in this direction would be fruitful for complicated response spaces.

\section{Experimental information and additional results}
\label{appsec:additional-info-exps}
\label{subsec:experiment-info}
All our experiments were conducted in MATLAB using the TreeBagger class %
for training regression forests and quantile regression forests. Default parameters were used for all datasets apart from the synthetic dataset of Section~\ref{subsec:synthetic}. Details for the datasets used in our experiments are provided in Table~\ref{table:metadata}. We make two additional comments: 
\begin{enumerate}
    \item The Protein dataset on UCI has 9 features. However, our data loading code erroneously skipped the ninth feature --- thus $d = 8$ in Table~\ref{table:metadata}.\footnote{We thank Yachong Yang for pointing this out.} All results in Section~\ref{sec:numerical-experiments} are based on the dataset that does not use the ninth feature. In Appendix~\ref{appsec:more-exps}, we present Table \ref{table:overall-comparison-width} and \ref{table:overall-comparison-coverage} style results using all 9 features, under the dataset name Protein2. We observe that for all conformal methods, the mean-width values for Protein2 differ only slightly from those of Protein (as expected, including the ninth feature leads to a small decrease in the mean-width for all methods). Importantly, the relative performance of the conformal methods remains exactly the same, and hence the conclusions regarding competitiveness of QOOB remain unaffected. While we did not reproduce the detailed plots of Section~\ref{sec:numerical-experiments} for Protein2, we suspect that the relative behavior of the conformal methods is same for Protein and Protein2 across different hyperparameters.
    
    \item The Kernel dataset contains 4 output variables corresponding to 4 measurements of the same entity. The output variable we used is the average of these values.
\end{enumerate}

\begin{table}[!t]
	\caption{Meta-data for the datasets used in our experiments. $N$ refers to the total number of data-points from which we create 100 versions by independently drawing 1000 data points randomly (as described in the beginning of Section~\ref{sec:numerical-experiments}). $d$ refers to the feature dimension.}
	\centering
	\resizebox{\textwidth}{!}{\begin{tabular}{|c|l|c|c|}
			\hline
			\textbf{Dataset} & \textbf{$N$} & \textbf{$d$} &\textbf{URL (\url{http://archive.ics.uci.edu/ml/datasets/*})}  \\\hline  
			Blog & 52397 & 280 & \url{BlogFeedback}  \\ \hline
			Protein & 45730 & 8 & \url{Physicochemical+Properties+of+Protein+Tertiary+Structure} \\ \hline
			Concrete & 1030 & 8 & \url{Concrete+Compressive+Strength}  \\ \hline
			News & 39644 & 59 & \url{Online+News+Popularity}  \\ \hline
			Kernel%
			& 241600 & 14 &\url{SGEMM+GPU+kernel+performance}  \\ \hline
			Superconductivity & 21263 & 81& \url{Superconductivty+Data} \\ \hline
	\end{tabular}}
	\label{table:metadata}
\end{table}

In the following subsection, we present experimental results on 5 additional UCI datasets.

\subsection{Experiments with additional datasets}
\label{appsec:more-exps}

\begin{table}[!t]
	\caption{Meta-data for the additional datasets presented in Appendix~\ref{appsec:more-exps}. $N$ refers to the total number of data-points from which we create 100 versions by independently drawing 1000 data points randomly (as described in the beginning of Section~\ref{sec:numerical-experiments}). $d$ refers to the feature dimension.}
	\centering
	\resizebox{\textwidth}{!}{\begin{tabular}{|c|l|c|c|}
			\hline
			\textbf{Dataset} & \textbf{$N$} & \textbf{$d$} &\textbf{URL (\url{http://archive.ics.uci.edu/ml/datasets/*})}  \\\hline  
			Airfoil & 1503 & 5 & \url{Airfoil+Self-Noise}  \\ \hline
			Electric & 10000 & 12 & \url{Electrical+Grid+Stability+Simulated+Data+} \\ \hline
			Cycle & 9568 & 4 & \url{Combined+Cycle+Power+Plant}  \\ \hline
			WineRed & 4898 & 11 & \url{wine-quality/winequality-red.csv}  \\ \hline
			WineWhite%
			& 1599 & 11 &\url{wine-quality/winequality-white.csv}  \\ \hline
			Protein2 & 45730 & 9& \url{Physicochemical+Properties+of+Protein+Tertiary+Structure} \\ \hline
	\end{tabular}}
	\label{table:metadata-additional}
\end{table}
In Section~\ref{sec:numerical-experiments}, we performed an in-depth comparison of QOOB with 5 other conformal methods, on 6 datasets, with varying nominal quantile level, number of trees, and sample size. In this section, we add more breadth by presenting results on 5 additional UCI datasets for default hyperparameters: number of trees = 100, nominal quantile level = $2\alpha$, number of samples = 1000. This is the counterpart of Tables \ref{table:overall-comparison-width} and \ref{table:overall-comparison-coverage} for the additional datasets. We also present results with the corrected version of the Protein dataset, called Protein2. The metadata for the datasets is contained in Table~\ref{table:metadata-additional}. The results are in Tables \ref{table:overall-comparison-width-additional} and \ref{table:overall-comparison-coverage-additional}.
\begin{table}[t]
	\caption{Mean-width~\eqref{eq:mean-width} of conformal methods with regression forests ($\alpha = 0.1$). Average values across 100 simulations are reported with the standard deviation in brackets. }
	\centering
	\resizebox{\textwidth}{!}{
		\begin{tabular}{c|c|c|c|c|c|c}
			\hline
			\textbf{Method} & {Airfoil} & {Electric} & {Cycle}& {WineRed} & {WineWhite} & {Protein2} \\ \hline \hline 
			SC-100 & 11.90 & 0.067 & 13.80 & 2.04 & 2.35 & 16.76 \\[-0.15in] 
			& (0.08) & (0.00) & (0.07) & (0.01) & (0.01) & (0.08) \\\hline
			Split-CQR-100 (2$\alpha$) & 11.40 & 0.069 & 13.83 & 2.42 & 2.99 & 14.07 \\[-0.15in] 
			& (0.09) & (0.00) & (0.07) & (0.08) & (0.04) & (0.08)\\\hline
			8-fold-CC-100 & 10.42 & 0.062 & 13.25 & 2.01 & 2.32 & 16.41 \\[-0.15in] 
			& (0.02) & (0.00) & (0.04) & (0.01) & (0.01) & (0.04) \\\hline
			OOB-CC-100 & 10.11 & \textbf{0.061} & 13.18 & \textbf{2.00} & \textbf{2.31} & 16.38 \\[-0.15in] 
			& (0.02) & (0.00) & (0.04) & (0.01) & (0.01) & (0.04) \\\hline
			OOB-NCC-100 & 10.25 & 0.065 & 14.01 & \textbf{2.00} & 2.35 & 14.97 \\ [-0.15in]
			& (0.03) & (0.00) & (0.06) & (0.01) & (0.01) & (0.05)\\\hline
			QOOB-100 ($2\alpha$) & \textbf{9.80} & 0.064 & \textbf{13.12} & 2.53 & 2.94 & \textbf{13.73} \\[-0.15in] 
			& (0.02) & (0.00) & (0.04) & (0.04) & (0.00) & (0.05)\\\hline
	\end{tabular}}
	\label{table:overall-comparison-width-additional}
\end{table}
\begin{table}[!h]
	\caption{Mean-coverage~\eqref{eq:mean-coverage} of conformal methods with regression forests ($\alpha = 0.1$).  Average values across 100 simulations are reported. The standard deviation of these average values are zero up to two significant digits. }
	\centering
	\resizebox{\textwidth}{!}{
		\begin{tabular}{c|c|c|c|c|c|c}
			\hline
			\textbf{Method} & {Airfoil} & {Electric} & {Cycle}& {WineRed} & {WineWhite} & {Protein2} \\ \hline \hline 
			SC-100 & 0.90 & 0.90 & 0.90 & 0.90 & 0.90 & 0.90 \\ \hline
			Split-CQR-100 ($2\alpha$) & 0.90 & 0.90 & 0.90 & 0.96 & 0.98 & 0.90 \\ \hline
			8-fold-CC-100 & 0.91 & 0.91 & 0.91 & 0.90 & 0.90 & 0.91 \\ \hline
			OOB-CC-100 & 0.91 & 0.91 & 0.90 & 0.90 & 0.90 & 0.90 \\ \hline
			OOB-NCC-100 & 0.91 & 0.91 & 0.91 & 0.90 & 0.91 & 0.91 \\ \hline
			QOOB-100 ($2\alpha$) & 0.92 & 0.92 & 0.91 & 0.98 & 0.99 & 0.91 \\ \hline
	\end{tabular}}
	\label{table:overall-comparison-coverage-additional}
\end{table}
We make the following observations:
\begin{enumerate}
    \item   On 3 out of 6 datasets, QOOB is the best performing method, and on the remaining 3, OOB-CC is the best performing method. Thus while Split-CQR was the best performing method on some of the datasets in the main paper, this is not true for any of the additional datasets. 
\item Specifically, QOOB performs better than Split-CQR on all datasets but one (WineRed). 
\item We observe that for both Wine datasets, the quantile based methods (QOOB and Split-CQR) overcover: they have significantly high mean-widths compared to the non-quantile based methods, and the coverage is more than $0.96$.  This is in stark contrast to all other experimental results in this paper. On investigation, we observed that the Wine datasets have very light tails --- most of the mass is concentrated on three output values: 5,6,7. Thus conditional quantile estimation can fail on some points for quantile levels away from 0.5; this seems to have occurred in our experiments. 
\end{enumerate}

Based on Tables \ref{table:overall-comparison-width} and \ref{table:overall-comparison-width-additional},  we recommend QOOB as the most reliable general purpose conformal method.

\end{document}